\documentclass[11pt]{article}
\usepackage[letterpaper,hmargin=1in,vmargin=1in]{geometry}

\usepackage[ruled]{algorithm2e}
\usepackage{caption}
\usepackage{hyperref}
\usepackage{adjustbox,multirow}
\usepackage{verbatim,color}
\usepackage{amsmath,amsthm,amsfonts,latexsym,verbatim,threeparttable,amsfonts,rotating,pifont}

\usepackage{enumitem}
\usepackage{latexsym}
\usepackage{mathtools}
\usepackage{multirow}
\usepackage{multicol}
\usepackage{mwe}
\usepackage{makecell}
\usepackage{pgfplots}
\usepackage[section]{placeins}
\usepackage{pifont}
\usepackage{ragged2e}
\usepackage{rotating}
\usepackage{subcaption}
\usepackage{tabu}
\usepackage{tabularx}
\usepackage{threeparttable}
\usepackage{tikz}
\usepackage[normalem]{ulem}
\usepackage{url}
\usepackage{verbatim}
\usepackage{wrapfig}
\usepackage{xcolor}
\usepackage{xspace}
\usepackage[framemethod=tikz]{mdframed}
\usepackage{algpseudocode}
\input{symbols_env}
\newcommand{\F}{\ensuremath{\mathbb{F}}}
\newcommand{\defined}{\ensuremath{\stackrel{def}{=}}}

\newcommand{\Order}{\ensuremath{\mathcal{O}}}
\newcommand{\starF}{\ensuremath{{F}^{\star}}}

\newcommand{\ckt}{\ensuremath{\mathsf{cir}}}

\newcommand{\z}{\ensuremath{\mathfrak{z}}}

\newcommand{\committed}{\ensuremath{\mathsf{committed}}}
\newcommand{\false}{\ensuremath{\mathsf{false}}}
\newcommand{\true}{\ensuremath{\mathsf{true}}}
\newcommand{\Coin}{\ensuremath{\mathsf{Coin}}}

\newcommand{\D}{\ensuremath{\mathsf{D}}}
\newcommand{\Adv}{\ensuremath{\mathsf{Adv}}}

\newcommand{\Partyset}{\ensuremath{\mathcal{P}}}
\newcommand{\PartySet}{\ensuremath{\mathcal{P}}}
\newcommand{\Honest}{\ensuremath{\mathcal{H}}}

\newcommand{\Bad}{\ensuremath{\mathcal{C}}}

\newcommand{\WCORE}{\ensuremath{\mathcal{W}}}

\newcommand{\R}{\ensuremath{\mathcal{R}}}

\newcommand{\CSet}{\ensuremath{\mathcal{C}}}

\newcommand{\ESet}{\ensuremath{\mathcal{E}}}
\newcommand{\FSet}{\ensuremath{\mathcal{F}}}

\newcommand{\Support}{\ensuremath{\mathcal{SS}}}

\newcommand{\CoreSet}{\ensuremath{\mathcal{CS}}}

\newcommand{\Star}[1]{\ensuremath{(n, #1){-}\mbox{star}}}

\newcommand{\OK}{\ensuremath{\texttt{OK}}}
\newcommand{\NOK}{\ensuremath{\texttt{NOK}}}

\newcommand{\init}{\ensuremath{\texttt{init}}}
\newcommand{\echo}{\ensuremath{\texttt{echo}}}
\newcommand{\ready}{\ensuremath{\texttt{ready}}}

\newcommand{\Sender}{\ensuremath{\mathsf{S}}}

\newcommand{\flag}{\ensuremath{\mathsf{flag}}}

\newcommand{\Vote}{\ensuremath{\Pi_{\mathsf{Vote}}}}

\newcommand{\VSS}{\ensuremath{\Pi_{\mathsf{VSS}}}}

\newcommand{\HSh}{\ensuremath{\Pi_{\mathsf{VSS}}}}

\newcommand{\RecPriv}{\ensuremath{\Pi_{\mathsf{RecPriv}}}}
\newcommand{\BatchBeaver}{\ensuremath{\Pi_{\mathsf{Beaver}}}}
\newcommand{\TripTrans}{\ensuremath{\Pi_{\mathsf{TripTrans}}}}
\newcommand{\TripExt}{\ensuremath{\Pi_{\mathsf{TripExt}}}}
\newcommand{\TripSh}{\ensuremath{\Pi_{\mathsf{TripSh}}}}

\newcommand{\Offline}{\ensuremath{\Pi_{\mathsf{PreProcessing}}}}

\newcommand{\SBA}{\ensuremath{\Pi_{\mathsf{SBA}}}}
\newcommand{\BCAST}{\ensuremath{\Pi_\mathsf{BC}}}
\newcommand{\PiACast}{\ensuremath{\Pi_\mathsf{ACast}}}
\newcommand{\ABA}{\ensuremath{\Pi_\mathsf{ABA}}}
\newcommand{\HBA}{\ensuremath{\Pi_\mathsf{BA}}}
\newcommand{\PiABA}{\ensuremath{\Pi_{\mathsf{ABA}}}}
\newcommand{\ACS}{\ensuremath{\Pi_{\mathsf{ACS}}}}

\newcommand{\PiMPC}{\ensuremath{\Pi_{\mathsf{CirEval}}}}

\newcommand{\OEC}{\ensuremath{\mathsf{OEC}}}
\newcommand{\RSDec}{\ensuremath{\mathsf{RSDec}}}

\newcommand{\WPS}{\ensuremath{\Pi_\mathsf{WPS}}}

\newcommand{\CoinFlip}{\ensuremath{\Pi_{\mathsf{CoinFlip}}}}
\newcommand{\StarAlgo}{\ensuremath{\mathsf{AlgStar}}}
\newcommand{\PiAsynWPS}{\ensuremath{\Pi_{\mathsf{AWPS}}}}
\newcommand{\PiSynWPS}{\ensuremath{\Pi_{\mathsf{SWPS}}}}
\newcommand{\BGP}{\ensuremath{\Pi_{\mathsf{BGP}}}}

\newcommand{\TimeABA}{\ensuremath{T_\mathsf{ABA}}}

\newcommand{\TimeSBA}{\ensuremath{T_\mathsf{SBA}}}

\newcommand{\TimeHBA}{\ensuremath{T_\mathsf{BA}}}

\newcommand{\TimeBCAST}{\ensuremath{T_\mathsf{BC}}}

\newcommand{\TimeHSh}{\ensuremath{T_\mathsf{VSS}}}
\newcommand{\TimeWPS}{\ensuremath{T_\mathsf{WPS}}}

\newcommand{\TimeOffline}{\ensuremath{T_\mathsf{{TripGen}}}}

\newcommand{\TimeTripSh}{\ensuremath{T_\mathsf{{TripSh}}}}

\newcommand{\TimeACS}{\ensuremath{T_\mathsf{ACS}}}

\newcommand{\TimeVSS}{\ensuremath{T_\mathsf{VSS}}}
\newcommand{\TimeBGP}{\ensuremath{T_\mathsf{BGP}}}

\newenvironment{myitemize}
{\begin{list}{$\bullet$}{ 
\itemindent=-0.1in
\itemsep=0.0in
\parsep=0.0in
\topsep=0.0in
\partopsep=0.0in}}{\end{list}}
\newcounter{itemcount}

\newenvironment{mydescription}
{\setcounter{itemcount}{0}\begin{list}
{\arabic{itemcount}.}{\usecounter{itemcount} \itemindent=-0.5cm
\itemsep=0.0in
\parsep=0.0in
\topsep=5pt
\partopsep=0.0in}}{\end{list}}

\makeatletter
\DeclareRobustCommand*\cal{\@fontswitch\relax\mathcal}
\makeatother

\newtheorem{theorem}{Theorem}[section]

\newtheorem{lemma}[theorem]{Lemma}

\theoremstyle{definition}
\newtheorem{definition}[theorem]{Definition}

\usepackage{color}
\usepackage{framed}

\begin{document}

\title{Perfectly-Secure Synchronous MPC with Asynchronous Fallback Guarantees\footnote{A preliminary version of this article was
 published as an extended abstract in PODC 2022 \cite{ACC22}. This is the full and elaborate version, with complete proofs.}}
\author{Ananya Appan\footnote{The work was done when the author was a student at the International Institute of Information Technology, Bangalore India.} \and Anirudh Chandramouli\footnote{The work was done when the author was a student at the International Institute of Information Technology, Bangalore India. The author would like to thank Google Research for travel support
 to attend and present the preliminary version of the paper at PODC 2022.}
 \and  Ashish Choudhury\footnote{International Institute of Information Technology, Bangalore India.
 Email: {\tt{ashish.choudhury@iiitb.ac.in}}. This research is an outcome of the R \& D work undertaken in the project under the Visvesvaraya PhD  Scheme of  
 Ministry of Electronics \& Information Technology, Government of India, being implemented by Digital India Corporation
  (formerly Media Lab Asia). The author is also
   thankful to the Electronics, IT \& BT Government of Karnataka for supporting this work under the CIET project.} 
 }


\date{}
\maketitle

\begin{abstract}
 Secure {\it multi-party computation} (MPC) is a fundamental problem in secure distributed computing.
 An MPC protocol allows a set of $n$ mutually distrusting parties to carry out any joint computation of their private inputs,
  without disclosing any additional information about their inputs. 
   MPC with {\it information-theoretic} security (also called {\it unconditional security}) provides the strongest security guarantees
   and remains secure even against {\it computationally unbounded} adversaries. {\it Perfectly-secure} MPC protocols is a class
   of information-theoretically secure
    MPC protocols, which provides all the security guarantees in an {\it error-free} fashion.
   The focus of this work is perfectly-secure MPC. Known protocols are designed {\it assuming} either a {\it synchronous}
   or {\it asynchronous} communication network. It is well known that perfectly-secure {\it synchronous} MPC protocol is possible
   as long as adversary can corrupt any $t_s < n/3$ parties. 
   On the other hand, perfectly-secure {\it asynchronous} MPC protocol can tolerate up to $t_a < n/4$ corrupt parties. 
   A natural question is does there exist a {\it single} MPC protocol for the setting where the parties are {\it not aware}
   of the exact network type and which can tolerate up to $t_s < n/3$ corruptions 
    in a synchronous network and up to $t_a < n/4$ corruptions
    in an {\it asynchronous} network. We design such a {\it best-of-both-worlds}
    perfectly-secure MPC protocol, provided $3t_s + t_a < n$ holds.
    
    For designing our protocol, we design two important building blocks, which are of independent interest.
    The first building block is a best-of-both-worlds {\it Byzantine agreement} (BA) protocol tolerating
    $t < n/3$ corruptions and which remains secure, {\it both} in a synchronous as well as asynchronous
    network. The second building block is a polynomial-based best-of-both-worlds 
    {\it verifiable secret-sharing} (VSS) protocol, which can tolerate up to 
      $t_s$ and $t_a$ corruptions in a {\it synchronous} and in an {\it asynchronous} network respectively. 
\end{abstract}

\section{Introduction}
\label{sec:intro}
  Consider a set of $n$ mutually distrusting parties  $\Partyset = \{P_1, \ldots, P_n \}$, where each $P_i$ has some private input. 
   The distrust among the parties is modeled as a centralized adversary, who can control 
   any $t$ out of the $n$ parties in a {\it Byzantine} (malicious)
    fashion and force them to behave arbitrarily during the execution of any protocol. 
    An MPC protocol \cite{Yao82,GMW87,BGW88,CCD88,RB89} allows the parties to securely compute any known function of their 
   private inputs, such that the {\it honest} parties (who are not under adversary's control) obtain the correct output,
   irrespective of the behaviour of the adversary. Moreover, adversary does not learn any additional
   information about the inputs of the honest parties, beyond what can be revealed by the
   function output and the inputs of the corrupt parties.
   If the adversary is {\it computationally bounded} then the notion of security achieved is called 
   {\it conditional security} (also known as {\it cryptographic security}) \cite{Yao82,GMW87,GRR98}.
   On the other hand, {\it unconditionally secure} protocols (also known as {\it information-theoretically} secure protocols) 
   provide security against {\it computationally unbounded} adversaries \cite{BGW88,CCD88}.
   Unconditionally secure protocols provide {\it ever-lasting} security, as their security is {\it not} based on
   any computational-hardness assumptions. Moreover, compared to conditionally secure protocols, the protocols are simpler
   and faster by several order of magnitude, as they are based on very simple operations, such as polynomial interpolation and
   polynomial evaluation over finite fields.
   Unconditionally secure protocols can be further categorized as {\it perfectly-secure} MPC 
   protocols \cite{BGW88,GRR98,DN07,BH08,GLS19,AAY21}, where all security properties are achieved in an {\it error-free} fashion.
   On the other hand, {\it statistically-secure} MPC protocols \cite{RB89,CDDHR99,BH06,BFO12,GLS19}
   allow for a negligible error in the achieved security properties.

    Traditionally, MPC protocols are designed {\it assuming} either a {\it synchronous} or {\it asynchronous} communication model.
   In {\it synchronous} MPC (SMPC) protocols, parties are assumed to be synchronized with respect to a global clock and there is a
   {\it publicly-known} upper bound on message delays.
   Any SMPC protocol operates as a sequence of communication {\it rounds}, where in each round,
   every party performs some computation, sends messages to other parties and receives messages sent by the other parties, in that order.
   Consequently, if during a round a receiving party {\it does not} receive an expected message from a designated sender party by
    the end of that round,
    then the receiving party has
   the assurance that the sender party is definitely {\it corrupt}. 
      Though synchronous communication model is highly appealing in terms of its simplicity, 
   in practice, it might be very difficult to guarantee such strict time-outs over the channels in
    real-world networks like the Internet. 
    Such networks are better modeled through the {\it asynchronous} communication model \cite{CanettiThesis}. 
    
    An {\it asynchronous} MPC (AMPC) protocol operates over an {\it asynchronous} network, 
   where the messages can be arbitrarily, yet finitely delayed. 
   The only guarantee in the model is that every sent message is {\it eventually} delivered. 
   Moreover, the messages {\it need not} be delivered in the same order in which they were sent.
   Furthermore, to model the worst case scenario, the sequence of message delivery is assumed to be under the control
   of the adversary. Unlike SMPC protocols, the protocol execution in an AMPC protocol occurs as a sequence of
   {\it events}, which depend upon the order in which the parties receive messages.
   Comparatively, AMPC protocols are more challenging to design than SMPC protocols. This is because inherently, in any AMPC
   protocol, a receiving party {\it cannot}
    distinguish between a {\it slow} sender party (whose messages are arbitrarily delayed in the network) and a {\it corrupt} 
   sender party (who does not send any messages). 
   Consequently, in any AMPC protocol with up to $t_a$ corruptions, at any stage of the protocol, no party can afford to 
   receive messages from {\it all} the parties. This is because the corrupt parties may never send their messages and hence the wait
   could turn out be an endless wait. Hence, as soon as a party receives messages from any subset of $n - t_a$ parties, it has to proceed
   to the next stage of the protocol. However, in this process, messages from up to $t_a$ potentially slow, but honest parties, may get
   ignored. In fact, in {\it any} AMPC protocol, it is  {\it impossible} to ensure that the inputs of all  {\it honest} parties are considered for
    the computation and inputs of up to
   $t_a$ (potentially honest) parties may have to be ignored, since waiting for all $n$ inputs  may turn out to be an endless wait. 
   The advantage of AMPC protocols over SMPC protocols
    is that the time taken to produce the output depends upon the {\it actual speed} of the underlying network.
    In more detail, for an SMPC protocol, the participants have to {\it pessimistically} set the global
    delay $\Delta$ on the message delivery to a large value to ensure that the messages sent by every party
     at the beginning of a round reach to their destination within time $\Delta$.
      But if the actual delay $\delta$ in the network is such that $\delta < < \Delta$, then the
     protocol {\it fails} to take advantage of the faster network and its running time will be still proportional to $\Delta$.
      
      The focus of this work is {\it perfectly-secure} MPC. It is well known that 
       perfectly-secure SMPC is possible if and only if adversary can corrupt up to
       $t_s < n/3$ parties \cite{BGW88}. On the other hand, 
       perfectly-secure AMPC is possible if and only if adversary can corrupt up to
       $t_a < n/3$ parties \cite{BCG93}.
\paragraph{\bf Our Motivation and Our Results:}
 As discussed above, known SMPC and AMPC protocols are designed under the {\it assumption} that the parties are
  {\it aware} of the exact network type. 
       We envision a scenario where the parties {\it are not} aware of the {\it exact} network type and
       aim to design a single MPC protocol, which remains secure, {\it both} in a synchronous, {\it as well as} in an
       asynchronous network. We call such a protocol as
          a {\it best-of-both-worlds} protocol, since it offers the 
         best security properties, both in the synchronous and the asynchronous communication model.
         While there exist best-of-both-worlds {\it conditionally-secure} MPC
          protocols \cite{BZL20,DHL21}, to the best of our knowledge, {\it no} prior work has ever addressed the problem of
          getting a best-of-both-worlds {\it perfectly-secure} MPC protocol. Motivated by this,
         we ask the following question:
        \begin{center}
        {\it Is there a best-of-both-worlds perfectly-secure MPC protocol, that remains
         secure under $t_s$ corruptions in a synchronous network, and 
        under $t_a$ corruptions in an asynchronous network, where $t_a < t_s$}?        
        \end{center}
        We show the existence of a perfectly-secure MPC protocol 
         with the above guarantees, provided $3t_s + t_a < n$ holds.\footnote{This automatically implies that $t_s < n/3$ and $t_a < n/4$
         holds, which are necessary for designing perfectly-secure MPC protocol in a synchronous and an asynchronous network respectively.} 
         Note that we are interested in the case where $t_a < t_s$, as otherwise the
         question is trivial to solve.         
          More specifically, if $t_s = t_a$, then the necessary condition of AMPC implies
           that $t_s < n/4$ holds. Hence, one can use {\it any existing} perfectly-secure
	AMPC protocol, which will be secure under $t_s$ corruptions {\it even}
         in a synchronous network. Moreover, by ensuring appropriate time-outs, it can be guaranteed
          that in the protocol, the inputs of {\it all} honest parties are considered for the computation,
          if the network is {\it synchronous}.
           Our goal is to achieve a resilience {\it strictly greater} than $n/4$
            and {\it close} to $n/3$, if the underlying network is {\it synchronous}. For
	 example, if $n = 8$, then {\it existing} perfectly-secure SMPC protocols
	can tolerate up to $2$ corrupt parties, while {\it existing} perfectly-secure
	AMPC protocols can tolerate up to $1$ fault. On the other hand, using our best-of-both-worlds protocol, 
       one can tolerate up to $2$ faults in a {\it synchronous} network and up to $1$ fault in an
       {\it asynchronous} network, {\it even} if the parties are {\it not aware} of the exact network type.         
\subsection{Technical Overview}
 We assume that the function to be securely computed
 is represented by some arithmetic circuit $\ckt$ over a finite field $\F$, consisting of linear and non-linear (multiplication) gates.
   Following \cite{BGW88}, the goal is then 
   to securely ``evaluate" $\ckt$ in a {\it secret-shared} fashion, such that all the values during the circuit-evaluation 
  are $t$-shared, as per the Shamir's secret-sharing scheme \cite{Sha79}, where $t$ is the maximum number of corrupt 
  parties.\footnote{A value $s \in \F$ is said to be $t$-shared, if there is some
   $t$-degree polynomial $f_s(\cdot)$ over $\F$ with $f_s(0) = s$  and every (honest) $P_i$ has a distinct point on $f_s(\cdot)$, which is
   called $P_i$'s share of $s$.}
     Intuitively, this guarantees that an adversary controlling up to $t$ parties {\it does not} learn any additional information
      during the circuit-evaluation, as the shares of the corrupt parties does not reveal anything additional about the actual shared values.
      The {\it degree-of-sharing} $t$ is set to $t < n/3$ and $t<n/4$ in SMPC and AMPC protocols respectively.
     Since, in our best-of-both-worlds protocol, the parties will {\it not} be aware of the {\it exact} network type, 
    we need to ensure that {\it all} the values during circuit-evaluation
   are {\it always} secret-shared with the degree-of-sharing being $t = t_s$, {\it even} if the network is {\it asynchronous}.
  
  For shared circuit-evaluation, 
  we follow the Beaver's paradigm \cite{Bea91}, 
  where multiplication gates are evaluated using random $t_s$-shared {\it multiplication-triples} of the form
   $(a, b, c)$, where $c = a \cdot b$ (due to the linearity of Shamir's secret-sharing, linear gates can be
   evaluated {\it non-interactively}).
    The shared multiplication-triples are generated 
    in a {\it circuit-independent} preprocessing phase,
    using
    the framework of \cite{CP17}, which shows how to use any polynomial-based {\it verifiable secret-sharing} (VSS)  \cite{CGMA85} and a
    {\it Byzantine agreement} (BA) protocol \cite{PSL80} to generate shared random multiplication-triples.
    The framework works both in a synchronous as well as in an asynchronous network, where 
     the parties are {\it aware} of the exact network type.  
     However, there are several challenges to adapt the framework if the parties are
     {\it unaware} of the exact network type, which we discuss next.
       \paragraph{\bf First Challenge --- A Best-of-Both-Worlds Byzantine Agreement (BA) Protocol:}
     Informally, a BA protocol \cite{PSL80} allows the parties with private inputs 
     to reach agreement on a common output ({\it consistency}), where the output is the input of the honest parties, if all honest parties
     participate in the protocol with the {\it same} input ({\it validity}).
               {\it Perfectly-secure} BA protocols can be designed tolerating  $t < n/3$ corruptions, 
               both in a {\it synchronous} network \cite{PSL80}, as well as in an {\it asynchronous} network \cite{CR93,ADH08,BCP20}.            
            However, the {\it termination} (also called {\it liveness}) guarantees are {\it different} for {\it synchronous} BA (SBA)
            and {\it asynchronous} BA (ABA). (Deterministic) SBA protocols ensure that all honest parties obtain their output after some 
            {\it fixed} time
             ({\it guaranteed liveness}). On the other hand, to circumvent the FLP impossibility result \cite{FLP85},
             ABA protocols are randomized and provide what is called as {\it almost-surely liveness} \cite{ADH08,BCP20}.
             Namely, the parties obtain an output, {\it asymptotically} with probability $1$, if they continue running the protocol. 
              SBA protocols become {\it insecure} when executed in an {\it asynchronous} network, if even a single expected message
              from an {\it honest} party is delayed. On the other hand, ABA protocols 
              when executed in a {\it synchronous} network, can provide {\it only} almost-surely liveness, instead of guaranteed liveness.
             
             The {\it first} challenge to adapt the framework of \cite{CP17} in the best-of-both-worlds setting 
                          is to get a {\it perfectly-secure} BA protocol, which provides security {\it both} in a synchronous as well
                          as in an asynchronous network.
                          Namely, apart from providing the consistency and validity properties in both
                          types of network, the protocol should provide guaranteed liveness in a {\it synchronous} network
                          and almost-surely liveness in an {\it asynchronous} network. 
         We are {\it not} aware of any BA protocol with the above properties.
            Hence, we 
            present a perfectly-secure BA protocol tolerating $t < n/3$ corruptions, with the above properties.
            Since our BA protocol is slightly technical, we defer the details to Section \ref{sec:HBA}.

    \paragraph{\bf Second Challenge --- A best-of-both-worlds VSS Protocol:}
     Informally, in a polynomial based VSS protocol,
     there exists a designated {\it dealer} $\D$ with a $t$-degree polynomial, where $t$ is the maximum number
     of {\it corrupt} parties, possibly including $\D$. The protocol allows 
    $\D$ to distribute points on this polynomial to the parties in a ``verifiable" fashion, such that 
        the view of the adversary remains independent of $\D$'s polynomial for an
         {\it honest} $\D$.\footnote{Hence the protocol allows
    $\D$ to generate a $t$-sharing of the constant term of the polynomial, which is also called as $\D$'s {\it secret}.}
    In a {\it synchronous} VSS (SVSS) protocol,
    every party has the correct point after some known time-out, say $T$ ({\it correctness} property).
    The {\it verifiability} guarantees that even a {\it corrupt} $\D$
    is bound to distribute points on some $t$-degree polynomial within time $T$  ({\it strong-commitment} property).
    Perfectly-secure SVSS is possible if and only if $t < n/3$ \cite{DDWY93}. 
    For an {\it asynchronous} VSS (AVSS) protocol, the {\it correctness} property guarantees that for an {\it honest} 
    $\D$, the honest parties {\it eventually} receive points on $\D$'s polynomial.
    However, a {\it corrupt}
    $\D$ may not invoke the protocol in the first place and the parties {\it cannot} distinguish this scenario from the case when $\D$'s
     messages are arbitrarily delayed.
    This is unlike the {\it strong-commitment} of SVSS where, if 
     the parties do not obtain an output within time $T$, then the parties {\it publicly} conclude that $\D$ is corrupt.
     Hence, the {\it strong-commitment} of AVSS guarantees that if $\D$ is {\it corrupt} {\it and}
     if some honest party obtains a point on $\D$'s polynomial, then all honest parties eventually obtain
     their respective points on this polynomial.  Perfectly-secure AVSS 
    is possible if and only if $t < n/4$ \cite{BCG93,ADS20}.

    Existing SVSS protocols \cite{GIKR01,FGGRS06,KKK09,CCP21}
    become completely insecure in an {\it asynchronous} network, even if a single expected message
    from an {\it honest} party is delayed. On the other hand, existing AVSS protocols \cite{BCG93,BH07,PCR15,CCP21}
    only work when $\D$'s polynomial has degree $t < n/4$ and become insecure
    if there are {\it more} than $n/4$ corruptions (which can happen in our context, if the network is {\it synchronous}).
    
       The {\it second} challenge to adapt the framework of \cite{CP17} in our setting
    is to get a perfectly-secure VSS protocol, which provides security against $t_s$ and $t_a$ corruptions in a synchronous and in
     an asynchronous network respectively,
     where $\D$'s polynomial is {\it always} a
      $t_s$-degree polynomial, {\it irrespective} of the network type. 
    We are not aware of any VSS protocol with these guarantees.
     We present a best-of-both-worlds 
     perfectly-secure VSS protocol satisfying the above properties, provided $3t_s + t_a < n$ holds. 
    Our VSS protocol satisfies the {\it correctness} requirement of SVSS and AVSS in a {\it synchronous} and an {\it asynchronous}
    network respectively. However, it {\it only} satisfies the {\it strong-commitment} requirement of AVSS, {\it even} if the network is 
    {\it synchronous}. This is because 
   a potentially {\it corrupt} $\D$ {\it may not} invoke the protocol and the parties will {\it not} be aware of the exact network type.
   We stress that this does not hinder us from deploying our VSS
   protocol in the framework of \cite{CP17}.
   Since our VSS protocol is slightly technical, we defer the details to Section \ref{sec:HVSS}.

      \subsection{Related Work}
       best-of-both-worlds protocols
         have been studied very recently. 
         The work of  \cite{BKL19} shows that the condition $2t_s + t_a < n$ is necessary and sufficient for 
         best-of-both-worlds {\it conditionally-secure}
        BA, tolerating {\it computationally bounded} adversaries.
        Using the same condition, the works of \cite{BZL20,DHL21} 
        present  {\it conditionally-secure} MPC protocols. 
        Moreover, the same condition has been used in \cite{BKL21} to design a best-of-both-worlds
        protocol for atomic broadcast (a.k.a. state machine replication). 
         Furthermore, the same condition has been used recently
        in \cite{GLW22} to design a best-of-both-worlds approximate agreement protocol against computationally-bounded
        adversaries. 
                 
         A common principle used in \cite{BKL19,BZL20,DHL21} to design best-of-both-worlds protocol for a specific task $T$, which
         could be either BA or MPC,
          is the following: the parties first run
         a {\it synchronous} protocol for task $T$ with threshold $t_s$ assuming a synchronous network, which 
         {\it also} provides certain security guarantees in an asynchronous environment, tolerating $t_a$ corruptions. 
         After the known ``time-out" of the synchronous protocol, the parties
          run an {\it asynchronous} protocol for $T$ with threshold $t_a$, which also provides certain security guarantees in the presence of 
          $t_s$ corruptions. The input for the asynchronous protocol is decided based
         on the output the parties receive after the time-out of the synchronous protocol. The overall output  is then decided based
         on the output parties receive from the asynchronous protocol. If the task $T$ is
          MPC, then this means that the parties need to evaluate the circuit {\it twice}.
         We also follow a similar design principle as above, for our BA protocol. 
         However, for MPC, we {\it do not} require the parties to run two protocols and evaluate the circuit twice.
          Rather the parties need to evaluate the circuit {\it only once}.

\section{Preliminaries and Definitions}
\label{sec:Prelims}
We follow the pairwise secure-channel model, where the parties in $\PartySet = \{P_1, \ldots, P_n \}$ are connected by
 pairwise private and authentic channels. 
  The distrust in the system is modeled by a {\it computationally unbounded} Byzantine (malicious) adversary $\Adv$, who can corrupt
   a subset of the parties and force them to behave in any {\it arbitrary} fashion, during the execution of a protocol. We assume a
   {\it static} adversary, who decides the set of corrupt parties at the beginning of the protocol execution.
   The underlying network can be synchronous or asynchronous, with parties being {\it unaware} about the exact type.
   In a {\it synchronous} network, every sent message  is delivered in the same order, within some known fixed time
   $\Delta$. The adversary $\Adv$ can control up to $t_s$ parties in a synchronous network.

    In an {\it asynchronous} network, messages are sent with an arbitrary, yet finite delay,
      and {\it need not} be delivered in the same order. The only guarantee is that every sent message 
      is {\it eventually} delivered. The exact sequence of message delivery is decided by a {\it scheduler} and
      to model the worst case scenario, the scheduler is assumed to be under the control of $\Adv$. 
      The adversary can control up to $t_a$ parties in an asynchronous network. 
      
      We assume that $t_a < t_s$ and $3t_s + t_a < n$ holds. This automatically implies that
      $t_s < n/3$ and $t_a < n/4$ holds, which are necessary for any SMPC and AMPC protocol respectively.
      All computations in our protocols are done over a finite field $\F$, where $|\F| > 2n$  and where
      $\alpha_1, \ldots, \alpha_n, \beta_1, \ldots, \beta_n$ are publicly-known, distinct, non-zero elements
      from $\F$. For simplicity and without loss of generality, we assume that each $P_i$ has a private input $x^{(i)} \in \F$,
      and the parties want to securely compute a function $f:\F^n \rightarrow \F$.
      Without loss of generality, $f$ is represented by an arithmetic circuit $\ckt$ over $\F$, 
      consisting of linear and non-linear (multiplication) gates \cite{Gol04}, where
      $\ckt$ has $c_M$ number of multiplication gates and has a multiplicative depth of $D_M$.     
\paragraph{\bf Termination Guarantees of Our Sub-Protocols:}
 For simplicity, we will {\it not} be specifying any termination criteria for our sub-protocols. And the parties will keep on participating
  in these
  sub-protocol instances, even after receiving their outputs. The termination criteria of our MPC protocol will ensure that once a party terminates
   the MPC protocol, it terminates
  all underlying sub-protocol instances.
   We will use existing {\it randomized} ABA protocols 
   which ensure that the honest parties (eventually) obtain their respective output {\it almost-surely}.
   This means that the probability that an honest party obtains its output after participating for {\it infinitely} many rounds approaches $1$ {\it asymptotically} \cite{ADH08,MMR15,BCP20}.
   That is:
   \[ \underset{T \rightarrow \infty}{\mbox{lim}} \mbox{Pr}[\mbox{An honest } P_i \mbox{ obtains its output by local time } T] = 1,\]
   where the probability is over the random coins of the honest parties and the adversary in the 
   protocol. 
  The  property of {\it almost-surely} obtaining the output carries over to the ``higher" level protocols, where 
  ABA is used as a building block. 
  We will say that the ``{\it honest parties obtain some output almost-surely from (an asynchronous) protocol $\Pi$}" to mean that
  every {\it honest} $P_i$ {\it asymptotically} obtains its output in $\Pi$ with probability $1$, in above the sense.

      We next discuss the properties of polynomials over $\F$, which are used in our protocols.
       \paragraph{\bf Polynomials Over a Field:}
     A $d$-degree {\it univariate polynomial} over $\F$ is of the form
     \[ f(x) = a_0  + \ldots + a_d x^d,\]
      where each $a_i \in \F$. 
     An $(\ell, \ell)$-degree {\it symmetric bivariate polynomial} over $\F$
     is of the form
    \[ F(x, y) =  \sum_{i, j = 0}^{i = \ell, j = \ell}r_{ij}x^i y^j,\]
     where each $r_{ij} \in \F$ and
    where $r_{ij} = r_{ji}$ holds for all $i, j$. This automatically implies that
     $F(\alpha_j, \alpha_i) = F(\alpha_i, \alpha_j)$ holds, for all
    $\alpha_i, \alpha_j$. Moreover, $F(x, \alpha_i) = F(\alpha_i, y)$ also holds, for every $\alpha_i$.
    Given an $i \in \{1, \ldots, n \}$ and an $\ell$-degree polynomial $F_i(x)$, we say that $F_i(x)$ {\it lies} on 
   an $(\ell, \ell)$-degree symmetric bivariate polynomial $F(x, y)$, if $F(x, \alpha_i) = F_i(x)$ holds.
   
   We now state some standard known results, related to polynomials over $\F$. 
    It is a well known fact that there always exists a unique $d$-degree univariate polynomial, passing through
   $d+1$ distinct points. A generalization of this result for bivariate polynomials is that
    if there are ``sufficiently many" univariate polynomials which are ``pair-wise consistent",
   then together they lie on a unique bivariate polynomial. Formally:  
    \begin{lemma}[\bf \cite{CD05,AL17}]
  \label{lemma:bivariateI}
  Let $f_{i_1}(x), \ldots, f_{i_{q}}(x)$ be $\ell$-degree univariate polynomials over $\F$, where
 $q \geq \ell + 1$ and
 $i_1, \ldots, i_{q} \in \{1, \ldots, n \}$, such that
  $f_i(\alpha_j) = f_j(\alpha_i)$ holds for all
  $ i, j \in  \{i_1, \ldots, i_{q} \}$.  Then 
  $f_{i_1}(x), \ldots, f_{i_q}(x)$  lie on a unique
  $(\ell, \ell)$-degree symmetric bivariate polynomial, say $\starF(x, y)$.
  \end{lemma}
   In existing (as well as our) VSS protocol, $\D$ on having a $t$-degree polynomial $q(\cdot)$ as input, 
    embeds $q(\cdot)$
   into a random $(t, t)$-degree symmetric bivariate polynomial $F(x, y)$ at $x = 0$. 
   And each party $P_i$ then receives the $t$-degree univariate polynomial  $f_i(x) = F(x, \alpha_i)$. 
   Here $t$ is the maximum number of parties which can be under the control of $\Adv$.   
   This ensures 
   that $\Adv$ by learning at most $t$ polynomials lying on $F(x, y)$, 
   does not learn anything about $F(0, 0)$. Intuitively, this is because $\Adv$ will fall short of
   at least one point on $F(x, y)$ to uniquely interpolate it. 
     In fact, it can be shown that for every pair of $t$-degree polynomials $q_1(\cdot), q_2(\cdot)$ 
 such that $q_1(\alpha_i) = q_2(\alpha_i) = f_i(0)$ holds for every $P_i \in \Bad$ (where $\Bad$ is the set of parties
  under $\Adv$), the distribution of the polynomials
 $\{f_i(x) \}_{P_i \in \Bad}$ when $F(x, y)$ is chosen based on $q_1(\cdot)$, 
 is identical to the distribution when $F(x, y)$ is chosen based on $q_2(\cdot)$. Formally:
     
   \begin{lemma}[\bf \cite{CD05,AL17}]
  \label{lemma:bivariateII}
 Let $\Bad \subset \Partyset$ and $q_1(\cdot) \neq q_2(\cdot)$ be 
   $d$-degree polynomials where $d \geq |\Bad|$ such that $q_1(\alpha_i) = q_2(\alpha_i)$
     for all $P_i \in \Bad$. Then the probability distributions $\Big \{ \{F(x, \alpha_i) \}_{P_i \in \Bad} \Big \}$ and $\Big \{ \{F'(x, \alpha_i) \}_{P_i \in \Bad} \Big \}$ are identical, 
    where $F(x, y)$ and $F'(x, y)$ are random $(d, d)$-degree symmetric bivariate polynomials, such that
  $F(0, y) = q_1(\cdot)$ and $F'(0, y) = q_2(\cdot)$ holds.  
  \end{lemma}
We next give the definition of $d$-sharing, which is central to our protocols.
\begin{definition}[{\bf $d$-sharing}]
A value $s \in \F$ is said to be $d$-shared, if there exists a $d$-degree  {\it sharing-polynomial}, say $f_s(\cdot)$, with
 $f_s(0) = s$, such that every (honest) $P_i$ has the {\it share} $s_i = f_s(\alpha_i)$.
  The vector of shares of $s$ corresponding to the (honest) parties $P_i$ is called a $d$-sharing of $s$, 
   denoted by $[s]_d$. We will omit the degree $d$ from the notation $[\cdot]_d$ if it is clear from the context.
\end{definition}
It is easy to see that $d$-sharing satisfies the {\it linearity} property; i.e.~given 
 $[a]_d$ and $[b]_d$, then $[c_1 \cdot a + c_2 \cdot b]_d = c_1 \cdot [a]_d + c_2 \cdot [b]_d$ holds, where
 $c_1, c_2 \in \F$ are {\it publicly-known}. In general, consider any arbitrary linear function $g: \F^{\ell} \rightarrow \F^m$ and let 
 $u^{(1)}, \ldots, u^{(\ell)}$ be $d$-shared. When we say that {\it parties locally compute $([v^{(1)}]_d, \ldots, [v^{(m)}]_d) = g([u^{(1)}]_d, \ldots, [u^{(\ell)}]_d)$},
  we mean that the parties {\it locally} apply the function $g$ on their respective shares of $u^{(1)}, \ldots, u^{(\ell)}$,
   to get their respective shares of
  $v^{(1)}, \ldots, v^{(m)}$.       
  \subsection{Existing Primitives}
  We next discuss the existing primitives used in our protocols.
  \paragraph{\bf Online Error-Correction (OEC) \cite{BCG93}:}
   Let $\Partyset'$ be a subset of parties, containing 
    at most $t$ corrupt parties.
    And let there exist some $d$-degree polynomial $q(\cdot)$ with every (honest)
   $P_i \in \Partyset'$ having a point $q(\alpha_i)$. 
       The goal is 
    to make some {\it designated} party, say $P_R$, reconstruct $q(\cdot)$.
    For this, each $P_i \in \Partyset'$ sends $q(\alpha_i)$ to $P_R$, who 
    keeps waiting till it receives
    $d + t + 1$ points, all of which lie on a {\it unique} $d$-degree polynomial. 
     This step requires $P_R$ to repeatedly apply the Reed-Solomon (RS) error-correction procedure \cite{MS81}
      and try to recover $q(\cdot)$, upon
    receiving a new point from the parties in $\Partyset'$. 
    Once $P_R$ receives $d + 1 + 1$ points lying on a $d$-degree polynomial, say
     $q'(\cdot)$, then $q'(\cdot) = q(\cdot)$. This is because among these $d + t + 1$ points, at least $d + 1$ are from {\it honest}
     parties in $\Partyset'$, which uniquely determine $q(\cdot)$.
     If $d < (|\Partyset'| - 2t)$, then in an {\it asynchronous} network, $P_R$ {\it eventually} receives $d + t + 1$ points (from the {\it honest} parties in $\Partyset'$)
     lying on $q(\cdot)$ and recovers $q(\cdot)$. Moreover, in a {\it synchronous} network, it will take at most $\Delta$ time
     for $P_R$ to recover $q(\cdot)$, since the points of the honest parties will be delivered within $\Delta$ time.
     We denote the above procedure by $\OEC(d, t, \Partyset')$, which is presented in Appendix \ref{app:ExistingPrimitives},
     along with its properties.
  \paragraph{\bf Finding $\Star{t}$ \cite{BCG93}:}
  Let $G$ be an undirected graph over $\Partyset$. Then a pair $(\ESet,\FSet)$ where $\ESet \subseteq \FSet \subseteq \Partyset$
     is called an $\Star{t}$, if all the following hold.
       \begin{myitemize}
       \item[--] $|\ESet| \geq n - 2t$;
       \item[--] $|\FSet| \geq n - t$;
       \item[--] There exists an edge between every $P_i \in \ESet$ and every $P_j \in \FSet$.
       \end{myitemize}
     The work of 
    \cite{BCG93} presents an efficient algorithm (whose running time is polynomial in $n$),
     which we denote as  $\StarAlgo$. The algorithm always outputs an $\Star{t}$
     $(\ESet,\FSet)$, provided $G$ contains a clique of size at least $n - t$.
  \paragraph{\bf Asynchronous Reliable Broadcast (Acast):} 
   We use the Bracha's Acast protocol \cite{Bra84}, where there exists a designated {\it sender} $\Sender \in \Partyset$
   with input $m \in \{0, 1 \}^{\ell}$. The protocol allows $\Sender$ to send $m$
   {\it identically} to all the parties, in the presence of any $t < n/3$ corruptions, possibly including $\Sender$.
    While the protocol has been primarily designed for an {\it asynchronous}
   network, it also provides certain guarantees in a {\it synchronous} network, as stated in Lemma \ref{lemma:Acast}. 
  Notice that the protocol {\it does not} provide any liveness if $\Sender$ is {\it corrupt}, irrespective of the network type.
  This is because a {\it corrupt} $\Sender$ may not invoke the protocol in the first place. 
  Moreover in a {\it synchronous} network, if $\Sender$ is {\it corrupt} and if the honest parties compute an output, then they {\it may not}
  get the output at the same time. And there may be a difference of
  at most $2 \Delta$ time within which the honest parties compute their output. 
  The Acast protocol and proof of Lemma \ref{lemma:Acast} are available in Appendix \ref{app:ExistingPrimitives}.
  \begin{lemma}
  \label{lemma:Acast}
Bracha's Acast protocol $\PiACast$ achieves the following in the presence of up to $t < n/3$ corruptions, where $\Sender$
 has an input $m \in \{0, 1 \}^{\ell}$ for the protocol.
 \begin{myitemize}
\item[--] {\it Asynchronous Network}: 
     \begin{myitemize}
      \item[--] {\it (a) $t$-Liveness}: If $\Sender$ is {\it honest}, then all honest parties eventually obtain an output.
       \item[--] {\it (b) $t$-Validity}: If $\Sender$ is {\it honest}, then every honest party with an output, outputs $m$.
       \item[--] {\it (c) $t$-Consistency}: If $\Sender$ is {\it corrupt} and some honest party outputs $m^{\star}$,
        then every honest party {\it eventually} outputs $m^{\star}$.
    \end{myitemize}    
  \item[--] {\it Synchronous Network}:
         \begin{myitemize}
            \item[--] {\it (a) $t$-Liveness}: If $\Sender$ is {\it honest}, then all honest parties obtain an output within time $3\Delta$.
            \item[--] {\it (b) $t$-Validity}: If $\Sender$ is {\it honest}, then every honest party with an output, outputs $m$.
            \item[--] {\it (c) $t$-Consistency}: If $\Sender$ is {\it corrupt} and some honest party outputs $m^{\star}$ at time $T$, then every
             honest $P_i$ outputs $m^{\star}$ by the end of time $T + 2 \Delta$.
        \end{myitemize}      
  \item[--] Irrespective of the
  network type, $\Order(n^2 \ell)$ bits are communicated by the honest parties.
\end{myitemize}
\end{lemma}
We next discuss few terminologies with respect to $\PiACast$, which we use throughout the paper.
\paragraph{\bf Terminologies for Using $\PiACast$:}
 We will say that  ``{\it $P_i$ Acasts $m$}" to mean that $P_i$ acts as a sender $\Sender$ and invokes an instance of $\PiACast$
  with input $m$ and the parties participate in this instance. Similarly, we 
   will say that ``{\it $P_j$ receives $m$ from the Acast of $P_i$}" to mean that $P_j$ outputs $m$
  in the corresponding instance of $\PiACast$.

\section{best-of-both-worlds Perfectly-Secure Byzantine Agreement}
\label{sec:HBA}
  In this section, we present our best-of-both-worlds perfectly-secure Byzantine agreement (BA) protocol. 
   We begin with the definition of BA, which is a modified version of \cite{BKL19}, as we {\it do not} require any 
 {\it termination} guarantees. In the definition, we consider the case where the inputs of the parties is a single bit. However, the
  definition can be easily extended for the case when the inputs are bit-strings.
\begin{definition}[{\bf Byzantine Agreement (BA)} \cite{BKL19}]
\label{def:BA}
Let $\Pi$ be a protocol for the parties in $\Partyset$ with up to $t$ corrupt parties,
  where every $P_i$ has an input $b_i \in \{0, 1\}$ and a possible output from  
 $\{0, 1, \bot \}$.
     \begin{myitemize}
    \item[--] {\it $t$-Guaranteed Liveness}: $\Pi$ has guaranteed liveness, if all honest parties obtain an output.
       \item[--] {\it $t$-Almost-Surely Liveness}:  $\Pi$ has almost-surely liveness, if almost-surely, 
        all honest parties obtain some output.
      \item[--] {\it $t$-Validity}: $\Pi$ has $t$-validity, if the following hold:
      if all honest parties have input $b$, then every honest party with an output, outputs $b$. 
      \item[--] {\it $t$-Weak Validity}: $\Pi$ has $t$-weak validity, if the following hold:
       if all honest parties have input $b$, then every honest party with an output, outputs $b$ 
        or $\bot$. 
      \item[--] {\it $t$-Consistency}: $\Pi$ has $t$-consistency, 
      if all honest parties with an output, output the same value.
      \item[--] {\it $t$-Weak Consistency}: $\Pi$ has $t$-Weak Consistency, 
      if  all honest parties with an output, output either a common $v \in \{0, 1 \}$ or $\bot$.
      \end{myitemize}
 Protocol $\Pi$ is called a {\it $t$-perfectly-secure synchronous-BA} (SBA) protocol,
  if in a {\it synchronous} network, it achieves all the following:
      \begin{myitemize}
      \item[--] {\it $t$-guaranteed liveness};
      \item[--] {\it $t$-Validity};
      \item[--] {\it $t$-Consistency}.
      \end{myitemize}
 Protocol $\Pi$ is called a {\it $t$-perfectly-secure asynchronous-BA} (ABA) protocol, if in an {\it asynchronous network},
  it achieves the following:
     \begin{myitemize}
     \item[--] {\it $t$-almost-surely liveness};
     \item[--] {\it $t$-Validity};
     \item[--] {\it $t$-Consistency}.
     \end{myitemize}
 \end{definition}
 \noindent To design our best-of-both-worlds BA protocol, we will be using an
  {\it existing} perfectly-secure SBA and a perfectly-secure ABA protocol, whose properties
  we review next.
  \paragraph{\bf Existing $t$-Perfectly-Secure SBA:}
   We assume the existence of a $t$-perfectly-secure SBA protocol tolerating $t < n/3$ corruptions, which 
   {\it also} provides 
   {\it $t$-guaranteed liveness} in an {\it asynchronous} network.\footnote{We stress that we {\it do not} require any other
   property from the SBA protocol in an {\it asynchronous} network.}. For the sake of
    communication efficiency, we choose the recursive phase-king based
   {\it $t$-perfectly-secure} SBA protocol $\BGP$ of \cite{BGP92}.
   The protocol incurs a communication of $\Order(n^2 \ell)$ bits, if the inputs of the parties
   are of size $\ell$ bits. If the network is {\it synchronous}, then in protocol 
   $\BGP$, at time $\TimeBGP = (12n - 6) \cdot \Delta$, all honest parties have an output
    (see Lemma 10.7 of \cite{BA}).
   To ensure guaranteed liveness in an {\it asynchronous} network, the parties can simply run the protocol and then 
    check if any output is obtained at local time $(12n - 6) \cdot \Delta$.
    In case no output is obtained, then $\bot$ is taken as the output. The properties of $\BGP$ are summarized in 
    Lemma \ref{lemma:ExistingSBAGuarantees}.
  \begin{lemma}[\cite{BGP92,BA}]
 \label{lemma:ExistingSBAGuarantees}
 Let $t < n/3$. Then there exists a protocol $\BGP$ with the following properties, where all parties participate with an 
  input of size $\ell$ bits. 
   \begin{myitemize}
   \item[--] The protocol incurs a communication of $\Order(n^2 \ell)$ bits from the honest parties.
   \item[--] The protocol is a $t$-perfectly-secure SBA protocol, where all honest parties have an output within 
   time $\TimeBGP = (12n - 6) \cdot \Delta$.
   \item[--] In an asynchronous network, all honest parties have an output from $\{0, 1\}^{\ell} \cup \{ \bot\}$, within
   local time $(12n - 6) \cdot \Delta$.
   \end{myitemize}
 \end{lemma}
 \paragraph{\bf Existing $t$-Perfectly-Secure ABA:}
  Existing perfectly-secure ABA protocols achieve the following properties. 
 \begin{lemma}[\cite{ADH08,BCP20}]
 \label{lemma:ABAGuarantees}
 Let $t < n/3$. Then there exists a {\it randomized} protocol $\PiABA$, achieving the following properties,
  where the inputs of each party is a bit.
 \begin{myitemize}
 \item[--] {\it Asynchronous Network}: The protocol is a {\it $t$-perfectly-secure} ABA protocol and provides the following liveness
  guarantees.
     \begin{myitemize}
     \item[--]  If the inputs of all {\it honest} parties are same, then $\ABA$ achieves $t$-guaranteed liveness;
     \item[--] Else $\ABA$ achieves $t$-almost-surely liveness.
     \end{myitemize}
 \item[--] {\it Synchronous Network}: The protocol achieves $t$-validity, $t$-consistency and the following liveness guarantees.
    \begin{myitemize}
    \item[--] If the inputs of all {\it honest} parties are same, then $\ABA$ achieves $t$-guaranteed liveness and 
    all honest parties obtain their output within time $\TimeABA = k \cdot \Delta$ for some constant $k$. 
     \item[--] Else $\ABA$ achieves $t$-almost-surely liveness
    and requires $\Order(\mbox{poly}(n) \cdot \Delta)$ expected time to generate the output.
    \end{myitemize}
 \item[--] Irrespective of the network type, the protocol incurs the following amount of communication from the 
  honest parties.
      \begin{myitemize}
      \item[--] If the inputs of all honest parties are the same, then the protocol incurs a communication of $\Order(\mbox{poly(n)} \log|\F|)$ bits;
      \item[--] Else, it incurs an expected communication of
       $\Order(\mbox{poly(n)} \log|\F|)$ bits.\footnote{Looking ahead, the number of invocations of $\PiABA$ in our protocol will be a 
       {\it constant} and {\it independent} of the size of the circuit $\ckt$. Hence, we do not focus on the ``exact"
       communication complexity of $\PiABA$.}
      \end{myitemize}
 \end{myitemize}
 \end{lemma}
Protocol $\PiABA$ is designed using a {\it weaker} ``variant" of AVSS called  {\it shunning} AVSS (SAVSS) \cite{ADH08,BCP20}, 
 which {\it cannot} be used for circuit-evaluation.  
  We provide a brief overview of the ABA protocols of \cite{ADH08,BCP20} and a brief outline of the proof of 
         Lemma \ref{lemma:ABAGuarantees} in Appendix \ref{app:ExistingABAAnalysis}.

  From the above discussion, we note that protocol $\ABA$ {\it cannot} be considered as a {\it best-of-both-worlds} BA protocol.
  This is because the protocol 
   achieves {\it $t$-guaranteed liveness} in a {\it synchronous} network, {\it only} when {\it all} honest parties have the same input. 
   In case, the parties have a mixed bag of inputs, then the parties may end up running the protocol forever,
    without having
   an output, even if the
   network is {\it synchronous}, though the probability of this event is asymptotically $0$. 
      We design a {\it perfectly-secure} BA protocol, which solves this problem and which is secure in {\it any} network. 
  To design
the protocol, we need a special type of broadcast protocol, which we design first.
  \subsection{Synchronous Broadcast with Asynchronous Guarantees}
    We begin with the definition of broadcast, adapted from \cite{BKL19}, 
     where we {\it do not} put any {\it termination} requirement.
\begin{definition}[{\bf Broadcast} \cite{BKL19}]
\label{def:BCAST}
Let $\Pi$ be a protocol for the parties in $\PartySet$ consisting of up to
 $t$ corrupt parties, where a sender $\Sender \in \Partyset$ has input $m \in \{0, 1\}^{\ell}$, and parties obtain a possible output 
  from $\{0, 1 \}^{\ell} \cup \{\bot \}$.
  \begin{myitemize}
   \item[--] {\it $t$-Liveness}: $\Pi$ has $t$-liveness, if all honest parties obtain some output.
    \item[--] {\it $t$-Validity}: $\Pi$ has $t$-validity, if the following holds:
      if $\Sender$ is {\it honest}, then every honest party with an output, outputs $m$.
     \item[--] {\it $t$-Weak Validity}: $\Pi$ has $t$-validity, if the following holds:
      if $\Sender$ is {\it honest}, then every honest party outputs either $m$ or $\bot$.
     \item[--] {\it $t$-Consistency}: $\Pi$ has $t$-consistency, if the following holds: 
     if $\Sender$ is {\it corrupt}, then every honest party with an output, has a common output.
     \item[--] {\it $t$-Weak Consistency}: $\Pi$ has $t$-weak consistency, if the following holds: 
     if $\Sender$ is {\it corrupt}, then every honest party with an output, outputs a common $m^{\star} \in \{0, 1 \}^{\ell}$ or $\bot$.
  \end{myitemize}
Protocol $\Pi$ is called a {\it $t$-perfectly-secure broadcast} protocol, if it has the following properties:
  \begin{myitemize}
  \item[--]  {\it $t$-Liveness};
  \item[--] {\it $t$-Validity}; 
  \item[--] {\it $t$-Consistency}.
  \end{myitemize}
\end{definition}
We next design a {\it special} broadcast protocol $\BCAST$, which is a $t$-perfectly-secure broadcast protocol in a {\it synchronous} network. Additionally, in an {\it asynchronous} network, the protocol achieves {\it $t$-liveness}, {\it $t$-weak validity} and {\it $t$-weak consistency}.
 Looking ahead, we will combine the protocols $\BCAST, \BGP$ and $\PiABA$ to get our
  best-of-both-worlds BA protocol.

 Before proceeding to design $\BCAST$, we note that the existing Bracha's Acast protocol $\PiACast$
 {\it does not} guarantee the same properties as $\BCAST$. Specifically, for a {\it corrupt} $\Sender$, 
  there is {\it no} liveness guarantee (irrespective of the network type). Moreover, in a {\it synchronous}
 network, if $\Sender$ is {\it corrupt} and honest parties obtain an output, then they  
  {\it may not} obtain an output within the same time
    (see Lemma \ref{lemma:Acast}).\footnote{Looking ahead, this property from
    $\BCAST$ will be crucial when we use it in our best-of-both-worlds BA protocol.}
         Interestingly, our instantiation of $\BCAST$ is based on $\PiACast$, by carefully
          ``stitching" it with the protocol $\BGP$.

   The idea behind $\BCAST$ is the following: sender $\Sender$ first Acasts its message. 
   If the network is {\it synchronous} and $\Sender$ is {\it honest}, then within time 
   $3\Delta$, every honest party should have received $\Sender$'s message.
   To verify this, 
   the parties start participating in an instance of $\BGP$ at (local) time $3\Delta$,
    with their respective inputs being the output obtained from $\Sender$'s Acast
   at time $3\Delta$.
    If there is no output at time $3\Delta$ from $\Sender$'s Acast,
     then the input for $\BGP$ is $\bot$. Finally, at time $3\Delta + \TimeBGP$, parties output
   $m^{\star}$, if it has been received from the 
   Acast of $\Sender$ {\it and} if it is the output of $\BGP$ as well; otherwise the parties output $\bot$. 
   
   It is easy to see that the protocol has now {\it guarantees} liveness in
  {\it any} network (irrespective of $\Sender$), since all parties will have some output at (local) time $3\Delta + \TimeBGP$.
   Moreover, {\it consistency} is achieved for a {\it corrupt} $\Sender$ in a {\it synchronous} network, with {\it all} honest parties
   obtaining a common output at the {\it same} time. 
   This is because if any {\it honest} party obtains an output $m^{\star}\neq \bot$ at time 
   $3\Delta + \TimeBGP$, then {\it at least} one {\it honest} party must have received $m^{\star}$ from $\Sender$'s Acast by time $3\Delta$.
   And so by time $3\Delta + \TimeBGP$, {\it all} honest parties will receive $m^{\star}$ from $\Sender$'s Acast. 
 \paragraph {\bf Eventual Consistency and Validity in Asynchronous Network:}
 In $\BCAST$, the parties set a ``time-out" of $3\Delta + \TimeBGP$, due to which it
   provides
  {\it weak validity} and {\it weak consistency}  in an {\it asynchronous} network. This is because
   some {\it honest} parties may receive $\Sender$'s message from the Acast of $\Sender$ within the timeout, 
   while others may fail to do so. 
   The time-out is essential, as we need {\it liveness} from
   $\BCAST$ in {\it both} synchronous and asynchronous network, when 
   $\BCAST$ is used later in our best-of-both-worlds BA protocol.
       
      Looking ahead, we will use $\BCAST$ in our VSS protocol for broadcasting protocol.
       Due to the  {\it weak validity} and {\it weak consistency} properties, we may end up in a scenario where
       one subset of {\it honest} parties may output a common value different from $\bot$ at the end of the time-out, while 
      others may output $\bot$. For the security of the VSS protocol, we would require even the latter subset of (honest) parties 
      to {\it eventually} output the common non-$\bot$ value, if the parties {\it continue} participating in $\BCAST$. To achieve this goal,     
     every party who outputs $\bot$ at time $3\Delta + \TimeBGP$, ``switches" its output to $m^{\star}$, if
     it {\it eventually} receives $m^{\star}$ from $\Sender$'s Acast. 
    We stress that this switching is {\it only} for the parties who obtained $\bot$
    at time $3\Delta + \TimeBGP$. 
    To differentiate between the two ways of obtaining output, we use the terms {\it regular-mode}
    and {\it fallback-mode}. The regular-mode refers to the process of deciding the output at time $3\Delta + \TimeBGP$, while
    the fallback-mode refers to  the process of deciding the output beyond time
     $3\Delta + \TimeBGP$.\footnote{The fallback-mode is never triggered when $\BCAST$ is used in our best-of-both-worlds BA
     protocol. It will be triggered (along with the regular-mode) in our VSS protocol.}
   
      If the network is {\it asynchronous} and $\Sender$ is {\it honest}, then from the {\it liveness} and {\it validity} of $\PiACast$,
    every honest party {\it eventually} obtains $m$ from $\Sender$'s Acast.
    Hence, through the fallback-mode, every honest party who outputs $\bot$ at the time-out of $3\Delta + \TimeBGP$, eventually
    outputs $m$. Moreover, even if $\Sender$ is {\it corrupt}, the fallback-mode 
    {\it will not} lead to different honest parties obtaining different non-$\bot$ outputs 
    due to the {\it consistency} property of $\PiACast$.
       \begin{protocolsplitbox}{$\BCAST$}{Synchronous broadcast with asynchronous guarantees.}{fig:DS}
\centerline{\underline{(Regular Mode)}}
%
   \begin{myitemize}
   \item[--] On having the input $m \in \{0, 1 \}^{\ell}$, sender $\Sender$ Acasts $m$.
   \item[--] {\color{red} At time $3\Delta$}, each $P_i \in \Partyset$ participates in an instance of $\BGP$, 
    where the input for $\BGP$ is set as follows:
       \begin{myitemize}
       \item[--] $P_i$ sets  $m^{\star}$ as the input,  if $m^{\star} \in \{0, 1\}^{\ell}$ is received from the Acast of 
       $\Sender$;
       \item[--] Else $P_i$ sets $\bot$ as the input  (encoded as a default $\ell$-bit string).
       \end{myitemize}
  \item[--] {\bf (Local Computation)}: {\color{red} At time $3\Delta + \TimeBGP$}, each 
  $P_i \in \Partyset$ computes its output through {\it regular-mode} as follows:
     \begin{myitemize}
     \item[--] $P_i$ outputs $m^{\star} \neq \bot$, if $m^{\star}$ is received from the Acast of $\Sender$ {\it and} $m^{\star}$ is computed as the 		output during the instance of $\BGP$;
     \item[--] Else, $P_i$ outputs $\bot$. 
     \end{myitemize}
  Each $P_i \in \PartySet$ keeps participating in the protocol, even after computing the output. \\~\\
\end{myitemize}
\centerline{\underline{{\color{blue}(Fallback Mode)}}}
 \begin{myitemize}
 \item[--] {\color{blue} Every $P_i \in \Partyset$ who has computed the output $\bot$ at time $3\Delta + \TimeBGP$,
 changes it to $m^{\star}$, if $m^{\star}$ is received by $P_i$ from the Acast of $\Sender$.
 }
\end{myitemize}
\end{protocolsplitbox}

We next prove the properties of the protocol $\BCAST$.
\begin{theorem}
\label{thm:DS}
Protocol $\BCAST$ achieves the following properties in the presence of any $t < n/3$ corruptions,
 where $\Sender$ has an input $m \in \{0, 1 \}^{\ell}$ and where $\TimeBCAST =  (12n - 3) \cdot \Delta$.   
 \begin{myitemize}
   \item[--] {\it Synchronous} network: 
        \begin{myitemize}
           \item[--] {\it (a) $t$-Liveness}: At time $\TimeBCAST$, every honest party has an output, through regular-mode. 
	    \item[--] {\it (b) $t$-Validity}: If $\Sender$ is {\it honest}, then at time $\TimeBCAST$, each honest party
            outputs $m$ through regular-mode.
	    \item[--] {\it (c) $t$-Consistency}: If $\Sender$ is {\it corrupt}, then the output of every honest party is the same at 
	    time $\TimeBCAST$ through regular-mode.     
	    \item[--] {\it (d) $t$-Fallback Consistency}: If $\Sender$ is {\it corrupt} and some honest party
	     outputs $m^{\star} \neq \bot$ at time 
	     $T$ through fallback-mode, then every honest party outputs $m^{\star}$ by time $T + 2\Delta$ through fallback-mode.    
         \end{myitemize} 
\item[--] {\it Asynchronous Network}:
   \begin{myitemize}
     \item[--] {\it (a) $t$-Liveness}: At local time $\TimeBCAST$, every honest party has an output, through regular-mode. 
      \item[--] {\it (b) $t$-Weak Validity}: If $\Sender$ is {\it honest}, then at local time $\TimeBCAST$, 
       each honest party outputs $m$ or $\bot$ through regular-mode.
       \item[--] {\it (c) $t$-Fallback Validity}: If $\Sender$ is {\it honest}, then each honest party
        who outputs $\bot$ at local time $\TimeBCAST$ through regular-mode, eventually outputs 
         $m$ through fallback-mode.
        \item[--] {\it (d) $t$-Weak Consistency}: If $\Sender$ is {\it corrupt}, then at local time $\TimeBCAST$, 
    each honest party outputs either a common $m^{\star} \neq \bot$ or $\bot$, through regular-mode.
       \item[--] {\it (e) $t$-Fallback Consistency}: If $\Sender$ is {\it corrupt} and some honest party outputs $m^{\star} \neq \bot$ at local time
    $T$, either through regular or fallback-mode, 
         then every honest party eventually outputs $m^{\star}$, either through regular or fallback-mode.
   \end{myitemize}
   \item[--] Irrespective of the network type, the protocol incurs a communication of $\Order(n^2 \ell)$ bits from the honest parties.
   \end{myitemize}
\end{theorem}
\begin{proof}
The {\it liveness} (both for the {\it synchronous} as well {\it asynchronous} network) simply follows from the fact that 
 every honest party outputs something (including $\bot$) at (local)
  time $\TimeBCAST = 3\Delta + \TimeBGP$, where $\TimeBGP = (12n - 6) \cdot \Delta$.
   We next prove the  rest of the properties of the protocol in the {\it synchronous} network, for which we rely on the properties of
   Acast and $\BGP$ in the {\it synchronous} network.
 
 If $\Sender$ is {\it honest}, then due to the {\it liveness} and {\it validity} properties of $\PiACast$ in the {\it synchronous} network, 
  at time $3\Delta$, 
  every {\it honest} party $P_i$ receives $m$
  from the Acast of $\Sender$. Hence, every honest party participates with input $m$ in the instance of $\BGP$.
  From the {\it guaranteed liveness} and {\it validity} properties
   of $\BGP$ in {\it synchronous} network, at time $3\Delta + \TimeBGP$, every honest party will have
  $m$ as the output from $\BGP$. Hence, each honest party has the output $m$ at time $\TimeBCAST$, thus proving that {\it validity} is achieved.

    For {\it consistency}, we consider a {\it corrupt} $\Sender$.
    We first note that each honest party will have the {\it same} output from the instance of $\BGP$ at time $\TimeBCAST$, which follows from the
    {\it consistency} property of $\BGP$ in {\it synchronous} network.
    If all honest honest parties have the output $\bot$ for $\BCAST$ at time $\TimeBCAST$, then consistency holds trivially. So consider the case when
    some {\it honest} party, say $P_i$, has the output $m^{\star} \neq \bot$ for $\BCAST$ at time $\TimeBCAST$. 
    This implies that the output of $\BGP$ is $m^{\star}$ for every honest party.
    Moreover, it also implies that at time $3\Delta$, at least one {\it honest} party, say $P_h$, has received $m^{\star}$ from the Acast of $\Sender$. 
    Otherwise, all honest parties would participate with input $\bot$ in the instance of $\BGP$ and from the {\it validity} of $\BGP$ in the {\it synchronous}
    network, every honest party would compute $\bot$ as the output during $\BGP$, which is a contradiction.
    Since $P_h$ has received $m^{\star}$ from $\Sender$'s Acast at time $3\Delta$, it follows from the {\it consistency} property of 
    $\PiACast$ in the {\it synchronous}
    network that {\it all} honest parties will receive
     $m^{\star}$ from $\Sender$'s Acast by time $5\Delta$. Moreover, $5\Delta < (12n - 3) \cdot \Delta$ holds.
    Consequently, by
     time $(12n - 3) \cdot \Delta$, {\it all} honest parties will receive $m^{\star}$ from $\Sender$'s Acast and will have
     $m^{\star}$ as the output of
    $\BGP$ and hence, output $m^{\star}$ for $\BCAST$.   
    
    For {\it fallback consistency}, we have to consider a {\it corrupt} $\Sender$. Let $P_h$ be an {\it honest} party who outputs
     $m^{\star} \neq \bot$ at time $T$ through fallback-mode.
     Since the steps of fallback-mode are executed after time $\TimeBCAST$, it follows that 
     $T > \TimeBCAST$. We first note that this implies that {\it every} honest party has output $\bot$ at 
    time $\TimeBCAST$, through regular-mode.
     This is because, from the proof of the {\it consistency} property of $\BCAST$, if any {\it honest} party has an output 
    $m' \neq \bot$ at time $\TimeBCAST$, then {\it all} honest parties (including $P_h$) also must have computed the output 
    $m'$ at time $\TimeBCAST$, through regular-mode.
    And hence, $P_h$ will never change its output to $m^{\star}$.\footnote{Recall that in the protocol the parties who obtain
    an output different from $\bot$ at time $\TimeBCAST$, never change their output.} 
    Since $P_h$ has computed the output $m^{\star}$, it means that at time $T$, it has received $m^{\star}$ from the Acast of
    $\Sender$. It then follows from the {\it consistency} of $\PiACast$ in the {\it synchronous} network that every honest party will also receive 
    $m^{\star}$ from the Acast of $\Sender$, latest by time $T + 2 \Delta$ and output $m^{\star}$ through fallback-mode.

   We next prove the properties of the protocol $\BCAST$ in an {\it asynchronous} network, for which we depend upon the properties of $\PiACast$ in the {\it asynchronous} network.   
    The {\it weak-validity} property follows from the {\it validity} property of $\PiACast$ in the {\it asynchronous} network, which ensures that
    no honest party $P_i$ ever receives an $m'$ from the Acast of $\Sender$ where $m' \neq m$.
    So if at all $P_i$ outputs a value different from $\bot$ at time $\TimeBCAST$, it has to be $m$.
       The {\it weak-consistency} property follows using similar arguments
   as used to prove {\it consistency} in the {\it synchronous} network, but relying instead on the {\it validity} and 
     {\it consistency}
     properties of $\PiACast$ in the asynchronous network. The latter property ensures that even if the adversary has full control over message scheduling in the
   {\it asynchronous} network, it cannot ensure that
    for a {\it corrupt} $\Sender$, two different honest parties end up receiving $m_1$ and $m_2$ from the Acast of $\Sender$, where $m_1\neq m_2$. 
    
    For {\it fallback validity}, consider an {\it honest} $\Sender$ and let 
     $P_i$ be an {\it honest} party,
      who outputs $\bot$ at (local) time $\TimeBCAST$ through regular-mode.
       Since the parties keep on participating in the protocol beyond time $\TimeBCAST$,
       it follows from the {\it liveness} and {\it validity} properties of $\PiACast$ in the {\it asynchronous} network that party $P_i$ will {\it eventually} 
     receive
      $m$ from the Acast of $\Sender$ through the fallback-mode of $\BCAST$. Consequently, party $P_i$ eventually changes its output from
     $\bot$ to $m$.
     
     For {\it fallback consistency}, we consider a {\it corrupt} $\Sender$ and let 
     $P_j$ be an {\it honest} party, who outputs some $m^{\star}$ at time $T$ where $T \geq \TimeBCAST$.
      This implies that $P_j$ has
      obtained $m^{\star}$ from the Acast of $\Sender$.
           Now, consider an arbitrary {\it honest} $P_i$. From the {\it liveness} and {\it weak consistency} properties of $\BCAST$ in {\it asynchronous} network,
     it follows that $P_i$ outputs either $m^{\star}$ or $\bot$ at local time $\TimeBCAST$, through the regular-mode.
     If $P_i$ has output $\bot$, then from the {\it consistency}
     property of $\PiACast$ in the {\it asynchronous} network, it follows that 
     $P_i$ will also eventually obtain $m^{\star}$ from the Acast of $\Sender$ through the fallback-mode of $\BCAST$. 
     Consequently, party $P_i$ eventually changes its output from
     $\bot$ to $m^{\star}$.
     
     The {\it communication complexity} follows from the communication complexity of $\BGP$ and $\PiACast$.
\end{proof}
\noindent We next discuss few terminologies for $\BCAST$, which we use in the rest of the paper.
\paragraph {\bf Terminologies for $\BCAST$:}  
 When we say that  ``{\it $P_i$ broadcasts $m$}", we mean that $P_i$ invokes $\BCAST$ as  $\Sender$
  with input $m$ and the parties participate in this instance. Similarly,
   when we say that ``{\it $P_j$ receives $m$ from the broadcast of $P_i$ through regular-mode}",
    we mean that $P_j$ has the output $m$
   at time $\TimeBCAST$, during the instance of $\BCAST$.
   Finally, when we say that ``{\it $P_j$ receives $m$ from the broadcast of $P_i$ through fallback-mode}",
      we mean that $P_j$ has the output $m$
      after time $\TimeBCAST$
  during the instance of $\BCAST$. 
  \subsection{$\BCAST + \PiABA \Rightarrow$ best-of-both-worlds BA}
We now show how to combine the protocols $\BCAST$ and $\PiABA$ to get our best-of-both-worlds BA protocol $\HBA$.
 For this, we use
 an idea used in \cite{BKL19}, to get a best-of-both-worlds BA protocol with {\it conditional} security. In the protocol, 
 every party first broadcasts its input bit through an instance of $\BCAST$. If the network is
  {\it synchronous}, then all honest parties should have received the inputs of all the
  (honest) sender parties from their broadcasts through regular-mode, within time $\TimeBCAST$.
 Consequently, at time $\TimeBCAST$, the parties decide an output for all the $n$ instances
  of $\BCAST$. Based on these outputs, the parties decide their respective inputs for
 an instance of the $\PiABA$ protocol. Specifically, if “sufficiently many” outputs from the $\BCAST$
 instances are found to be the same, then the parties consider it as their input for
  the $\PiABA$ instance. Else, they stick to their original inputs. The overall output
  of the protocol is then set to be the output from $\PiABA$.
\begin{protocolsplitbox}{$\HBA$}{The best-of-both-worlds BA protocol.
  The above code is executed by every $P_i \in \PartySet$.}{fig:BA}
\justify
\begin{myitemize}
\item[--] 
 On having input $b_i \in \{0, 1 \}$, broadcast $b_i$.
\item[--] For $j = 1, \ldots, n$, let $b_i^{(j)} \in \{0, 1, \bot \}$ be received from the broadcast of $P_j$ through regular-mode. 
  Include $P_j$ to a set $\R$, if $b_i^{(j)} \neq \bot$. 
  Compute the input $v_i^{\star}$ for an instance of $\ABA$ as follows.
    \begin{myitemize}
    \item[--] If $|\R| \geq n - t$, then set $v_i^{\star}$ to the majority bit among 
    the $b_i^{(j)}$ values of the parties in $\R$.\footnote{If there is no majority, then set $v_i^{\star} = 1$.} 
    \item[--] Else set $v_i^{\star} = b_i$.
    \end{myitemize}
\item[--] {\color{red} At time $\TimeBCAST$}, participate in an instance of $\ABA$ with input $v_i^{\star}$.
 Output the result of $\ABA$.
\end{myitemize}
\end{protocolsplitbox}

We next prove the properties of the protocol $\HBA$. We note that 
 protocol $\HBA$ is invoked only $\Order(n^3)$ times in our MPC protocol, which is {\it independent} of $\ckt$.
  Consequently,  we do not 
  focus on the {\it exact} communication complexity of $\HBA$. 
\begin{theorem}
\label{thm:HBA}
Let $t < n/3$ and let $\PiABA$ be a randomized protocol, satisfying the conditions as per Lemma \ref{lemma:ABAGuarantees}.
   Then $\HBA$ achieves the following, where every party participates with an input bit.
\begin{myitemize}
\item[--]  {\bf Synchronous Network}: 
  The protocol is a {\it $t$-perfectly-secure} SBA protocol,
  where all honest parties obtain an output 
  within time $\TimeHBA = \TimeBCAST + \TimeABA$. 
  The protocol incurs a communication of $\Order(\mbox{poly(n)} \log|\F|)$ bits from the honest parties.
\item[--]  {\bf Asynchronous Network}: The protocol is a {\it $t$-perfectly-secure} ABA protocol
 with an expected communication of $\Order(\mbox{poly(n)} \log|\F|)$ bits.
\end{myitemize}
\end{theorem}
\begin{proof}
We start with the properties in a {\it synchronous} network. The {\it $t$-liveness} property of $\BCAST$ in the {\it synchronous}
 network guarantees
 that all honest parties will have some output, from each instance of $\BCAST$ through regular-mode, at time $\TimeBCAST$. Moreover, 
  the {\it $t$-validity} and {\it $t$-consistency} properties of $\BCAST$ in the {\it synchronous} network
   guarantee that irrespective of the sender parties,
  {\it all} honest parties will have a common output from {\it each} individual instance of $\BCAST$, at time $\TimeBCAST$.
  Now since the parties decide their respective inputs for the instance of $\ABA$ {\it deterministically} based on the individual outputs
  from the $n$ instances of $\BCAST$ at time $\TimeBCAST$, it follows that 
  all honest parties participate with a {\it common} input in the protocol $\ABA$. Hence, all honest parties obtain an output by the end of
  time $\TimeBCAST + \TimeABA$, thus ensuring {\it $t$-guaranteed liveness} of $\HBA$. 
  Moreover, the {\it $t$-consistency} property of $\ABA$ in the {\it synchronous} network
   guarantees that all honest parties have a {\it common} output from the instance of $\ABA$, which is taken as the output of
   $\HBA$,  thus
  proving the {\it $t$-consistency} of $\HBA$.

    For proving the {\it validity} in the synchronous network, let all {\it honest} parties have the same input bit $b$. 
         From the 
    {\it $t$-consistency} of $\BCAST$ in the {\it synchronous} network,
    all honest parties will receive $b$ as the output at time $\TimeBCAST$ in all the 
    $\BCAST$ instances, corresponding to the {\it honest} sender parties.
    Since there are at least $n - t$ honest parties, 
     it follows that all honest parties will find a {\it common} subset $\R$ in the protocol,
    as the set of honest parties constitutes a candidate $\R$. Moreover, all honest parties will be present in $\R$, as $n - t > t$ holds.
    Since the set of honest parties constitute a majority in $\R$, it follows that 
    all honest parties will participate with input $b$ in the
    instance of $\ABA$ and hence output $b$ at the end of $\ABA$, which follows from the
   {\it $t$-validity} of $\ABA$ in the {\it synchronous} network. This proves the {\it $t$-validity} of $\HBA$.
   
   We next prove the properties of $\HBA$ in an {\it asynchronous} network. 
    The {\it $t$-consistency} of the protocol $\HBA$ follows from the {\it $t$-consistency} of the protocol
    $\ABA$ in the {\it asynchronous} network, since the overall output of the protocol $\HBA$ is same as the output
    of the protocol $\ABA$. The {\it $t$-liveness} of the protocol $\BCAST$ in the {\it asynchronous}
    network guarantees that all honest parties will have some output from all the $n$ instances of $\BCAST$
    at local time $\TimeBCAST$ through regular-mode. Consequently, all honest parties will participate with some input 
    in the instance of $\ABA$. The {\it $t$-almost-surely liveness} of
    $\ABA$ in the {\it asynchronous} network then implies the
    {\it $t$-almost-surely liveness} of
    $\HBA$.
    
    For proving the validity in an {\it asynchronous} network, let all {\it honest} parties have the same input bit $b$. 
     We claim that all {\it honest} parties participate with input
    $b$ during the instance of $\ABA$. The {\it $t$-validity} of $\ABA$ in the {\it asynchronous} network then
    automatically implies the  {\it $t$-validity} of $\HBA$. 
    
    To prove the above claim, consider an arbitrary
    {\it honest} party $P_h$. There are two possible cases. If $P_h$ fails to find a subset $\R$ satisfying the
    protocol conditions, then the claim holds trivially, as $P_h$ participates in the instance of $\ABA$ with its input for
    $\HBA$, which is the bit $b$. So consider the case when $P_h$ finds a subset $\R$, such that
    $|\R| \geq n - t$ and where corresponding to each $P_j \in \R$, 
    party $P_h$ has computed an output $b_h^{(j)} \in \{0, 1 \}$ at local time $\TimeBCAST$ during the instance $\BCAST^{(j)}$,
    through regular-mode.
    Now consider the subset of {\it honest} parties in the set $\R$.
    Since $t < n/3$, it follows that $n - 2t > t$ and hence the {\it majority} of the parties in $\R$ will be {\it honest}. 
       Moreover, $P_h$ will compute the output $b$ at local time $\TimeBCAST$ in the instance of $\BCAST$, corresponding to
    {\it every honest} $P_j$ in $\R$, which follows from the {\it $t$-weak validity} of $\BCAST$ in the {\it asynchronous} network.
    From these arguments, it follows that $P_h$ will
     set $b$ as its input for the instance of $\ABA$, thus proving the claim.
    
   The communication complexity, both in a synchronous as well as in an asynchronous network, follows easily from the 
   protocol steps and from the communication complexity of $\SBA$ and $\ABA$.
\end{proof}

\section{best-of-both-worlds Perfectly-Secure VSS}
\label{sec:HVSS}
In this section, we present our best-of-both-worlds VSS protocol $\HSh$. In the protocol, there 
   exists a designated {\it dealer}
   $\D \in \PartySet$. The input for $\D$ consists of  $L$ number of $t_s$-degree polynomials $q^{(1)}(\cdot), \ldots, q^{(L)}$, where $L \geq 1$.
    And each (honest) $P_i$ is supposed to ``verifiably" receive 
    the {\it shares} $\{q(\alpha_i)\}_{\ell = 1, \ldots, L}$. Hence, the goal is to generate a $t_s$-sharing of 
     $\{q(0) \}_{\ell = 1, \ldots, L}$.\footnote{Note that the degree of $\D$'s polynomials is {\it always} $t_s$, {\it irrespective} of the
     underlying network type.}
   If $\D$ is {\it honest}, then in an {\it asynchronous} network, each (honest) $P_i$ {\it eventually} gets its shares, while in a {\it synchronous}
   network, $P_i$ gets its shares after some {\it fixed} time, such that the view of the adversary remains independent of $\D$'s polynomials. 
   The {\it verifiability} here ensures that if $\D$ is {\it corrupt}, then either no honest party obtains any output (if $\D$ does not invoke the protocol), or
    there exist $L$ number of $t_s$-degree polynomials, such that 
    $\D$ is ``committed" to these polynomials and 
    each honest $P_i$ gets its shares lying on these polynomials.
    Note that in the latter case, we {\it cannot} bound the time within which honest parties will have their shares, even
   if the network is {\it synchronous}. This is because a {\it corrupt} $\D$ may delay sending the messages arbitrarily and the parties {\it will not} know the exact network
   type.
   To design $\HSh$, we first design a ``weaker" primitive called {\it weak polynomial-sharing} (WPS), whose
    security guarantees are {\it identical} to that of VSS for an
   {\it honest} $\D$. However, for a {\it corrupt} $\D$, the security guarantees are ``weakened", as only a subset of the 
    honest parties may get their shares of the committed polynomials.
  \subsection{The best-of-both-worlds Weak Polynomial-Sharing (WPS) Protocol}
   For simplicity, we explain our WPS protocol $\WPS$, assuming $\D$ has a {\it single} $t_s$-degree polynomial $q(\cdot)$ as input.
    Later we discuss the modifications needed to handle $L$ polynomials efficiently. 
     Protocol $\WPS$ is obtained by carefully
      ``stitching" a {\it synchronous} WPS protocol with an {\it asynchronous} WPS protocol. We first explain these two individual protocols, followed by the
     procedure to stitch them together, where the parties will {\it not be} knowing the exact network type.
  \paragraph{\bf WPS in an Asynchronous Network:}
   In an {\it asynchronous} network, one can consider the following
   protocol $\PiAsynWPS$:
    $\D$ embeds $q(\cdot)$ in a random
    $(t_s, t_s)$-degree symmetric bivariate polynomial $Q(x, y)$ at $x = 0$ and distributes univariate polynomials lying on $Q(x, y)$ to respective parties.
  To verify whether $\D$ has distributed ``consistent" polynomials, the parties check for the
  pair-wise consistency of their supposedly common points and make public the results 
  through $\OK$ messages, if the tests are ``positive".
  Based on the $\OK$ messages, the parties prepare a {\it consistency graph} and look for an $\Star{t_a}$, say $(\ESet', \FSet')$.
   If $\D$ is {\it honest}, then $(\ESet', \FSet')$ will be obtained eventually, since the honest parties form a clique of size at least $n - t_a$.
   The existence of $(\ESet', \FSet')$ guarantees that the polynomials of the {\it honest} parties in 
  $\FSet'$ lie on a single
    $(t_s, t_s)$-degree symmetric bivariate polynomial $Q^{\star}(x, y)$, where $Q^{\star}(x, y) = Q(x, y)$ for an {\it honest} $\D$. This is because 
    $\ESet'$ has at least $t_s + 1$ {\it honest} parties with pair-wise consistent polynomials, defining $Q^{\star}(x, y)$.
    And the polynomial of every {\it honest} party in $\FSet'$ is pair-wise consistent with the polynomials of every {\it honest} party in $\ESet'$.
    The parties {\it outside}
    $\FSet'$ obtain their polynomials lying on $Q^{\star}(x, y)$ by applying OEC on the common points on these polynomials received from the parties in
    $\FSet'$.  Every $P_i$ then outputs $Q^{\star}(0, \alpha_i)$ as its share, which is same as $q^{\star}(\alpha_i)$, where $q^{\star}(\cdot) = Q^{\star}(0, y)$.
    Note that $\PiAsynWPS$ actually constitutes a VSS in the {\it asynchronous} network, as for an {\it honest} $\D$ every honest party eventually gets its share. 
    On the other hand, for a
    {\it corrupt} $\D$, every honest party eventually gets its share, if some honest party gets its share.
    \paragraph{\bf WPS for Synchronous Network:}
    $\PiAsynWPS$ fails in a {\it synchronous} network if there are $t_s$ corruptions. This is because
     only $n - t_s$ honest parties are guaranteed and hence the 
   parties may {\it fail} to find an $\Star{t_a}$. The existence of an $\Star{t_s}$, say $(\ESet, \FSet)$, in the consistency graph
   is not ``sufficient" to conclude that $\D$ has distributed consistent polynomials, lying on a $(t_s, t_s)$-degree symmetric bivariate polynomial.
    This is because if $\D$ is {\it corrupt}, then $\ESet$ is guaranteed to have {\it only}  $n - 2t_s - t_s > t_a$ {\it honest} parties, with pair-wise consistent polynomials.
    Whereas to define a $(t_s, t_s)$-degree symmetric bivariate polynomial, we need more than $t_s$ pair-wise consistent polynomials.
    On the other hand, if $\D$ is {\it corrupt}, then the honest parties in $\FSet$ {\it need not} constitute a clique. As a result, the polynomials of the
    honest parties in $\FSet$ need not be pair-wise consistent and hence need not
    lie on a 
     $(t_s, t_s)$-degree symmetric bivariate polynomial.
   
   To get rid of the above problem, the parties instead
    look for a ``special" $\Star{t_s}$ $(\ESet, \FSet)$, where the polynomials of all {\it honest} parties in $\FSet$ are guaranteed to lie on a  
   single $(t_s, t_s)$-degree symmetric bivariate polynomial. Such a special $(\ESet, \FSet)$ is bound to exist for an {\it honest}
   $\D$. 
   %
   Based on the above idea, protocol $\PiSynWPS$ for a {\it synchronous} network proceeds as follows.
   For ease of understanding, we explain the protocol as a sequence of communication {\it phases}, with the parties being {\it synchronized} in each phase. 
   
    In the {\it first} phase, $\D$ distributes the univariate polynomials, during the {\it second} phase
    the parties perform pair-wise consistency tests and during the {\it third} phase the parties make public 
    the results of {\it positive} tests. Additionally, the parties
    {\it also} make public the results of ``negative" tests through $\NOK$ messages and their respective versions of the disputed points
    (the
   $\NOK$ messages {\it were not} required for $\PiAsynWPS$). The parties then construct the consistency graph.
   Next $\D$ {\it removes} all the parties from consideration in its consistency graph, who
   have made public ``incorrect" $\NOK$ messages, whose version of the disputed points are {\it incorrect}. 
   Among the {\it remaining} parties, $\D$ checks for the presence of a set of at least $n - t_s$ parties $\WCORE$, such that the polynomial of 
   {\it every} party in $\WCORE$ is publicly confirmed to be pair-wise consistent with the polynomials of at least $n - t_s$ parties within $\WCORE$.
   If a $\WCORE$ is found, then $\D$ checks for the presence of an $\Star{t_s}$, say $(\ESet, \FSet)$ among  $\WCORE$
   and broadcasts $(\WCORE, \ESet, \FSet)$ during the {\it fourth} phase, if $\D$ finds $(\ESet, \FSet)$.
          The parties upon receiving $(\WCORE, \ESet, \FSet)$ verify if $\WCORE$ is of size at least $n - t_s$ and every party in $\WCORE$ has an edge with at least
          $n - t_s$ parties within $\WCORE$ in their local copy of the consistency graph. The parties also check whether indeed 
          $(\ESet, \FSet)$ constitutes an $\Star{t_s}$ among the parties within $\WCORE$. 
       Furthermore, the parties now {\it additionally} verify whether any pair of parties $P_j, P_k$
   from $\WCORE$ have made public  ``conflicting"
   $\NOK$ messages during the {\it third} phase. That is, if there exists any $P_j, P_k \in \WCORE$ who made public 
   $\NOK$ messages with $q_{jk}$ and $q_{kj}$ respectively during the {\it third} phase such that $q_{jk} \neq q_{kj}$, then $\WCORE$ is {\it not} accepted.
    The idea here is that if $\D$ is {\it honest}, then at least one of $P_j, P_k$ is bound to be {\it corrupt}, whose corresponding $\NOK$ message is incorrect. 
    Since 
   $\D$ also would have seen these public $\NOK$ messages during the third phase, it should have have discarded the corresponding corrupt party, 
   before finding $\WCORE$. Hence if
   a $(\WCORE, \ESet, \FSet)$  is accepted at the end of fourth phase, then the polynomials of all {\it honest}
   parties in $\WCORE$ are guaranteed to be pair-wise consistent and lie on a single $(t_s, t_s)$-degree symmetric bivariate polynomial, say
   $Q^{\star}(x, y)$, where $Q^{\star}(x, y) = Q(x, y)$ for an {\it honest} $\D$. This will further guarantee that the polynomials of all {\it honest}
   parties in $\FSet$ {\it also} lie on $Q^{\star}(x, y)$, as $\FSet \subseteq \WCORE$.

   If a $(\WCORE, \ESet, \FSet)$ is accepted, then each $P_i \in \WCORE$ outputs the constant term of its univariate polynomial as its share.
   On the other hand, the parties outside $\WCORE$ attempt to obtain their corresponding polynomials lying on 
   $Q^{\star}(x, y)$ by applying OEC on the common points on these polynomials received from the parties in
    $\FSet$. And if a $t_s$-degree polynomial is obtained, then the constant term of the polynomial is set as the share.
    For  an {\it honest}
   $\D$, each {\it honest} $P_i$ will be present in $\WCORE$ and hence will have the share $q(\alpha_i)$. On the other hand, if a $\WCORE$ is accepted
   for a {\it corrupt} $\D$, then all the {\it honest} parties in $\WCORE$ (which are at least $t_s + 1$ in number) will have their shares lying on $q^{\star}(\cdot) = Q^{\star}(0, y)$.
   Moreover, even if an {\it honest} party $P_i$ {\it outside} $\WCORE$ is able to compute its share, then it is the same as
   $q^{\star}(\alpha_i)$ due to the OEC mechanism. However, for a {\it corrupt} $\D$, all the {\it honest} parties {\it outside} $\WCORE$ may {\it not} be able to obtain their
   desired share, as $\FSet$ is guaranteed to have only $n - 2t_s > t_s + t_a$ {\it honest} parties and OEC may fail. It is precisely for this reason that
   $\PiSynWPS$ {\it fails} to qualify as a VSS.  
   \paragraph{\bf $\PiSynWPS + \PiAsynWPS \Rightarrow$ best-of-both-worlds WPS Protocol $\WPS$:}
   We next discuss how to combine protocols $\PiSynWPS$ and $\PiAsynWPS$ to get our best-of-both-worlds
    WPS protocol called $\WPS$. In protocol $\WPS$ (Fig \ref{fig:WPS}), the parties first run $\PiSynWPS$ {\it assuming} a {\it synchronous} network,
 where $\BCAST$ is used to make any value public by setting $t = t_s$ in the protocol $\BCAST$.
  If $\D$ is {\it honest} then in a {\it synchronous} network, the first, second, third and fourth phase of $\PiSynWPS$ would have been over by time
  $\Delta, 2\Delta, 2\Delta + \TimeBCAST$ and $2\Delta + 2\TimeBCAST$ respectively and 
   by time $2\Delta + 2\TimeBCAST$,
   the parties should have accepted a $(\WCORE, \ESet, \FSet)$. 
    However, in an {\it asynchronous} network, parties may have different ``opinion" regarding the acceptance of a $(\WCORE, \ESet, \FSet)$. 
     This is because it may be possible that only a subset of the honest parties accept a $(\WCORE, \ESet, \FSet)$ within local time 
     $2\Delta + 2\TimeBCAST$.
    Hence at time $2\Delta + 2\TimeBCAST$, the parties run an instance of our best-of-both-worlds BA protocol $\HBA$, to check whether
    any $(\WCORE, \ESet, \FSet)$ is accepted.
   
    If the parties conclude that a
   $(\WCORE, \ESet, \FSet)$ is accepted, then the parties compute their {\it WPS-shares} as per $\PiSynWPS$. 
   However, we need to ensure that for a {\it corrupt} $\D$ in a {\it synchronous} network, 
   if the polynomials of the {\it honest} parties in $\WCORE$ are {\it not} pair-wise consistent, then the corresponding conflicting
   $\NOK$ messages are received within time $2\Delta + \TimeBCAST$ (the time required for the  {\it third} phase of $\PiSynWPS$
   to be over), so that $(\WCORE, \ESet, \FSet)$ is {\it not} accepted.
   This is ensured by enforcing even a {\it corrupt} $\D$ to send the respective polynomials of all the {\it honest} parties in $\WCORE$ by time $\Delta$, so that the 
   pair-wise consistency test between every pair of {\it honest} parties in $\WCORE$ is over by time $2\Delta$. 
    For this, the parties are asked to wait for some ``appropriate" time,
   before starting the pair-wise consistency tests and also before making public the results of pair-wise consistency tests. 
   The idea is to ensure that if the polynomials of the {\it honest} parties in $\WCORE$ are {\it not} delivered within time $\Delta$ (in a {\it synchronous} network),
    then the results of the
   pair-wise consistency tests also get {\it delayed} beyond time $2\Delta + \TimeBCAST$ (the time-out of the {\it third} phase of $\PiSynWPS$).
   This in turn will ensure that no
    $\WCORE$ is accepted within time $2\Delta + 2\TimeBCAST$.
    
     If the parties conclude that no
 $(\WCORE, \ESet, \FSet)$ is accepted within time $2\Delta + 2\TimeBCAST$, then it implies that
  either $\D$ is {\it corrupt} or the network is {\it asynchronous} and hence the parties resort to 
  $\PiAsynWPS$. However, $\D$ {\it need not} have to start afresh and distribute polynomials on a 
  ``fresh" bivariate polynomial.
  Instead, $\D$ continues with the consistency graph formed using the $\OK$ messages received 
 as part of $\PiSynWPS$ and searches for an $\Star{t_a}$. If $\D$ is {\it honest} and the network is {\it asynchronous},
  then the parties eventually obtain their shares.
  
  Notice that in an {\it asynchronous} network, it might be possible that the parties (through $\HBA$) conclude 
  that $(\WCORE, \ESet, \FSet)$ is accepted, if some honest party(ies) accepts a $(\WCORE, \ESet, \FSet)$ within the time-out $2\Delta + 2\TimeBCAST$.
  Even in this case, the polynomials of all {\it honest} parties in $\WCORE$ lie on a single 
  $(t_s, t_s)$-degree symmetric bivariate polynomial $Q^{\star}(x, y)$ for a {\it corrupt} $\D$.
   This is because there will be at least $n - 2t_s - t_a > t_s$ {\it honest} parties in $\ESet$ with pair-wise consistent polynomials, defining
  $Q^{\star}(x, y)$, and the polynomial of every {\it honest} party in $\FSet$ will be pair-wise consistent with the polynomials of {\it all honest} 
  parties in $\ESet$ and hence lie on $Q^{\star}(x, y)$ as well. Now consider any {\it honest} $P_i \in (\WCORE \setminus \FSet)$.
  As part of $\PiSynWPS$, it is ensured that the polynomial of $P_i$ is consistent with the polynomials of at least $n - t_s$ parties among $\WCORE$.
  Among these $n - t_s$ parties, at least $n - 2t_s - t_a > t_s$ will be {\it honest} parties from $\FSet$.
  Thus, the polynomial of $P_i$ will also lie on $Q^{\star}(x, y)$. 
     \begin{protocolsplitbox}{$\WPS(\D, q(\cdot))$}{The best-of-both-worlds weak polynomial-sharing protocol for a single polynomial.}{fig:WPS}
\justify
  \begin{myitemize}
   \item {\bf Phase I --- Sending Polynomials}: 
	\begin{myitemize}
	      \item[--] $\D$ on having the input $q(\cdot)$, chooses a random $(t_s, t_s)$-degree symmetric bivariate polynomial
	      $Q(x, y)$ such that $Q(0, y) = q(\cdot)$ and sends $q_i(x) = Q(x, \alpha_i)$ to each party $P_i \in \Partyset$. 
	\end{myitemize}
    \item {\bf Phase II --- Pair-Wise Consistency}: Each $P_i \in \PartySet$ on receiving a $t_s$-degree polynomial $q_i(x)$ from $\D$ does the following.
             \begin{myitemize}
            \item[--] {\color{red} Wait till the local time becomes a multiple of $\Delta$} and then send $q_{ij} =  q_i(\alpha_j)$ to $P_j$, for $j = 1, \ldots, n$.   
             \end{myitemize}
     \item {\bf Phase III --- Publicly Declaring the Results of Pair-Wise Consistency Test}: Each $P_i \in \Partyset$ does the following.
           \begin{myitemize}
           \item[--] Upon receiving $q_{ji}$ from $P_j$, {\color{red} wait till the local time becomes a multiple of $\Delta$}.
           If a $t_s$-degree polynomial $q_i(x)$ has been received from $\D$, then
           do the following.
           \begin{myitemize}
           \item[--] Broadcast $\OK(i, j)$, if $q_{ji} = q_i(\alpha_j)$ holds.
           \item[--] Broadcast $\NOK(i, j, q_i(\alpha_j))$, if $q_{ji} \neq q_i(\alpha_j)$ holds.
           \end{myitemize}
          \end{myitemize}      
    \item {\bf Local computation --- Constructing Consistency Graph:} Each $P_i \in \Partyset$ does the following.
          \begin{myitemize}
           \item[--] Construct a {\it consistency graph} $G_i$ over $\Partyset$, where the edge $(P_j, P_k)$ is included in $G_i$, if $\OK(j, k)$ and $\OK(k, j)$ is received from the broadcast of $P_j$ and $P_k$ respectively, either through the regular-mode or fall-back mode.	
          \end{myitemize}
     \item {\bf Phase IV  --- Checking for an $\Star{t_s}$}: $\D$ does the following in its consistency graph
     $G_\D$ at  time $2\Delta + \TimeBCAST$.   
      \begin{myitemize}
          \item[--] Remove edges incident with $P_i$, if $\NOK(i, j,q_{ij})$ is received from the broadcast of $P_i$ through regular-mode
           and $q_{ij} \neq Q(\alpha_j,\alpha_i)$.
          \item[--] Set $\WCORE = \{P_i: \mathsf{deg}(P_i) \geq n-t_s\}$, where $\mathsf{deg}(P_i)$ denotes the degree of $P_i$ in $G_\D$.
           \item[--] Remove $P_i$ from $\WCORE$, if $P_i$ is {\it not} incident with at least $n - t_s$ parties in $\WCORE$. Repeat this step till no more parties can be 
           removed from $\WCORE$.
          \item[--] Run algorithm $\StarAlgo$ on $G_\D[\WCORE]$, where $G_\D[\WCORE]$ denotes the subgraph of $G_\D$ induced by the vertices in $\WCORE$.
           If an $\Star{t_s}$, say $(\ESet, \FSet)$, is obtained, then broadcast $(\WCORE, \ESet, \FSet)$.
      \end{myitemize}
     \item {\bf Local Computation --- Verifying and Accepting $(\WCORE, \ESet, \FSet)$: } Each $P_i \in \Partyset$ does the following at time $2\Delta + 2\TimeBCAST$.
	     \begin{myitemize}
    	    \item[--] If a $(\WCORE, \ESet, \FSet)$ is received from $\D$'s broadcast through regular-mode, then {\it accept} it if following {\it were true} at time $2\Delta + \TimeBCAST$:
    	    \begin{myitemize}
                \item[--] There exist no $P_j, P_k \in \WCORE$, such that 
                $\NOK(j, k, q_{jk})$ and $\NOK(k, j, q_{kj})$ messages {\it were} received from the broadcast of $P_j$ and $P_k$ respectively through regular-mode, 
                where $q_{jk} \neq q_{kj}$.
                \item[--] In the consistency graph $G_i$,
                  $\mathsf{deg}(P_j) \ge n-t_s$  for all $P_j \in \WCORE$.
                \item[--] In the consistency graph $G_i$, every $P_j \in \WCORE$ has 
                  edges with at least $n - t_s$ parties from $\WCORE$. 
                \item[--] $(\ESet, \FSet)$ {\it was} an $\Star{t_s}$ in the induced graph $G_i[\WCORE]$.
                \item[--] For every $P_j, P_k \in \WCORE$ where the edge $(P_j, P_k)$ is present in $G_i$, the $\OK(j, k)$ and $\OK(k, j)$ messages {\it were} received from the
                 broadcast of $P_j$ and $P_k$ respectively, through regular-mode.
            \end{myitemize} 
        \end{myitemize}
   \item {\bf Phase V --- Deciding Whether to Go for an $\Star{t_a}$}:
     At time $2\Delta + 2\TimeBCAST$, each $P_i \in \Partyset$ participates in an instance of $\HBA$ with input 
     $b_i = 0$ if a $(\WCORE, \ESet, \FSet)$ was accepted, else with input $b_i = 1$, and {\color{red} waits for time $\TimeHBA$}.
   \item {\bf Local Computation --- Computing WPS-share Through $\WCORE$}:
     If the output of $\HBA$ is 0, then each $P_i \in \Partyset$ computes its {\it WPS-Share} $s_i$ (initially set to $\bot$) as follows.
            \begin{myitemize}
            \item[--] If a $(\WCORE, \ESet, \FSet)$ is not yet received then wait till a $(\WCORE, \ESet, \FSet)$ is received from $\D$'s broadcast through fall-back mode.
            \item[--] If $P_i \in \WCORE$, then output $s_i = q_i(0)$.
            \item[--] Else, initialise a support set $\Support_i$ to $\emptyset$. 
             If $q_{ji}$ is received from $P_j \in \FSet$, include $q_{ji}$ to $\Support_i$. Keep executing $\OEC(t_s, t_s, \Support_i)$, 
             till a $t_s$-degree polynomial, say $q_i(\cdot)$, is obtained. Then, output $s_i = q_i(0)$.                             
                \end{myitemize}
     \item {\bf Phase VI --- Broadcasting an $\Star{t_a}$}: If the output of $\HBA$ is 1, then $\D$ does the following.
        \begin{myitemize}
            \item[--] After every update in the consistency graph $G_\D$, run $\StarAlgo$ on $G_\D$. If an $\Star{t_a}$, say
             $(\ESet', \FSet')$, is obtained, then broadcast $(\ESet', \FSet')$.
        \end{myitemize}
   \item \textbf{Local Computation --- Computing WPS-share Through $\Star{t_a}$}: If the output of $\HBA$ is 1, then each $P_i$ does the following to compute
    its {\it WPS-Share}.
        \begin{myitemize}
        \item[--] Waits till an $\Star{t_a}$ $(\ESet', \FSet')$ is obtained from the broadcast of $\D$, either through regular or fall-back mode. Upon receiving, wait
         till $(\ESet', \FSet')$ becomes an $\Star{t_a}$ in $G_i$.
       \item[--] If $P_i \in \FSet'$, then output $s_i= q_i(0)$. 
       \item[--] Else, initialise a support set $\Support_i$ to $\emptyset$. 
            If $q_{ji}$ is received from $P_j \in \FSet'$, include $q_{ji}$ to $\Support_i$. Keep executing $\OEC(t_s, t_s, \Support_i)$, 
             till a $t_s$-degree polynomial, say $q_i(\cdot)$, is obtained. Then, output $s_i = q_i(0)$.    
        \end{myitemize}
  \end{myitemize}    
%
\end{protocolsplitbox}

We next proceed to prove the properties of the protocol $\WPS$. We begin with showing that if $\D$ is {\it honest}, then the adversary does not learn
 anything additional about $q(\cdot)$, irrespective of the network type.
\begin{lemma}[{\it $t_s$-Privacy}]
 \label{lemma:WPSPrivacy}
In protocol $\WPS$, if $\D$ is honest, then irrespective of the network type, the view of the adversary remains independent of $q(\cdot)$.
 \end{lemma}
\begin{proof}
Let $\D$ be {\it honest}. We consider the worst case scenario, when the adversary controls up to $t_s$ parties. 
 We claim that throughout the protocol, the adversary learns at most $t_s$ univariate polynomials lying on $Q(x, y)$. 
  Since $Q(x, y)$ is a random 
  $(t_s, t_s)$-degree- symmetric-bivariate polynomial, it then follows from Lemma \ref{lemma:bivariateII} that the view of the adversary will be independent of 
   $q(\cdot)$. 
   We next proceed to prove the claim.

 Corresponding to every {\it corrupt} $P_i$, the adversary learns $Q(x, \alpha_i)$. Corresponding to every {\it honest} $P_i$, the adversary learns $t_s$ distinct points on 
  $P_i$'s univariate polynomial $Q(x, \alpha_i)$, through the pair-wise consistency checks. However, these points were already included in the view of the adversary
    (through the univariate polynomials under adversary's control).
    Hence no additional information about the polynomials of the honest parties is revealed during the pair-wise consistency checks.
   Furthermore, no {\it honest} $P_i$ ever broadcasts $\NOK(i, j, q_i(\alpha_j))$, corresponding to
    any {\it honest} $P_j$.  This is because the pair-wise consistency check will always pass for every pair of {\it honest} parties.
\end{proof}
We next prove the correctness property in a {\it synchronous} network.
\begin{lemma}[{\it $t_s$-Correctness}]
 \label{lemma:WPSSynCorrectness}
In protocol $\WPS$, if $\D$ is honest and the network is synchronous, then each honest $P_i$ outputs 
 $q(\alpha_i)$ at time $\TimeWPS = 2\Delta +  2\TimeBCAST + \TimeHBA$.
 \end{lemma}
\begin{proof}
Let $\D$ be {\it honest} and the network be {\it synchronous} with up to $t_s$ corruptions. 
During phase I, every {\it honest} party $P_j$ receives $q_j(x) = Q(x, \alpha_j)$ from $\D$ within time $\Delta$.
 Hence during phase II, 
every {\it honest} $P_j$ sends $q_{jk}$ to every $P_k$, which takes at most $\Delta$ time to be delivered.
 Hence, by time $2\Delta$, every {\it honest} $P_j$
 receives $q_{kj}$ from every {\it honest} $P_k$, such that $q_{kj} = q_j(\alpha_k)$ holds. 
 Consequently, during phase III, every {\it honest} $P_j$ broadcasts $\OK(j, k)$ corresponding to every {\it honest} $P_k$, and vice versa. From the {\it $t_s$-validity} 
 property of $\BCAST$ in the {\it synchronous} network, 
 it follows that every honest $P_i$ receives $\OK(j, k)$ and $\OK(k, j)$ from the broadcast of every {\it honest} $P_j$ and every {\it honest} $P_k$ respectively,
 through regular-mode,
  at
  time $2\Delta + \TimeBCAST$. Hence, the edge $(P_j, P_k)$ will be added to the consistency 
   graph $G_i$, corresponding to every honest $P_j, P_k$. 
    Furthermore, from the {\it $t_s$-consistency} property of $\BCAST$, the graph $G_i$ will be the same for every honest party $P_i$ (including $\D$)
     at time $2\Delta + \TimeBCAST$.   Moreover,
     if $\D$ receives  an {\it incorrect} $\NOK(i, j, q_{ij})$ message from the broadcast of any {\it corrupt}
    $P_i$ through regular-mode at time $2\Delta + \TimeBCAST$, where $q_{ij} \neq Q(\alpha_j, \alpha_i)$, then $\D$ removes all the edges incident with
  $P_i$ in $\D$'s consistency graph $G_\D$. Dealer
  $\D$ then computes the set $\WCORE$, and all {\it honest} parties will be present in $\WCORE$. Moreover, the honest parties will
  form a clique of size at least $n - t_s$ in the induced subgraph $G_\D[\WCORE]$ at time $2\Delta + \TimeBCAST$
   and $\D$ will find an $\Star{t_s}$, say $(\ESet,\FSet)$, in $G_\D[\WCORE]$ and broadcast  $(\WCORE, \ESet, \FSet)$ during phase IV. 
  By the {\it $t_s$-validity} of $\BCAST$ in the {\it synchronous} network, 
  all honest parties will receive $(\WCORE, \ESet, \FSet)$ through regular-mode at time $2\Delta + 2\TimeBCAST$.
  Moreover, all honest parties will {\it accept}  $(\ESet, \FSet)$   and participate with input $0$ in the instance of $\HBA$.
 Hence, by the {\it $t_s$-validity} and {\it $t_s$-guaranteed liveness} of $\HBA$ in the {\it synchronous} network, 
 {\it every} honest party obtains the output $0$ in the instance of $\HBA$,  by time $2\Delta + 2\TimeBCAST + \TimeHBA$. 
 Now, consider an arbitrary {\it honest} party $P_i$. Since $P_i \in \WCORE$, party $P_i$ outputs $s_i = q_i(0) = Q(0, \alpha_i) = q(\alpha_i)$.
\end{proof}
We next prove the correctness property in an {\it asynchronous} network.
\begin{lemma}[{\it $t_a$-Correctness}]
\label{lemma:WPSAsynCorrectness}
In protocol $\WPS$, if $\D$ is honest and network is asynchronous, then almost-surely,
 each honest $P_i$ eventually outputs $q(\alpha_i)$.
\end{lemma}
\begin{proof}
Let $\D$ be {\it honest} and network be {\it asynchronous} with up to $t_a$ corruptions. We first note that every honest party participates with some input in the instance
  of $\HBA$ at local time
  $2\Delta + 2 \TimeBCAST$. Hence
  from the {\it $t_a$-almost-surely liveness} and {\it $t_a$-consistency} of $\HBA$ in an {\it asynchronous} network, 
  it follows that  almost-surely, the instance of $\HBA$ eventually generates 
   some common output, for all honest parties. Now there are two possible cases:
    \begin{myitemize}
     \item[--] {\bf The output of $\HBA$ is 0}: From the {\it $t_a$-validity} of $\HBA$ in the {\it asynchronous} network, 
      it follows that at least one {\it honest} party, say $P_h$, participated with input $0$ during the instance of $\HBA$. 
      This implies that $P_h$ 
       has accepted a $(\WCORE, \ESet, \FSet)$ at local $2\Delta +2\TimeBCAST$, which is
       received from the broadcast of $\D$, through {\it regular-mode}.
       Hence, by the {\it $t_a$-weak validity} and {\it $t_a$-fallback validity} properties
        of $\BCAST$ in the {\it asynchronous} network, 
            all honest parties will eventually receive $(\WCORE, \ESet, \FSet)$ from the broadcast of $\D$ and {\it accept} the triplet. 
            This is because 
            the consistency graphs of all honest parties will eventually have all the edges which were present in the consistency graph
            $G_h$ of $P_h$, at time $2\Delta +2\TimeBCAST$.
            We claim that {\it every honest} $P_i$ will eventually get $Q(x, \alpha_i)$.
             This will imply that eventually, every {\it honest} $P_i$ outputs $s_i = Q(0, \alpha_i) = q(\alpha_i)$.
            To prove the claim, consider an arbitrary {\it honest} party $P_i$. There are two possible cases.
            \begin{myitemize}
                \item[--] {\it $P_i \in \WCORE$}: In this case, $P_i$ already has received $Q(x, \alpha_i)$ from $\D$.
                \item[--] {\it $P_i \not \in \WCORE$}:
                In this case, there will be at least $n - t_s > 2t_s + t_a$ parties in $\FSet$, of which at most $t_a$ could be {\it corrupt}.
                Since $Q(x, \alpha_i)$ is a $t_s$-degree polynomial and $t_s < |\FSet| - 2t_a$, 
                from Lemma \ref{lemma:OECProperties}, it follows that by applying the OEC procedure on the common points on the
                polynomials $Q(x, \alpha_i)$, received from the parties in $\FSet$, party $P_i$ will eventually obtain 
                $Q(x, \alpha_i)$.
            \end{myitemize}
  \item[--] \textbf{The output of $\HBA$ is 1}: Since $\D$ is {\it honest}, every pair of honest parties $P_j, P_k$ eventually
   broadcast $\OK(j, k)$ and $\OK(k, j)$ messages respectively, as the pair-wise consistency check between them will eventually be successful. 
   From the {\it $t_a$-weak validity} and {\it $t_a$-fallback validity} of $\BCAST$, these messages are eventually delivered to every honest party.
   Also from the {\it $t_a$-weak consistency} and {\it $t_a$-fallback consistency}
    of $\BCAST$ in the {\it asynchronous} network, any $\OK$ message which is received by $\D$ from the broadcast of any {\it corrupt} party,
    will be eventually received by every other honest party as well.
  As there will be at least $n - t_a$ honest parties, a clique of size at least $n - t_a$ will eventually form in the consistency graph of every honest party. 
   Hence $\D$ will eventually find an $\Star{t_a}$, say $(\ESet', \FSet')$, in its consistency graph and broadcast it. From 
   the {\it $t_a$-weak validity} and {\it $t_a$-fallback validity} of $\BCAST$, this star will be eventually delivered to every honest party. 
   Moreover, $(\ESet', \FSet')$ will be eventually an $\Star{t_a}$ in every honest party's consistency graph.
       We claim that {\it every honest} $P_i$ will eventually get $Q(x, \alpha_i)$. 
       This will imply that eventually, every {\it honest} $P_i$ outputs $s_i = Q(0, \alpha_i) = q(\alpha_i)$.
            To prove the claim, consider an arbitrary {\it honest} party $P_i$. There are two possible cases.
            \begin{myitemize}
                \item[--] $P_i \in \FSet'$: In this case, $P_i$ already has  $Q(x, \alpha_i)$, received  from $\D$.
                \item[--] $P_i \notin \FSet'$:   In this case, there will be at least $n - t_a > 3t_s $ parties in $\FSet'$, of which at most $t_a$ could be {\it corrupt}, where $t_a < t_s$.
                Since $Q(x, \alpha_i)$ is a $t_s$-degree polynomial and $t_s < |\FSet'| - 2t_a$, from Lemma \ref{lemma:OECProperties}
                 it follows that by applying the OEC procedure on the common points on the
                polynomial $Q(x, \alpha_i)$, received from the parties in $\FSet'$, party $P_i$ will eventually obtain 
                $Q(x, \alpha_i)$.
            \end{myitemize}
    \end{myitemize}
\end{proof}

We next proceed to prove the weak commitment properties for a {\it corrupt} $\D$. However, before that we prove a helping lemma.
\begin{lemma}
\label{lemma:WPS}
Let the network be synchronous and let $\D$ be corrupt in the protocol $\WPS$. 
 If any one honest party receives a $(\WCORE, \ESet, \FSet)$ from the broadcast of $\D$ through regular-mode and accepts it at time $2\Delta + 2\TimeBCAST$, 
 then all the following hold.
 \begin{myitemize}
\item[--] All honest parties in $\WCORE$ receive their respective $t_s$-degree univariate polynomials from $\D$, within time $\Delta$.
\item[--] The univariate polynomials $q_i(x)$ of all honest parties $P_i$ in the set $\WCORE$ lie on a unique
 $(t_s, t_s)$-degree symmetric bivariate polynomial, say $Q^{\star}(x, y)$. 
\item[--] Within time $2\Delta + 2\TimeBCAST$, every honest party accepts $(\WCORE, \ESet, \FSet)$.
\end{myitemize}
\end{lemma}
\begin{proof}
Let $\D$ be {\it corrupt} and network be {\it synchronous} with up to $t_s$ corruptions. As per the lemma condition, let $P_h$ be an {\it honest} party, who 
  receives a $(\WCORE, \ESet, \FSet)$ from the broadcast of $\D$ through regular-mode and accepts
   it at time $2\Delta + 2\TimeBCAST$. Then from the protocol steps, the following must be
  true for $P_h$ at time $2\Delta + \TimeBCAST$:
 	     \begin{myitemize}
              \item[--] There does not exist any               
              $P_j, P_k \in \WCORE$,
               such that $\NOK(j, k, q_{jk})$ and $\NOK(k, j,q_{kj})$ messages are received by $P_h$, from the broadcast of 
              $P_j$ and $P_k$ respectively through regular-mode, where $q_{jk} \neq q_{kj}$.
               \item[--] In $P_h$'s consistency graph $G_h$,  $\mathsf{deg}(P_j) \ge n-t_s$  for all $P_j \in \WCORE$ and $P_j$ has edges with at least
               $n - t_s$ parties from $\WCORE$. 
               \item[--] $(\ESet, \FSet)$ constitutes an $\Star{t_s}$ in the induced subgraph $G_h[\WCORE]$, such that for every
           $P_j, P_k \in \WCORE$ where the edge $(P_j, P_k)$ is present in $G_h$, the messages $\OK(j, k)$ and $\OK(k, j)$ are
            received by $P_h$, from the broadcast of 
          $P_j$ and $P_k$ respectively, through regular-mode.
          \end{myitemize} 
 We prove the first part of the lemma through a contradiction. So let $P_j \in \WCORE$ be an {\it honest} party, who receives its
  $t_s$-degree univariate polynomial, say $q_j(x)$, from $\D$ at time $\Delta + \delta$, where
  $\delta > 0$. Moreover, let $P_k \in \WCORE$ be an {\it honest} party, different from $P_j$ (note that there are at least $n - 2t_s$ {\it honest} parties in $\WCORE$).
  As stated above, at time $2\Delta + \TimeBCAST$, party $P_h$ receives the message $\OK(k, j)$ from the broadcast of $P_k$ through regular-mode.
  From the protocol steps, $P_j$ waits till its local time becomes a multiple of $\Delta$, before it sends the points on its polynomial to other parties for pair-wise consistency tests.
 Hence, $P_j$ must have started sending the points after time $c \cdot \Delta$, where $c \geq 2$.
  Since the network is {\it synchronous}, the point $q_{jk} = q_j(\alpha_k)$
 must have been received by $P_k$ by time $(c + 1) \cdot \Delta$. 
  Moreover, from the protocol steps, even if $P_k$ receives these points at time $T$, where $c \cdot \Delta < T < (c + 1) \cdot \Delta$,
  it waits till time $(c + 1) \cdot \Delta$, before broadcasting the $\OK(k, j)$ message. 
   Since $P_k$ is {\it honest}, from the {\it $t_s$-validity} property of $\BCAST$ in the {\it synchronous} network,
  it will take {\it exactly} $\TimeBCAST$ time for the message $\OK(k, j)$ to be received through regular-mode, once it is broadcast. 
 This implies that $P_h$ will receive the message $\OK(k, j)$
 at time $(c + 1) \cdot \Delta + \TimeBCAST$, where $(c + 1) > 2$. However, this is a contradiction, 
  since the $\OK(k, j)$ message has been received by $P_h$ at time $2\Delta + \TimeBCAST$.
 
 To prove the second part of the lemma, we will show that 
   the univariate polynomials $q_j(x)$ of all the {\it honest} parties in $\WCORE$ are pair-wise consistent. 
   Since there are at least $n - 2t_s > t_s$ {\it honest} parties in $\WCORE$, 
             from Lemma \ref{lemma:bivariateI}, it  follows
              that the univariate polynomials $q_j(x)$ 
              of all the {\it honest} parties in $\WCORE$ lie on a unique $(t_s, t_s)$-degree symmetric bivariate polynomial, say $Q^{\star}(x,y)$. 
   So consider an arbitrary pair of {\it honest} parties $P_j, P_k \in \WCORE$. From the first part of the lemma, both $P_j$ and $P_k$ must have received their
    respective univariate polynomials
  $q_j(x)$ 
  and $q_k(x)$ by time $\Delta$. This further implies that $P_j$ and $P_k$ must have received the points $q_{kj} = q_k(\alpha_j)$ and 
   $q_{jk} = q_j(\alpha_k)$ respectively by time
  $2\Delta$. If $q_{kj} \neq q_{jk}$, then $P_j$ and $P_k$ would broadcast $\NOK(j, k, q_{jk})$ and $\NOK(k, j, q_{kj})$ 
  messages respectively, at time $2\Delta$.  
   Consequently,  from the {\it $t_s$-validity} property of $\BCAST$ in the {\it synchronous} network, 
    $P_h$ will receive these messages through regular-mode at time $2\Delta + \TimeBCAST$. 
     Consequently, $P_h$ {\it will not} accept $(\WCORE, \ESet, \FSet)$, which is a contradiction.
  
  To prove the third part of the lemma, we note that since $P_h$ has received 
   $(\WCORE, \ESet, \FSet)$ from the broadcast of $\D$ through regular-mode  at time $2\Delta + 2\TimeBCAST$, it implies
  that $\D$ must have started broadcasting  $(\WCORE, \ESet, \FSet)$ latest at time $2\Delta + \TimeBCAST$. 
   This is because it takes $\TimeBCAST$ time for the regular-mode of $\BCAST$ to produce an 
   output. From the
  {\it $t_s$-consistency} property of $\BCAST$ in the  {\it synchronous} network, it follows that every honest party will also receive $(\WCORE, \ESet, \FSet)$
   from the broadcast of $\D$ through regular-mode,
  at time $2\Delta + 2\TimeBCAST$. Since at time $2\Delta + \TimeBCAST$, party $P_h$ has received the $\OK(j, k)$ and $\OK(k, j)$ messages
   through regular-mode
   from the broadcast of 
  every $P_j, P_k \in \WCORE$ where $(P_j, P_k)$ is an edge in $P_h$'s consistency graph, 
   it follows that these messages started getting broadcast, latest at time $2\Delta$.
  From the {\it $t_s$-validity}
   and {\it $t_s$-consistency} properties of $\BCAST$ in the {\it synchronous} network, 
    it follows that every honest party receives these broadcast messages through regular-mode at time
   $2\Delta + \TimeBCAST$. Hence $(\ESet, \FSet)$ will constitute an $\Star{t_s}$ in the induced subgraph $G_i[\WCORE]$ of every honest party $P_i$'s 
   consistency-graph at time $2\Delta + \TimeBCAST$ 
  and consequently, every honest party accepts $(\WCORE, \ESet, \FSet)$.
\end{proof}

Now based on the above helping lemma, we proceed to prove the weak commitment properties of the protocol $\WPS$.
\begin{lemma}[{\it $t_s$-Weak Commitment}]
 \label{lemma:WPSSynWeakCommitment}
In protocol $\WPS$, if $\D$ is corrupt and network is synchronous, then either no honest party computes any output or there exists some
 $t_s$-degree polynomial, say $q^{\star}(\cdot)$, such that all the following hold.
   \begin{myitemize}
      \item[--] There are at least $t_s + 1$ honest parties $P_i$ who
          output the WPS-shares $q^{\star}(\alpha_i)$.
        \item[--]  If any honest $P_j$ outputs a WPS-share $s_j \in \F$, then
          $s_j = q^{\star}(\alpha_j)$ holds.
   \end{myitemize}
 \end{lemma}
 \begin{proof}
 Let $\D$ be {\it corrupt} and network be {\it synchronous} with up to $t_s$ corruptions. 
 If no honest party outputs any wps-share, then the lemma holds trivially. So consider the case when some honest party
   outputs a wps-share,  which is an element of $\F$. 
 Now, there are two possible cases.
    \begin{myitemize}
         \item[--] \textit{At time $2\Delta + 2\TimeBCAST$,  at least one honest party, say $P_h$, accepts a $(\WCORE, \ESet, \FSet)$, received 
         from the broadcast of $\D$ through regular-mode}:
         In this case, from Lemma \ref{lemma:WPS}, at time $2\Delta + 2\TimeBCAST$, every honest party will accept $(\WCORE, \ESet, \FSet)$.
         Hence every honest party participates in the instance of 
         $\HBA$ with input $0$. From the {\it $t_s$-validity} and {\it $t_s$-guaranteed liveness} properties of $\HBA$ in the {\it synchronous} network, 
          all honest parties will get the output $0$ during the instance of
           $\HBA$ by time $\TimeWPS = 2\Delta + 2\TimeBCAST + \TimeHBA$. From Lemma \ref{lemma:WPS}, 
          the univariate polynomials of all the {\it honest} parties in $\WCORE$ will lie on some
           $(t_s, t_s)$-degree symmetric bivariate polynomial, say $Q^{\star}(x,y)$. 
         Let $q^{\star}(\cdot) \defined Q^{\star}(0, y)$. 
         Now consider an arbitrary {\it honest}
          party $P_i$, who outputs a wps-share $s_i \in \F$. 
          We want to show that the
          condition $s_i = q^{\star}(\alpha_i)$ holds. 
          And there are at least $t_s + 1$ such {\it honest} parties $P_i$ who output their wps-share.
          There are two possible cases.
        \begin{myitemize}
            \item[--] {\it $P_i \in \WCORE$}: From the protocol steps, $P_i$ sets $s_i = Q^{\star}(0, \alpha_i)$,
             which is the same as $q^{\star}(\alpha_i)$.
            Since $\WCORE$ contains at least $t_s + t_a + 1$ honest parties, this also shows that at least $t_s + 1$ honest parties $P_i$ output 
            their respective wps-share $s_i \in \F$, which is the same as $q^{\star}(\alpha_i)$.
            \item[--] {\it $P_i \not \in \WCORE$}:
                In this case, $P_i$ sets $s_i = q_i(0)$, where $q_i(\cdot)$ is a $t_s$-degree univariate polynomial,
                 obtained by applying the OEC procedure with $d = t = t_s$,
                 on the values $q_{ji}$, received from the parties  $P_j \in \FSet$, during the pair-wise consistency checks.
                  Note that as part of OEC (see the proof of Lemma \ref{lemma:OECProperties}),
                  party $P_i$ verifies that at least $2t_s + 1$ $q_{ji}$ values from the parties in $\FSet$ lie on 
                  $q_i(\cdot)$. Now out of these $2t_s +1$ $q_{ji}$ values, at least $t_s + 1$ values are from the {\it honest} parties in $\FSet$. Furthermore, these 
                 $q_{ji}$ values from the {\it honest} parties in $\FSet$ are the same as $Q^{\star}(\alpha_i, \alpha_j)$, which is equal to 
                 $Q^{\star}(\alpha_j, \alpha_i)$ and uniquely determine
                 $Q^{\star}(x, \alpha_i)$; the last property holds since $Q^{\star}(x, y)$ is a symmetric bivariate polynomial.
                  This automatically implies that $q_i(x)$ is the {\it same} as $Q(x, \alpha_i)$ and hence 
                 $s_i = q^{\star}(\alpha_i)$, since two different $t_s$-degree polynomials can have at most $t_s$ common values.
        \end{myitemize}
         \item[--] \textit{At time $2\TimeBCAST + 2\Delta$, no honest party has accepted any $(\WCORE, \ESet, \FSet)$}: 
         This implies that all honest parties participate in the instance of $\HBA$ with input $1$. So by the {\it $t_s$-validity} and {\it $t_s$-guaranteed liveness} of
          $\HBA$ in the {\it synchronous} network,
         all honest parties obtain the output $1$ in the instance of $\HBA$. Let $P_h$ be the {\it first honest} party who outputs a wps-share, consisting of an element
         from $\F$. 
         This means that $P_h$ has received a pair $(\ESet', \FSet')$,
          from the broadcast of $\D$, such that $(\ESet', \FSet')$ constitutes an $\Star{t_a}$
         in $P_h$'s consistency graph. By the {\it $t_s$-consistency} and {\it $t_s$-fallback consistency} properties of $\BCAST$ in the
         synchronous network, 
         {\it all} honest parties receive $(\ESet', \FSet')$
         from the broadcast of $\D$. Moreover, since $(\ESet', \FSet')$ constitutes an $\Star{t_a}$ in $P_h$'s consistency graph, it will also constitute an
         $\Star{t_a}$ in every other honest party's consistency graph as well. This is because the $\OK(\star, \star)$ messages which are received by
         $P_h$ from the broadcast of the various parties in $\ESet'$ and $\FSet'$, are also received by every other honest party, either through regular-mode
         or fallback-mode.
         The last property follows from the {\it $t_s$-validity, $t_s$-consistency} and {\it $t_s$-fallback consistency} properties of $\BCAST$ in the
          synchronous network.        
         Since $|\ESet'| \geq n - 2t_a > 2t_s + (t_s-t_a) > 2t_s$, it follows that $\ESet'$ has at least $t_s + 1$ {\it honest} parties $P_i$, whose 
         univariate polynomials $q_i(x)$ are pair-wise consistent. Hence, from Lemma \ref{lemma:bivariateI}, these univariate polynomials
         lie on a unique 
         $(t_s,t_s)$-degree symmetric bivariate polynomial, say $Q^{\star}(x,y)$. 
         Similarly, since the univariate polynomial  $q_i(x)$ of every {\it honest} party in $\FSet'$ is pair-wise consistent with the univariate polynomials  $q_j(x)$ 
         of the {\it honest} parties in $\ESet'$,
         it implies that the univariate polynomials $q_i(x)$
         of all the {\it honest} parties in $\FSet'$ also lie on $Q^{\star}(x, y)$.
               Let $q^{\star}(\cdot) \defined Q^{\star}(0, y)$. We show that {\it every honest} $P_i$ 
               outputs a wps-share, which is the same as
          $q^{\star}(\alpha_i)$. For this it is enough to show that each honest $P_i$ gets $q_i(x) = Q^{\star}(x, \alpha_i)$,
         as $P_i$ outputs $q_i(0)$ as its wps-share, which will be then same as $q^{\star}(\alpha_i)$.
         Consider an arbitrary {\it honest} party $P_i$. There are two possible cases. 
         \begin{myitemize}
             \item[--] \textit{$P_i \in \FSet'$}: In this case, $P_i$ already has $Q^{\star}(x, \alpha_i)$, received from $\D$.
             \item[--] \textit{$P_i \notin \FSet'$}: In this case, there will be $n-t_a > 3t_s$ parties in $\FSet'$, of which at most $t_s$ could be corrupt.
              Moreover, $Q^{\star}(x, \alpha_i)$ is a $t_s$-degree polynomial and $t_s < |\FSet'| - 2t_s$ holds. 
              Hence from the properties of OEC (Lemma \ref{lemma:OECProperties}),  by applying the OEC procedure on the common points on the
                polynomial $Q^{\star}(x, \alpha_i)$ received from the parties in $\FSet'$, party $P_i$ will compute 
                $Q^{\star}(x, \alpha_i)$.
         \end{myitemize}
 \end{myitemize}
 \end{proof}

We finally prove the commitment property in an {\it asynchronous} network.
\begin{lemma}[{\it $t_a$-Strong Commitment}]
 \label{lemma:WPSAsynStrongCommitment}
In protocol $\WPS$, if $\D$ is corrupt and network is asynchronous, then either no honest party computes any output or there exist some 
 $t_s$-degree polynomial, say $q^{\star}(\cdot)$, such that almost-surely, every honest $P_i$ eventually outputs a wps-share
  $q^{\star}(\alpha_i)$.\footnote{Note that {\it unlike} the synchronous network,
   the commitment property in the {\it asynchronous} network is {\it strong}. That is, if at all any honest party outputs a wps-share, then {\it all} the honest
   parties are guaranteed to eventually output their wps-shares.}			   
 \end{lemma}
\begin{proof}
 Let $\D$ be {\it corrupt} and network be {\it asynchronous} with up to $t_a$ corruptions. If no honest party computes
  any output, then the lemma holds trivially. So consider the case when some honest party
   outputs a wps-share, consisting of an element of $\F$. We note that every honest party participates with some input in the instance of $\HBA$
    at local time
  $2\Delta + 2 \TimeBCAST$. Hence, from the {\it $t_a$-almost-surely liveness} and {\it $t_a$-consistency} properties
   of $\HBA$ in the {\it asynchronous} network, 
    almost-surely, all honest parties eventually compute a common output during 
    the instance of $\HBA$. Now there are two possible cases:
    \begin{myitemize}
     \item[--] {\bf The output of $\HBA$ is 0}: 
     From the {\it $t_a$-validity} of $\HBA$ in the {\it asynchronous} network, it implies that 
     at least one honest party, say $P_h$, participated with input $0$ during the instance of $\HBA$.
     This further implies that at local time $2\Delta +2\TimeBCAST$, party $P_h$
      has accepted a $(\WCORE,\ESet,\FSet)$, which has been received by $P_h$
       from the broadcast of $\D$, through regular-mode.             
            Hence, by the {\it $t_a$-weak consistency} and {\it $t_a$-fallback consistency} of $\BCAST$ in 
            the {\it asynchronous} network, 
            all honest parties will eventually receive $(\WCORE,\ESet, \FSet)$ from the broadcast of $\D$.
            There will be at least $n - 2t_s - t_a > t_s$ {\it honest} parties in $\ESet$, whose univariate polynomials $q_i(x)$ are pair-wise consistent and hence
            from Lemma \ref{lemma:bivariateI} lie on a unique          
            $(t_s, t_s)$-degree symmetric bivariate polynomial, say $Q^{\star}(x,y)$. Similarly, the univariate polynomial $q_j(x)$ of every {\it honest}
            $P_j \in \FSet$ will be pair-wise consistent with the univariate polynomials $q_i(x)$ of all the {\it honest} parties in $\ESet$ and hence lie on 
            $Q^{\star}(x, y)$ as well. Let $q^{\star}(\cdot) \defined Q^{\star}(0, y)$.
            We claim that {\it every honest} $P_i$ will eventually get $Q^{\star}(x, \alpha_i)$. 
            This will imply that eventually every {\it honest} $P_i$ outputs the wps-share $s_i = Q^{\star}(0, \alpha_i) = q^{\star}(\alpha_i)$.
            To prove the claim, consider an arbitrary {\it honest} party $P_i$. There are three possible cases.
            \begin{myitemize}
                \item[--] {\it $P_i \in \WCORE$ and $P_i \in \FSet$}: In this case, $P_i$ already has $q_i(x)$, received from $\D$. And since $P_i \in \FSet$, the condition                
                 $q_i(x) = Q^{\star}(x, \alpha_i)$ holds.
                \item[--] {\it $P_i \in \WCORE$ and $P_i \not \in \FSet$}: In this case, $P_i$ already has $q_i(x)$, received from $\D$.
                Since $|\WCORE| \geq n-t_s$ and $|\FSet| \geq n-t_s$, $|\WCORE \cap \FSet| \geq n-2t_s > t_s + t_a$. 
                From the protocol steps, the polynomial $q_i(x)$ is pair-wise consistent with the polynomial $q_j(x)$ of at least $n - t_s$ parties $P_j \in \WCORE$
                (since $P_i$ has edges with at least $n - t_s$ parties $P_j$ within $\WCORE$). Now among these $n - t_s$ parties, at least $n - 2t_s$ parties will be from $\FSet$, of which 
                at least $n - 2t_s - t_a > t_s$ parties will be {\it honest}. Hence, $q_i(x)$ is pair-wise consistent with the $q_j(x)$ polynomials of at least $t_s + 1$ 
                {\it honest} parties $P_j \in \FSet$. Now since the $q_j(x)$ polynomial of all the {\it honest} parties in $\FSet$ lie on $Q^{\star}(x, y)$, it implies that
                $q_i(x) = Q^{\star}(x, \alpha_i)$ holds.
                \item[--] {\it $P_i \not \in \WCORE$}:
                In this case, there will be $n - t_s$ parties in $\FSet$, of which at most $t_a$ could be {\it corrupt}.
                Since $Q^{\star}(x, \alpha_i)$ is a $t_s$-degree polynomial, and $t_s < |\FSet| - 2t_a$, from Lemma \ref{lemma:OECProperties}
                  it follows that by applying the OEC procedure on the common points on the
                $Q^{\star}(x, \alpha_i)$ polynomial received from the parties in $\FSet$, party $P_i$ will eventually obtain 
                $Q^{\star}(x, \alpha_i)$.
            \end{myitemize}
  \item[--] \textbf{The output of $\HBA$ is 1}: Let $P_h$ be the {\it first honest} party, who outputs a wps-share. 
         This means that $P_h$ has received a pair, say $(\ESet', \FSet')$,
          from the broadcast of $\D$, such that $(\ESet', \FSet')$ constitutes an $\Star{t_a}$
         in $P_h$'s consistency graph. By the {\it $t_a$-weak consistency} and {\it $t_a$-fallback consistency} properties of $\BCAST$ in the {\it asynchronous} network,
          {\it all} honest parties eventually receive $(\ESet', \FSet')$
         from the broadcast of $\D$. Moreover, since the consistency graphs are constructed based on the 
         broadcast $\OK$ messages and since
         $(\ESet', \FSet')$ constitutes an $\Star{t_a}$  in $P_h$'s consistency graph, from the {\it $t_a$-weak validity, $t_a$-fallback validity, $t_a$-weak consistency and 
         $t_a$-fallback consistency}
         properties of $\BCAST$ in the {\it asynchronous} network, the pair 
         $(\ESet', \FSet')$ will eventually constitute an $\Star{t_a}$ in every honest party's consistency graph. 
         Since $|\ESet'| \geq n-2t_a > 2t_s + (t_s-t_a) > 2t_s$, it follows that $\ESet'$ has at least $ t_s + 1$ {\it honest} parties, whose 
         univariate polynomials $q_i(x)$ are pair-wise consistent and hence from Lemma \ref{lemma:bivariateI}, lie on a unique
         degree-$(t_s,t_s)$ symmetric bivariate polynomial, say $Q^{\star}(x, y)$. 
         Similarly, since the univariate polynomial $q_j(x)$ of every {\it honest} party $P_j$ in $\FSet'$ is pair-wise consistent with the univariate polynomials 
         $q_i(x)$ of the {\it honest} parties in $\ESet'$,
         it implies that the univariate polynomial $q_j(x)$ of all the {\it honest} parties in $\FSet'$ also lie on $Q^{\star}(x, y)$.
         Let $q^{\star}(\cdot) \defined Q^{\star}(0, y)$. We show that {\it every honest} $P_i$ eventually
         outputs $q^{\star}(\alpha_i)$ as its wps-share. For this it is enough to show that each honest $P_i$ eventually gets 
         $q_i(x) = Q^{\star}(x, \alpha_i)$,
         as $P_i$ outputs $q_i(0)$ as its wps-share, which will be the same as $q^{\star}(\alpha_i)$.
         Consider an arbitrary {\it honest} $P_i$. There are two possible cases. 
         \begin{myitemize}
             \item[--] \textit{$P_i \in \FSet'$}: In this case, $P_i$ already has $Q^{\star}(x, \alpha_i)$, received from $\D$.
             \item[--] \textit{$P_i \notin \FSet'$}: In this case, $\FSet'$ has at least $n - t_a > 3t_s$ parties, of which at most $t_a$ could be corrupt.            
             Since $Q^{\star}(x, \alpha_i)$ is a $t_s$-degree polynomial and $t_s < |\FSet'| - 2t_a$, from Lemma \ref{lemma:OECProperties}, 
              it follows that by applying the OEC procedure on the common points on the
                polynomial $Q^{\star}(x, \alpha_i)$ received from the parties in $\FSet'$, party $P_i$ will eventually obtain 
                $Q^{\star}(x, \alpha_i)$.
         \end{myitemize}
    \end{myitemize}
\end{proof}
\begin{lemma}
 \label{lemma:WPSCommunication}
 Protocol $\WPS$ incurs a communication of $\Order(n^4 \log{|\F|})$ bits from the honest parties
   and invokes $1$ instance of $\HBA$.		   
 \end{lemma}
 \begin{proof}
 In the protocol, $\D$ sends a $t_s$-degree univariate polynomial to every party. As part of the pair-wise consistency checks, each pair of parties exchange $2$
  field elements. In addition, an {\it honest} party may broadcast an 
   $\NOK$ message, corresponding to every other party. As part of the $\NOK$ message,
  the honest party also broadcasts the corresponding common point on its univariate polynomial. Each such common point can be represented by $\log{|\F|}$ bits. 
   The communication complexity now follows from the communication complexity of the protocol $\BCAST$ (see Theorem \ref{thm:DS}).
 \end{proof}
 
 We next discuss the modifications needed in the protocol $\WPS$, if the input for $\D$ consists of $L$ number of
  $t_s$-degree polynomials.
\paragraph{\bf $\WPS$ for $L$ Polynomials:}
  If $\D$ has $L$ polynomials as input in protocol $\WPS$,
   then it embeds them
    into $L$ random
   $(t_s, t_s)$-degree symmetric bivariate polynomials and distributes the univariate polynomials lying on these bivariate polynomials, 
   to the respective parties. The parties then perform the 
   pair-wise consistency tests, by exchanging their supposedly common points on the bivariate polynomials.
    However, $P_i$ broadcasts a {\it single} $\OK(i, j)$ message for $P_j$, if the pair-wise consistency test is positive for {\it all} the
   $L$ supposedly common values between $P_i$ and $P_j$. 
    On the other hand, if the test fails for {\it any} of the $L$ supposedly common values, then 
   $P_i$ broadcasts a {\it single} $\NOK$ message, corresponding to the {\it least indexed} common value for which the test fails. 
    Hence, instead of constructing $L$ consistency graphs, 
    a {\it single} consistency graph is 
   constructed by each party.  As a result, $\D$ finds a {\it single} $(\WCORE, \ESet, \FSet)$ triplet and broadcast it.
   Similarly, a {\it single} instance of $\HBA$ is used to decide whether any $(\WCORE, \ESet, \FSet)$ is accepted.
   Finally, if no $(\WCORE, \ESet, \FSet)$ triplet is found and broadcast, then $\D$ looks for a {\it single}
   $\Star{t_a}$ $(\ESet', \FSet')$ and broadcasts it.

   To void repetition,
   we skip the formal details of the modified protocol and the proof of its properties, 
 as stated in Theorem \ref{thm:WPS}.
\begin{theorem}
\label{thm:WPS}
Let $n > 3t_s + t_a$ and let $\D$ has 
 $L$ number of $t_s$-degree polynomials $q^{(1)}(\cdot), \ldots, q^{(L)}(\cdot)$ as input for $\WPS$, where $L \geq 1$. Moreover, let 
 $\TimeWPS = 2 \Delta +  2\TimeBCAST + \TimeHBA$.
Then protocol $\WPS$ achieves the following properties.
  \begin{myitemize}
   \item[--] If $\D$ is {\it honest} then the following hold.
       \begin{myitemize}
       \item[--] {\it $t_s$-correctness}: In a synchronous network, each (honest) $P_i$ outputs
         $\{q(\alpha_i)\}_{\ell = 1, \ldots, L}$ at time $\TimeWPS$.
       \item[--] {\it $t_a$-correctness}: In an asynchronous network, almost-surely, each (honest) $P_i$ eventually outputs
        $\{q(\alpha_i)\}_{\ell = 1, \ldots, L}$.       
       \item[--] {\it $t_s$-privacy}: Irrespective of the network type, the view of the adversary remains independent of 
       the polynomials $q^{(1)}(\cdot), \ldots, q^{(L)}(\cdot)$.
       \end{myitemize}
    \item[--]  If $\D$ is {\it corrupt}, then either no 
     honest party computes any output or there exist $L$ number
     of $t_s$-degree polynomials, say $\{{q^{\star}}(\cdot)\}_{\ell = 1, \ldots, L}$, such that the following hold.
         \begin{myitemize}
         \item[--] {\it $t_s$-Weak Commitment}: In a synchronous network, at least $t_s + 1$ honest parties $P_i$
          output wps-shares $\{{q^{\star}}(\alpha_i)\}_{\ell = 1, \ldots, L}$. Moreover, if any honest $P_j$ outputs wps-shares
          $s^{(1)}_j, \ldots, \allowbreak s^{(L)}_j \in \F$, then
          $s_j = {q^{\star}}(\alpha_j)$ holds for $\ell = 1, \ldots, L$.
         \item[--] {\it $t_a$-Strong Commitment}: In an asynchronous network, almost-surely, each (honest) $P_i$  
          eventually outputs $\{{q^{\star}}(\alpha_i)\}_{\ell = 1, \ldots, L}$ as wps-shares.
         \end{myitemize}
    \item[--] Irrespective of the network type, the protocol incurs a communication of $\Order(n^2 L \log{|\F|} + n^4 \log{|\F|})$ 
     bits from the honest parties and invokes $1$ instance of $\HBA$.
  \end{myitemize}
\end{theorem}
\subsection{The VSS Protocol}
Protocol $\WPS$ fails to serve as a VSS because if $\D$ is {\it corrupt} and the network is {\it synchronous}, 
 then the (honest) parties {\it outside} $\WCORE$ may not obtain their required shares, lying on $\D$'s committed polynomials.
   Protocol $\VSS$ (see Fig \ref{fig:VSS}) fixes this shortcoming. For ease of understanding, we present the protocol assuming 
  $\D$ has a single $t_s$-degree polynomial as input and later discuss the modifications needed when $\D$ has $L$ such polynomials.
  The protocol has two ``layers" of communication involved.
 The  first layer is similar to $\WPS$ and identifies whether
  the parties accepted some $(\WCORE, \ESet, \FSet)$ within a specified time-out, such that the polynomials of all
  honest parties in $\WCORE$ lie on a single $(t_s, t_s)$-degree symmetric bivariate polynomial, say $Q^{\star}(x, y)$.
    If some $(\WCORE, \ESet, \FSet)$ is accepted, then the second layer of communication (which is coupled with the first layer) enables 
    even the (honest) parties outside $\WCORE$ to get their corresponding polynomials lying on $Q^{\star}(x, y)$.
    
     In more detail, to perform the pair-wise
    consistency check of the polynomials received from $\D$, each $P_j$ upon receiving $q_j(x)$ from $\D$, 
    shares the polynomial $q_j(x)$ by invoking an instance of $\WPS$ as a dealer. Any party $P_i$ who computes a WPS-Share in this instance of $\WPS$ either broadcasts
    an $\OK$ or $\NOK$ message for $P_j$, depending on whether the WPS-share lies on the polynomial which $P_i$ has received from $\D$. 
    The rest of the steps for computing $(\WCORE, \ESet, \FSet)$ and accepting it remains the same. 
     If some $(\WCORE, \ESet, \FSet)$ is accepted,
     then    any $P_i$ {\it outside} $\WCORE$  computes its polynomial lying on $Q^{\star}(x, y)$ as follows: $P_i$ checks for a subset $\Support_i \subseteq \FSet$ 
     of $t_s + 1$ parties
    $P_j$, such that $P_i$ is able to compute its WPS-share in the $\WPS$ instance invoked by $P_j$ as a dealer.
    Such an $\Support_i$ is bound to exist as there are at least $t_s + 1$ {\it honest} parties in $\FSet$ who are always included in $\Support_i$.
    While the WPS-shares corresponding to the {\it honest} parties in $\Support_i$ will be the common points on $Q^{\star}(x, \alpha_i)$, the same holds
    even for {\it corrupt} parties in $\Support_i$. This is because in order to be included in $\FSet$,
     such parties are ``forced" to share
    polynomials lying on $Q^{\star}(x, y)$, in their respective instances
    of $\WPS$. Now using the WPS-shares corresponding to the parties in $\Support_i$, party $P_i$ will be able to compute  
    $Q^{\star}(x, \alpha_i)$ and hence, its share.
  
  \begin{protocolsplitbox}{$\VSS(\D, q(\cdot))$}{best-of-both-worlds VSS protocol for a single polynomial.}{fig:VSS}
\justify
  \begin{myitemize}
   	\item {\bf Phase I --- Sending Polynomials}: $\D$ on having the input $q(\cdot)$, chooses a random $(t_s, t_s)$-degree symmetric bivariate polynomial
	      $Q(x, y)$ such that $Q(0, y) = q(\cdot)$ and sends $q_i(x) = Q(x, \alpha_i)$ to each party $P_i \in \Partyset$. 
	    \item {\bf Phase II --- Exchanging Common Values}: Each $P_i \in \Partyset$, upon receiving a
	    $t_s$-degree polynomial $q_i(x)$ from $\D$, {\color{red} waits till the current local time becomes a multiple of $\Delta$} and then does the following.
        \begin{myitemize}
            \item[--] Act as a dealer and invoke an instance $\WPS^{(i)}$ of $\WPS$ with input $q_i(x)$. 
            \item[--] For $j = 1, \ldots, n$, participate in the instance $\WPS^{(j)}$, if invoked by $P_j$ as a dealer, {\color{red} and wait for time $\TimeWPS$}.
        \end{myitemize}
    \item {\bf Phase III  --- Publicly Declaring the Results of Pair-Wise Consistency Test}: Each $P_i \in \Partyset$ {\color{red} waits till the local time becomes a multiple of $\Delta$} 
    and then does the following.
        \begin{myitemize}
           \item[--] If a WPS-share $q_{ji}$ is computed during the instance $\WPS^{(j)}$ and $q_i(x)$ has been received from $\D$, then:
              \begin{myitemize}
               \item[--]  Broadcast $\OK(i, j)$, if $q_{ji} = q_i(\alpha_j)$ holds.
               \item[--] Broadcast $\NOK(i, j, q_i(\alpha_j))$, if $q_{ji} \neq q_i(\alpha_j)$ holds.
            \end{myitemize}            
        \end{myitemize}      
     \item {\bf Local Computation --- Constructing Consistency Graph}:  Each $P_i \in \Partyset$ does the following.
          \begin{myitemize}
           \item[--] Construct a {\it consistency graph} $G_i$ over $\Partyset$, where the edge $(P_j, P_k)$ is included in 
           $G_i$, if $\OK(j, k)$ and $\OK(k, j)$ is received from the broadcast of $P_j$ and $P_k$ respectively, either through the regular-mode or fall-back mode.	
          \end{myitemize}
         \item {\bf Phase IV --- Constructing $\Star{t_s}$}: $\D$ does the following in its consistency graph
     $G_\D$ at  time $2\Delta + \TimeBCAST$.   
      \begin{myitemize}
          \item[--] Remove edges incident with $P_i$, if $\NOK(i, j,q_{ij})$ is received from the broadcast of $P_i$ through regular-mode
           and $q_{ij} \neq Q(\alpha_j,\alpha_i)$.
          \item[--] Set $\WCORE = \{P_i: \mathsf{deg}(P_i) \geq n-t_s\}$, where $\mathsf{deg}(P_i)$ denotes the degree of $P_i$ in $G_\D$.
           \item[--] Remove $P_i$ from $\WCORE$, if $P_i$ is {\it not} incident with at least $n - t_s$ parties in $\WCORE$. Repeat this step till no more parties can be 
           removed from $\WCORE$.
          \item[--] Run algorithm $\StarAlgo$ on $G_\D[\WCORE]$, where $G_\D[\WCORE]$ denotes the subgraph of $G_\D$ induced by the vertices in $\WCORE$.
           If an $\Star{t_s}$, say $(\ESet, \FSet)$, is obtained, then broadcast $(\WCORE, \ESet, \FSet)$.
      \end{myitemize}
   \item {\bf Local Computation --- Verifying and Accepting $(\WCORE,\ESet,\FSet)$}: Each $P_i \in \Partyset$ does the following at time $\Delta + \TimeWPS + 2\TimeBCAST$.
          \begin{myitemize}
    	    \item[--] If a $(\WCORE, \ESet, \FSet)$ is received from $\D$'s broadcast through regular-mode, then {\it accept} it if following {\it were true} at time 
	    $\Delta + \TimeWPS + \TimeBCAST$:
    	    \begin{myitemize}
                \item[--] There exist no $P_j, P_k \in \WCORE$, such that 
                $\NOK(j, k, q_{jk})$ and $\NOK(k, j, q_{kj})$ messages {\it were} received from the broadcast of $P_j$ and $P_k$ respectively through regular-mode, 
                where $q_{jk} \neq q_{kj}$.
                \item[--] In the consistency graph $G_i$,
                  $\mathsf{deg}(P_j) \ge n-t_s$  for all $P_j \in \WCORE$.
                \item[--] In the consistency graph $G_i$, every $P_j \in \WCORE$ has 
                  edges with at least $n - t_s$ parties from $\WCORE$. 
                \item[--] $(\ESet, \FSet)$ {\it was} an $\Star{t_s}$ in the induced graph $G_i[\WCORE]$.
                \item[--] For every $P_j, P_k \in \WCORE$ where the edge $(P_j, P_k)$ is present in $G_i$, the $\OK(j, k)$ and $\OK(k, j)$ messages {\it were} received from the
                 broadcast of $P_j$ and $P_k$ respectively, through regular-mode.
            \end{myitemize} 
        \end{myitemize}
    \item {\bf Phase V  --- Deciding Whether to Go for $\Star{t_a}$}: At time $\Delta + \TimeWPS + 2\TimeBCAST$, each party $P_i$ participates in an instance of $\HBA$ with input $b_i = 0$ if a $(\WCORE, \ESet, \FSet)$ is accepted, else with input $b_i = 1$ and {\color{red} waits for time $\TimeHBA$}.
    \item \textbf{Local Computation --- Computing VSS-Share Through $(\WCORE, \ESet, \FSet)$:} If $0$ is the output during
    the instance of $\HBA$, then each $P_i$ does the following.
		  \begin{myitemize}
			  \item[--] If a $(\WCORE,\ESet,\FSet)$ is not yet received, then
			  wait till it is received from $\D$'s broadcast through fall-back mode.  
			  \item[--] If $P_i \in \WCORE$, then output $q_i(0)$. 
              \item[--] Else, initialize $\Support_i$ to $\emptyset $. Include $P_j \in \FSet$ to $\Support_i$ if a wps-share $q_{ji}$ is computed during 
              $\WPS^{(j)}$. Wait till $|\Support_i| \geq t_s + 1 $. Then interpolate $\{(\alpha_j, q_{ji}) \}_{P_j \in \Support_i}$  to get a $t_s$-degree 
              polynomial, say $q_i(x)$, and output $q_i(0)$.
              \end{myitemize}
    \item {\bf Phase VI  --- Broadcasting $\Star{t_a}$}: If the output during the instance of $\HBA$ is $1$, 
    then $\D$ runs $\StarAlgo$ after every update in its consistency graph $G_\D$ and broadcasts $(\ESet', \FSet')$, if it finds an 
    $\Star{t_a}$ $(\ESet', \FSet')$.
    \item \textbf{Local Computation --- Computing VSS-Share Through $\Star{t_a}$}: If the output during the instance of $\HBA$ is
     $1$, then each $P_i$ does the following. 
        \begin{myitemize}
	       \item[--] Participate in any instance of $\BCAST$ invoked by $\D$ for broadcasting an $\Star{t_a}$ {\it only} after time 
	       $\Delta + \TimeWPS + 2\TimeBCAST + \TimeHBA$.		       
	       Wait till some $(\ESet', \FSet')$ is obtained from $\D$'s broadcast (through any mode), which constitutes an $\Star{t_a}$ in $G_i$.
           \item[--] If $P_i \in \FSet'$, then output $q_i(0)$. 
           Else, include $P_j \in \FSet'$ to $\Support_i$ (initialized to $\emptyset$) if a wps-share $q_{ji}$ is computed in
              $\WPS^{(j)}$. Wait till $|\Support_i| \geq t_s + 1 $. Then interpolate $\{(\alpha_j, q_{ji}) \}_{P_j \in \Support_i}$  to get a $t_s$-degree polynomial $q_i(x) $ and output $q_i(0)$.
                \end{myitemize}   
 \end{myitemize}
\end{protocolsplitbox}

We next proceed to prove the properties of the protocol $\VSS$. We first start by showing that if $\D$ is {\it honest}, then the view of the adversary remains independent of
 dealer's polynomial.
\begin{lemma}[{\it $t_s$-Privacy}]
 \label{lemma:VSSPrivacy}
In protocol $\VSS$, if $\D$ is honest, then irrespective of the network type, the view of the adversary remains independent of $q(\cdot)$.
 \end{lemma}
 \begin{proof}
 Let $\D$ be {\it honest}. We consider the worst case scenario when adversary controls up to $t_s$ parties. 
 We claim that throughout the protocol, the adversary learns at most $t_s$ univariate polynomials lying on $Q(x, y)$.
  Since $Q(x, y)$ is a random $(t_s, t_s)$-degree
  symmetric-bivariate polynomial, it then follows from Lemma \ref{lemma:bivariateII}, 
  that the view of the adversary will be independent of $q(\cdot)$. We next proceed to prove the claim.

 Corresponding to every {\it corrupt} $P_i$, the adversary learns $Q(x, \alpha_i)$. 
   Corresponding to every {\it honest} $P_i$, the adversary learns $t_s$ number of $q_i(\alpha_j)$ values through pair-wise consistency tests,
   as these values are computed as wps-shares, during the instance $\WPS^{(i)}$. 
    However, these values are already included in the view of the adversary (through the univariate polynomials under adversary's control).
   Additionally, from the {\it $t_s$-privacy} property of $\WPS$, the view of the adversary remains independent of $q_i(x)$ during $\WPS^{(i)}$, if $P_i$ is {\it honest}.
  Hence no additional information about the polynomials of the honest parties is revealed during the pair-wise consistency checks.
   Furthermore, no {\it honest} $P_i$ ever broadcasts $\NOK(i, j, q_{ij})$ corresponding to
    any {\it honest} $P_j $, since the pair-wise consistency check will always pass for every pair of {\it honest} parties.
 \end{proof}
 
 We next prove the correctness property in a {\it synchronous} network. 
 \begin{lemma}[{\it $t_s$-Correctness}]
 \label{lemma:VSSSynCorrectness}
In protocol $\VSS$, if $\D$ is honest and network is synchronous, then each honest $P_i$ outputs $q(\alpha_i)$ within time
 $\TimeVSS = \Delta + \TimeWPS +   2\TimeBCAST + \TimeHBA$.
 \end{lemma}
 \begin{proof}
 Let $\D$ be {\it honest} and network be {\it synchronous} with up to $t_s$ corruptions. During phase I, 
  all honest parties receive $q_i(x) = Q(x,\alpha_i)$ from $\D$ within time $\Delta$. Consequently during phase II,
   each {\it honest} $P_i$ invokes the instance $\WPS^{(i)}$ with input 
  $q_i(x)$. From the {\it $t_s$-correctness} of $\WPS$ in the {\it synchronous} network, 
  corresponding to each {\it honest} $P_j$, every {\it honest} $P_i$ computes the wps-share $q_{ji} = q_j(\alpha_i)$, at time $\Delta + \TimeWPS$.
  Consequently, during phase III,
   every {\it honest} party broadcasts an $\OK$ message for every other
   {\it honest} party, since $q_{ji} = q_i(\alpha_j)$ holds, for every pair of honest parties $P_i, P_j$.
       From the {\it $t_s$-validity} property of $\BCAST$ in the {\it synchronous} network, these $\OK$ messages are received by 
   every honest party through regular-mode at time $\Delta + \TimeWPS + \TimeBCAST$. Hence, there will be an edge between every pair of
  {\it honest} parties in the consistency graph of every honest party. 
  Moreover, if $\D$ receives an {\it incorrect} $\NOK(i, j, q_{ij})$ message from the broadcast of any {\it corrupt}
  $P_j$ through regular-mode at time $\Delta + \TimeWPS + \TimeBCAST$ where $q_{ij} \neq Q(\alpha_j, \alpha_i)$, then $\D$ removes all the edges incident with
  $P_i$ in $\D$'s consistency graph $G_\D$. 
  $\D$ then computes the set $\WCORE$ and all {\it honest} parties will be present in $\WCORE$. Moreover, the honest parties will
  form a clique of size at least $n - t_s$ in the subgraph $G_\D[\WCORE]$ at time $\Delta + \TimeWPS + \TimeBCAST$.
    Hence, $\D$ will find an $\Star{t_s}$ $(\ESet, \FSet)$ in $G_\D[\WCORE]$ and broadcast  $(\WCORE, \ESet, \FSet)$ during phase IV. 
  By the {\it $t_s$-validity} of $\BCAST$ in the {\it synchronous} network, 
  all honest parties will receive $(\WCORE, \ESet, \FSet)$ through regular-mode at time $\Delta + \TimeWPS + 2\TimeBCAST$.
  Moreover, all honest parties will accept {\it accept}  $(\WCORE,\ESet, \FSet)$ and participate with input $0$ in the instance of $\HBA$.
   By the {\it $t_s$-validity} and {\it $t_s$-guaranteed liveness}  of $\HBA$ in the {\it synchronous} network,
    the output of the $\HBA$ instance will be $0$ for every honest party at time
    $\Delta + \TimeWPS + 2\TimeBCAST + \TimeHBA$. 
    Now consider an arbitrary {\it honest} party $P_i$. Since $P_i \in \WCORE$, $P_i$ outputs $s_i = q_i(0) = Q(0, \alpha_i) = q(\alpha_i)$.
 \end{proof}
 
 We next prove the correctness property in the asynchronous network.
\begin{lemma}[{\it $t_a$-Correctness}]
\label{lemma:VSSAsynCorrectness}
In protocol $\VSS$, if $\D$ is honest and network is asynchronous, then almost-surely,
 each honest $P_i$ eventually outputs
  $q(\alpha_i)$.
\end{lemma}
\begin{proof}
Let $\D$ be {\it honest} and network be {\it asynchronous} with up to $t_a$ corruptions. We first note that every {\it honest} $P_i$ eventually broadcasts $\OK(i, j)$ message, corresponding to every
 {\it honest} $P_j$. This is because both $P_i$ and $P_j$ eventually receive $q_i(x) = Q(x, \alpha_i) $
  and $q_j(x) = Q(x, \alpha_j)$ respectively from $\D$. Moreover, $P_j$
  participates with input $q_j(\cdot)$ during $\WPS^{(j)}$. And from the {\it $t_a$-correctness} of $\WPS$ in the
   {\it asynchronous} network, 
  party $P_i$ eventually computes the wps-share $q_{ji} = q_j(\alpha_i)$ during
 $\WPS^{(j)}$. Moreover, $q_{ji} = q_{ij}$ holds. 
  Note that every honest party participates with some input in the instance of $\HBA$ at local time
  $\Delta + \TimeWPS + 2 \TimeBCAST$. Hence, from the {\it $t_a$-almost-surely liveness} and
  {\it $t_a$-consistency} properties of $\HBA$ in the asynchronous network,
  almost-surely, all honest parties eventually compute a common output during the instance of $\HBA$. 
  Now there are two possible cases:
    \begin{myitemize}
     \item[--] {\bf The output of $\HBA$ is $0$}: 
     From the {\it $t_a$-validity} of $\HBA$ in the {\it asynchronous} network, this means that at least one {\it honest} party, say
     $P_h$, participated with input $0$ during the instance of $\HBA$. This implies that 
      $P_h$
            has received $(\WCORE, \ESet,\FSet)$ from the broadcast of $\D$ through regular-mode and accepted it.             
            Hence, by the {\it $t_a$-weak validity} and {\it $t_a$-fallback validity} of 
            $\BCAST$ in the {\it asynchronous} network, 
            all honest parties will eventually receive $(\WCORE, \ESet,\FSet)$ from the broadcast of $\D$.
            We claim that {\it every honest} $P_i$ will eventually get $Q(x, \alpha_i)$.
             This will imply that eventually every {\it honest} $P_i$ outputs  $s_i = Q(0, \alpha_i) = q(\alpha_i) $.
            To prove the claim, consider an arbitrary {\it honest} party $P_i$. There are two possible cases.
            \begin{myitemize}
                \item[--] {\it $P_i \in \WCORE$}: In this case, $P_i$ already has $Q(x, \alpha_i) $, received from $\D$.
                \item[--] {\it $P_i \notin \WCORE$}: In this case, we first note that there will be at least $t_s + 1$ parties, who are eventually included in $\Support_i $. 
                This follows from the fact that there are at least
    $t_s + 1$ {\it honest} parties $P_j$ in $\FSet$. And corresponding to every {\it honest} $P_j \in \FSet$, party $P_i$ will eventually compute the wps-share
     $q_{ji}$ in the instance $\WPS^{(j)}$, which follows from the {\it $t_a$-correctness} of $\WPS$ in the {\it asynchronous} network.
    We next claim that corresponding to {\it every} $P_j \in \Support_i$, the value $q_{ji}$ computed by $P_i$ is the {\it same} as
    $Q(\alpha_j, \alpha_i)$. 
    
    The claim is obviously true for every {\it honest} $P_j \in \Support_i$, so consider a {\it corrupt} $P_j \in \Support_i$. We first note that
     the input polynomial $q_j(x)$ of $P_j$ during $\WPS^{(j)}$
      is the {\it same} as $Q(x, \alpha_j)$. This is because $P_j \in \WCORE$, since $\FSet \subseteq \WCORE$.
     And hence $P_j$ has edges with at least $n - t_s$ parties in $\WCORE$
     and hence with at least $n - t_s - t_a > t_s$ {\it honest}
     parties from $\WCORE$
    in $G_\D[\WCORE]$. Let $\Honest$ be the set of {\it honest} parties in $\WCORE$ with which $P_j$ has edges in $G_\D[\WCORE]$.
     This implies that {\it every} $P_k \in \Honest$ has broadcast $\OK(k, j)$ message after verifying that $q_{jk} = q_k(\alpha_j)$ holds, 
     where
     the polynomial $q_k(x)$ held by $P_k$ is the same as $Q(x, \alpha_k)$ and where
          the WPS-share $q_{jk}$ computed by $P_k$ during $\WPS^{(j)}$ is the {\it same} as $q_j(\alpha_k)$;
          the last property follows from the {\it $t_a$-strong commitment} of $\WPS$ in the {\it synchronous} network.
     Since $|\Honest| > t_s$, it implies that at least $t_s + 1$ {\it honest} parties $P_k$ have verified that $q_j(\alpha_k) = Q(\alpha_j, \alpha_k)$
     holds. This further implies that $q_j(x) = Q(x, \alpha_j)$, since two different $t_s$-degree polynomials can have at most $t_s$ common values.
    Since $P_i$ has computed the wps-share $q_{ji}$ during $\WPS^{(j)}$, from the {\it $t_a$-strong commitment} of $\WPS$ in {\it synchronous} network,  it follows that
     $q_{ji} = q_j(\alpha_i) = Q(\alpha_i, \alpha_j) = Q(\alpha_j, \alpha_i)$. The last equality follows since each $Q(x, y)$ is a symmetric bivariate polynomial.
            \end{myitemize}
      \item[--] {\bf The output of $\HBA$ is $1$}: As mentioned earlier, since $\D$ is {\it honest}, every pair of honest parties eventually
   broadcast $\OK$ messages corresponding to each other, as the pair-wise consistency check between them will be eventually positive. 
   From the {\it $t_a$-weak validity} and {\it $t_a$-fallback validity} of $\BCAST$ in the {\it asynchronous} network, these messages are 
   eventually delivered to every honest party.
   Also from the {\it $t_a$-weak consistency} and {\it $t_a$-fallback consistency} of $\BCAST$ in the {\it asynchronous} network, 
   any $\OK$ message which is received by $\D$, will be eventually received by every other honest party as well.
  As there will be at least $n - t_a$ honest parties, a clique of size at least $n - t_a$ will eventually form in the consistency graph of every honest party. 
   Hence $\D$ will eventually find an $\Star{t_a}$, say $(\ESet',\FSet')$, in its consistency graph and broadcast it. From 
   the {\it $t_a$-weak validity} and {\it $t_a$-fallback validity} of $\BCAST$ in the {\it asynchronous} network, 
    $(\ESet',\FSet')$ will be eventually received by every honest party. 
   Moreover, $(\ESet',\FSet')$ will be eventually an $\Star{t_a}$ in every honest party's consistency graph.
       We now claim that {\it every honest} $P_i$ will eventually get $Q(x, \alpha_i) $.
        This will imply that eventually every {\it honest} $P_i$ outputs $s_i = Q(0, \alpha_i) = q(\alpha_i) $.
            To prove the claim, consider an arbitrary {\it honest} party $P_i$. There are two possible cases.
            \begin{myitemize}
                \item[--] $P_i \in \FSet'$: In this case, $P_i$ already has $Q(x, \alpha_i) $, received from $\D$.
                \item[--] $P_i \notin \FSet'$: In this case, we note that there will be will be at least $t_s + 1$ parties, who are eventually included in $\Support_i $,
                 such that
                corresponding to {\it every} $P_j \in \Support_i$, the value $q_{ji}$ computed by $P_i$ during $\WPS^{(j)}$ is the same as 
                $Q(\alpha_j, \alpha_i)$. The proof for this will be similar as for the case when $P_i \notin \WCORE$ and the output of $\HBA$ is $0$ and so we skip the proof.
                            \end{myitemize}
   \end{myitemize}
\end{proof}

Before  we proceed to prove the {\it strong commitment} property in the {\it synchronous} network, we prove a helping lemma.
\begin{lemma}
\label{lemma:VSS}
Let $\D$ be corrupt and network be synchronous.  If any honest party receives a
 $(\WCORE, \ESet, \FSet)$ from the broadcast of $\D$ through regular-mode and accepts $(\WCORE, \ESet, \FSet)$ at time $\Delta + \TimeWPS + 2\TimeBCAST$, 
 then all the following hold:
 \begin{myitemize}
\item[--] All honest parties in $\WCORE$ have received their respective $t_s$-degree univariate polynomials from $\D$ within time $\Delta$.
\item[--] The univariate polynomials $q_i(x)$ of all honest parties $P_i$ in $\WCORE$ lie on a unique
 $(t_s, t_s)$-degree symmetric bivariate polynomial, say $Q^{\star}(x,y)$. 
\item[--] At time $\Delta + \TimeWPS + 2\TimeBCAST$, every honest party accepts $(\WCORE, \ESet, \FSet)$.
\end{myitemize}
\end{lemma}
\begin{proof}
Let $\D$ be {\it corrupt} and network be {\it synchronous} with up to $t_s$ corruptions. As per the lemma condition, let $P_h$ be an {\it honest} party, who 
 receives some $(\WCORE, \ESet, \FSet)$ from the broadcast of $\D$ through regular-mode and accepts
  it at time $\Delta + \TimeWPS + 2\TimeBCAST$. From the protocol steps, it then follows that the following must be
  true for $P_h$ at time $\Delta + \TimeWPS + \TimeBCAST$:
 	     \begin{myitemize}
                        \item[--] There does not exist any 
                        $P_j, P_k \in \WCORE$, 
                        such that $\NOK(j, k, q_{jk})$ and $\NOK(k, j,q_{kj})$ messages {\it were} received from the broadcast of 
              $P_j$ and $P_k$ respectively through regular-mode, such that $q_{jk} \neq q_{kj}$.
                   \item[--] In $P_h$'s consistency graph $G_h$,  $\mathsf{deg}(P_j) \ge n-t_s$  for all $P_j \in \WCORE$ and $P_j$ has edges with at least
                   $n - t_s$ parties from $\WCORE$. 
                   \item[--] $(\ESet, \FSet)$ {\it was} an $\Star{t_s}$ in the induced subgraph $G_h[\WCORE]$, such that
                   for every $P_j, P_k \in \WCORE$ where the edge $(P_j, P_k)$ is present in $G_h$, the $\OK(j, k)$ and $\OK(k, j)$ messages {\it were} received from the broadcast of 
              $P_j$ and $P_k$ respectively through regular-mode.
                    \end{myitemize}
 We prove the first part of the lemma through a contradiction. So let $P_j \in \WCORE$ be an {\it honest} party, who receives its
  $t_s$-degree univariate polynomial from $\D$, say $q_j(x)$, at time $\Delta + \delta$, where
  $\delta > 0$. Moreover, let $P_k \in \WCORE$ be an {\it honest} party such that $P_j$ has an edge with $P_k$
 (note that $P_j$ has edges with at least $n - 2t_s > t_s + t_a$  {\it honest} parties in $\WCORE$).
 As stated above, at time $\Delta + \TimeWPS + \TimeBCAST$, party $P_h$ has received the message $\OK(k, j)$ from the broadcast of $P_k$ through regular-mode.
 From the protocol steps, $P_j$ waits till its local time becomes a multiple of $\Delta$, before it participates with input $q_j(\cdot)$ in the instance
 $\WPS^{(j)}$.  Hence, $P_j$ must have invoked $\WPS^{(j)}$  at time $c \cdot \Delta$, where $c \geq 2$. 
   Since the network is {\it synchronous}, from the {\it $t_s$-correctness} of $\WPS$ in the {\it synchronous} network, party $P_k$ will compute its wps-share $q_{jk}$
   during $\WPS^{(j)}$ at time $c \cdot \Delta + \TimeWPS$. Hence the result of the pair-wise consistency test with $P_j$ will be available to $P_k$ at time
   $c \cdot \Delta + \TimeWPS$. As a result, $P_k$ starts broadcasting $\OK(k, j)$ message only at time $c \cdot \Delta + \TimeWPS$. 
   Since $P_k$ is {\it honest}, from the {\it $t_s$-validity} property of $\BCAST$ in the {\it synchronous} network,
 it will take {\it exactly} $\TimeBCAST$ time for the message $\OK(k, j)$ to be received through regular-mode, once it is broadcast. 
 This implies that $P_h$ will receive the message $\OK(k, j)$
 at time $c \cdot \Delta + \TimeWPS + \TimeBCAST$, where $c \geq 2$. 
  However, this is a contradiction, since the $\OK(k, j)$ message has been received by $P_h$ at time $\Delta + \TimeWPS + \TimeBCAST$.

  To prove the second part of the lemma, we will show that 
   the univariate polynomials $q_i(x)$ of all the {\it honest} parties $P_i \in \WCORE$ are pair-wise consistent. 
   Since there will be at least $n - 2t_s > t_s$ {\it honest} parties in $\WCORE$, 
             from Lemma \ref{lemma:bivariateII} this will imply that all these polynomials
               lie on a unique $(t_s, t_s)$-degree symmetric bivariate polynomial, say $Q^{\star}(x, y)$. 
  So consider an arbitrary pair of {\it honest} parties $P_j, P_k \in \WCORE$. From the first part of the claim, both $P_j$ and $P_k$ must have received their respective univariate polynomials
  $q_j(x)$ and $q_k(x)$ by time $\Delta $.
    Moreover, from the {\it $t_s$-correctness} property of $\WPS$ in the {\it synchronous} network, $P_j$ and $P_k$ will compute the wps-shares
  $q_{kj} = q_k(\alpha_j)$ and $q_{jk} = q_j(\alpha_k)$ at time $\Delta + \TimeWPS$ during  $\WPS^{(k)}$ and $\WPS^{(j)}$ 
  respectively.  Since $P_j$ and $P_k$ are honest, 
  if $q_{kj} \neq q_{jk}$, they would broadcast $\NOK(j, k, q_{jk})$ and $\NOK(k, j, q_{kj})$ messages respectively at
    time $\Delta + \TimeWPS$. From the {\it $t_s$-validity} property of $\BCAST$ in the {\it synchronous} network, $P_h$ will receive these messages through regular-mode at
     time $\Delta + \TimeWPS + \TimeBCAST$. Consequently, $P_h$ {\it will not} accept $(\WCORE, \ESet, \FSet)$, which is a contradiction.
  
  To prove the third part of the lemma, we note that since $P_h$ receives
   $(\WCORE, \ESet, \FSet)$ from the broadcast of $\D$ through regular-mode  at time $\Delta + \TimeWPS + 2\TimeBCAST$, it implies
  that $\D$ must have started broadcasting  $(\WCORE, \ESet, \FSet)$ latest at time $\Delta + \TimeWPS + \TimeBCAST$. This is because
   it takes $\TimeBCAST$ time for the regular-mode of $\BCAST$ to generate an output. From the
  {\it $t_s$-consistency} property of $\BCAST$ in the {\it synchronous} network, it follows that every honest party will also receive $(\WCORE, \ESet, \FSet)$
   from the broadcast of $\D$ through regular-mode
  at time $\Delta + \TimeWPS + 2\TimeBCAST$. Similarly, since at time $\Delta + \TimeWPS + \TimeBCAST$, party $P_h$ has received the $\OK(j, k)$ and $\OK(k, j)$ messages
  through regular-mode
   from the broadcast of 
  every $P_j, P_k \in \WCORE$ where $(P_j, P_k)$ is an edge in $P_h$'s consistency graph, it follows that these messages started getting broadcast latest
   at time $\Delta + \TimeWPS$.
  From the {\it $t_s$-validity}
  and {\it $t_s$-consistency} properties of $\BCAST$ in the {\it synchronous} network, it follows that every honest party receives these broadcast messages through 
   regular-mode at time
  $\Delta + \TimeWPS + \TimeBCAST$. Hence $(\ESet, \FSet)$ will constitute an $\Star{t_s}$ in the induced subgraph $G_i[\WCORE]$ of every honest party $P_i$'s 
   consistency-graph at time $\Delta + \TimeWPS + \TimeBCAST$ 
  and consequently, every honest party accepts $(\WCORE, \ESet, \FSet)$.  
\end{proof}

We next prove the strong commitment property in the synchronous network.
\begin{lemma}[{\it $t_s$-Strong Commitment}]
 \label{lemma:VSSSynStrongCommitment}
In protocol $\VSS$, if $\D$ is corrupt and network is synchronous, then either no honest party computes any output or there exist a
  $t_s$-degree polynomial, say $q^{\star}(\cdot)$, such that each honest $P_i$ eventually outputs $q^{\star}(\alpha_i)$, where the following hold.
     \begin{myitemize}
     \item[--] If any honest $P_i$ computes its output at time $\TimeVSS = \Delta + \TimeWPS +  2\TimeBCAST + \TimeHBA$,
      then every honest party obtains its output at time $\TimeVSS$.
     \item[--] If any honest $P_i$ computes its output at time $T$ where $T > \TimeVSS$, then every honest party computes its output by time
     $T + 2\Delta$.
     \end{myitemize}
  \end{lemma}
 \begin{proof}
  Let $\D$ be {\it corrupt} and network be {\it synchronous} with up to $t_s$ corruptions.
   If no honest party computes any output, then the lemma holds trivially. So consider the case when some {\it honest} party
   computes an output. 
   Now, there are two possible cases.
    \begin{myitemize}
         \item[--] \textbf{At least one honest party, say $P_h$, has received some
       $(\WCORE, \ESet, \FSet)$ from the broadcast of $\D$ through regular-mode and accepted $(\WCORE, \ESet, \FSet)$ at time $\Delta + \TimeWPS + 2\TimeBCAST$}:
       In this case, from Lemma \ref{lemma:VSS}, the polynomials $q_i(x)$ of all {\it honest} parties in $\WCORE$ are guaranteed to lie on a unique
        $(t_s, t_s)$-degree symmetric bivariate polynomial, say
       $Q^{\star}(x, y)$.
        As per the protocol steps, $P_h$ has also verified that $\FSet \subseteq \WCORE$, by checking that $(\ESet, \FSet)$ constitutes an $\Star{t_s}$ in
       the induced subgraph $G_h[\WCORE]$.
       Hence the polynomials $q_i(x)$ 
       of all {\it honest} parties in $\FSet$ also lie on $Q^{\star}(x, y) $. Moreover, from Lemma \ref{lemma:VSS}, all honest parties
       accept $(\WCORE, \ESet, \FSet)$ at time $\Delta + \TimeWPS +  2\TimeBCAST$. 
       Hence, every honest party participates in the instance of $\HBA$ with input $0$. Consequently, by the {\it $t_s$-validity} and {\it $t_s$-guaranteed liveness} properties of 
        $\HBA$ in the {\it synchronous} network, 
       all honest parties compute the output $0$ during the instance of $\HBA$ at time $\TimeVSS = \Delta + \TimeWPS +  2\TimeBCAST +  \TimeHBA$.
       Let $q^{\star}(\cdot) = Q^{\star}(0, y)$ and consider an arbitrary {\it honest}
       party $P_i$. We wish to show that $P_i$ has $q_i(x) = Q^{\star}(x, \alpha_i)$ at time $\TimeVSS$, which will imply that $P_i$ outputs
       $s_i = q_i(0)$ at time $\TimeVSS$, which will be the same as $q^{\star}(\alpha_i)$. 
        For this, we consider the following two possible cases.
        \begin{myitemize}
            \item[--] {\it $P_i \in \WCORE$}: In this case, $P_i$ has already received $q_i(x)$ from $\D$ within time $\Delta$. 
            This follows from Lemma \ref{lemma:VSS}.
            \item[--] {\it $P_i \notin \WCORE$}: In this case, we claim that at time $\TimeVSS$,
             there will be will be at least $t_s + 1$ parties from $\FSet$, who are included in $\Support_i $, such that
                corresponding to {\it every} $P_j \in \Support_i$, party $P_i$ will have the value $q_{ji}$, which will be the same as $Q^{\star}(\alpha_j, \alpha_i) $. 
              Namely, there are at least $t_s + 1$ {\it honest} parties in $\FSet$, who will be included in $\Support_i$
             and the claim will be trivially true for those parties, due to the {\it $t_s$-correctness} property of $\WPS$ in the synchronous network.
              On the other hand, if any {\it corrupt} $P_j \in \FSet$ is included in $\Support_i$, then 
              the input polynomial of $P_j$ during $\WPS^{(j)}$ will be pair-wise consistent with the polynomials of at least
              $t_s + 1$ {\it honest} parties in $\WCORE$ and hence will be the same as $Q^{\star}(x, \alpha_j)$. Moreover, 
               from the {\it $t_s$-weak commitment} of
              $\WPS$ in the {\it synchronous} network, the WPS-share $q_{ji}$ computed by $P_i$ during $\WPS^{(j)}$
              will be the same as $Q^{\star}(\alpha_i, \alpha_j)$, which will be the same as $Q^{\star}(\alpha_j, \alpha_i)$, since
              $Q^{\star}(x, y) $ is a symmetric bivariate polynomial.              
              Hence, $P_i$ will interpolate $q_i(x) $. 
        \end{myitemize}
         \item[--] \textbf{No honest party has received any
       $(\WCORE, \ESet, \FSet)$ from the broadcast of $\D$ through regular-mode and accepted $(\WCORE, \ESet, \FSet)$ at time $\Delta + \TimeWPS + 2\TimeBCAST$}: 
       This implies that all honest parties participate in the instance of
        $\HBA$ with input $1$. Hence, by the {\it $t_s$-validity} and {\it $t_s$-guaranteed liveness} of $\HBA$ in 
       the {\it synchronous} network,
         all honest parties obtain the output $1$ during the instance of $\HBA$ at time $\Delta + \TimeWPS + 2\TimeBCAST + \TimeHBA$. 
          Let $P_h$ be the {\it first honest} party, who computes an output. 
         This means that $P_h$ has received a pair $(\ESet', \FSet')$ from the broadcast of $\D$, such that $(\ESet', \FSet')$ constitutes an $\Star{t_a}$
         in $P_h$'s consistency graph. Let $T$ be the time when $(\ESet', \FSet')$ constitutes an $\Star{t_a}$ in $P_h$'s consistency graph. 
         This implies that at time $T$, party $P_h$ has $(\ESet', \FSet')$ from $\D$'s broadcast and also all the $\OK(\star, \star)$ messages, 
         from the broadcast of respective parties in $\ESet'$ and $\FSet'$.
        From the protocol steps, $T > \TimeVSS$, since the honest parties participate in the instance of $\BCAST$ through which 
         $\D$ has broadcast $(\ESet', \FSet')$ {\it only} after time $\TimeVSS$. 
         By the {\it $t_s$-consistency} and {\it $t_s$-fallback consistency} properties of $\BCAST$ in the {\it synchronous} network, 
         {\it all} honest parties will receive $(\ESet', \FSet')$
         from the broadcast of $\D$ by time $T + 2\Delta$. Moreover, 
         $(\ESet', \FSet')$ will constitute an $\Star{t_a}$ in every honest party's consistency graph by time $T + 2\Delta$.
         This is because all the $\OK$ messages which are received by $P_h$ from the broadcast of various parties in $\ESet'$ and $\FSet'$ 
         are guaranteed to be received by every honest party by time $T + 2\Delta$. 
         Since $|\ESet'| \geq n-2t_a > 2t_s + (t_s-t_a) > 2t_s$, it follows that $\ESet'$ has at least $t_s + 1$ {\it honest} parties.
         Moreover, 
          the univariate polynomials $(q_j(x), q_k(x))$ of every pair of {\it honest} parties $P_j, P_k \in \ESet'$ will be pair-wise consistent 
         and hence lie on a unique 
         $(t_s,t_s)$-degree symmetric bivariate polynomial, say $Q^{\star}(x, y)$.
                Similarly, the univariate polynomial $q_i(x)$ of every {\it honest} party $P_i$ in $\FSet'$ is pair-wise consistent with the univariate polynomials 
                $q_j(x)$ of {\it all} the {\it honest} parties in $\ESet'$
         and hence lie on $Q^{\star}(x, y)$ as well.
         Let $q^{\star}(\cdot) \defined Q^{\star}(0, y)$. We show that {\it every honest} $P_i$ outputs $q^{\star}(\alpha_i) $, by time $T + 2\Delta$.
          For this it is enough to show that each honest $P_i$ has $q_i(x) = Q^{\star}(x, \alpha_i)$ by time $T + 2\Delta$,
         as $P_i$ outputs $q_i(0)$, which will be the same as $q^{\star}(\alpha_i) $.
         Consider an arbitrary {\it honest} party $P_i$. There are two possible cases. 
         \begin{myitemize}
             \item[--] \textit{$P_i \in \FSet'$}: In this case, $P_i$ has already received $Q^{\star}(x, \alpha_i) $ from $\D$, well before time $T + 2\Delta$.
             \item[--] \textit{$P_i \notin \FSet'$}: In this case, we claim that by time $T + 2\Delta$, 
              there will be will be at least $t_s + 1$ parties from $\FSet'$, who are included in $\Support_i$, such that
                corresponding to {\it every} $P_j \in \Support_i$, party $P_i$ will have the value $q_{ji}$, which will be the same as $Q^{\star}(\alpha_j, \alpha_i) $. 
                The proof for this is very similar to the previous case when $P_i \notin \WCORE$ and the output of $\HBA$ is $0$.
              Namely every {\it honest} $P_j \in \FSet'$ will be included in $\Support_i$. This is because $P_j$ starts broadcasting $\OK$ messages for other parties in $\ESet'$
              only after invoking instance $\WPS^{(j)}$ with input $q_j(x)$. 
              Hence, by time $T + 2\Delta$, the WPS-share $q_{ji}$ from the instance $\WPS^{(j)}$ will be available with $P_i $. 
              On the other hand, if a {\it corrupt} $P_j \in \FSet'$ is included in $\Support_i$, then also the claim holds (the proof for this is similar to the proof of
              the {\it $t_a$-correctness} property in the {\it asynchronous} network in Lemma \ref{lemma:VSSAsynCorrectness}).    
         \end{myitemize}
 \end{myitemize}
 \end{proof}
 
 We finally prove the strong commitment property in an asynchronous network.
  \begin{lemma}[{\it $t_a$-Strong Commitment}]
 \label{lemma:VSSAsynStrongCommitment}
In protocol $\VSS$, if $\D$ is corrupt and network is asynchronous, then either no honest party computes any output or there exist some
  $t_s$-degree polynomial, say $q^{\star}(\cdot)$, such that almost-surely, every honest $P_i$ eventually outputs $q^{\star}(\alpha_i)$.			   
 \end{lemma}
 \begin{proof}
 Let $\D$ be {\it corrupt} and the network be {\it asynchronous} with up to $t_a$ corruptions. If no honest party computes any output, then the lemma holds trivially. So, consider the case when some honest party
  computes an output. We note that every {\it honest} party participates with some input in the instance of $\HBA$ at local time
 $\Delta + \TimeWPS + 2 \TimeBCAST$. Hence, from the {\it $t_a$-almost-surely liveness} and {\it $t_a$-consistency} properties of $\HBA$
  in the asynchronous network, 
   almost-surely, all honest parties eventually compute a common output during the instance of $\HBA$. 
   Now there are two possible cases:
    \begin{myitemize}
     \item[--] {\bf The output of $\HBA$ is 0}: 
     From the {\it $t_a$-validity} of $\HBA$ in the {\it asynchronous} network, it follows that at least one honest party, say $P_h$,
     participated with input $0$ during the instance of $\HBA$. 
     This means that $P_h$ has received some
       $(\WCORE, \ESet, \FSet)$ from the broadcast of $\D$ through regular-mode
        and accepted it at local time $\Delta + \TimeWPS + 2\TimeBCAST$.         
            Hence, by the {\it $t_a$-weak consistency} and {\it $t_a$-fallback consistency} of $\BCAST$ in the {\it asynchronous} network, 
            all honest parties will eventually receive $(\WCORE, \ESet, \FSet)$ from the broadcast of $\D$.
            There will be at least $n - 2t_s - t_a > t_s$ {\it honest} parties $P_i$ in $\ESet$, whose univariate polynomials $q_i(x)$ are pair-wise consistent and hence lie on
            a unique 
            $(t_s, t_s)$-degree symmetric bivariate polynomial, say $Q^{\star}(x,y)$.  
             Similarly, the univariate polynomial $q_i(x)$ of every {\it honest}
            $P_i \in \FSet$ will be pair-wise consistent with the univariate polynomials $q_j(x)$ of {\it all} the {\it honest} parties $P_j$ in $\ESet$ and hence will lie on 
            $Q^{\star}(x, y)$ as well. Let $q^{\star}(\cdot) \defined Q^{\star}(0, y)$.
            We claim that {\it every honest} $P_i$ will eventually have $Q^{\star}(x, \alpha_i) $. This will imply that eventually every {\it honest} $P_i$ outputs
            $s_i = Q^{\star}(0, \alpha_i) = q^{\star}(\alpha_i) $.
            To prove the claim, consider an arbitrary {\it honest} party $P_i$. There are three possible cases.
            \begin{myitemize}
                \item[--] {\it $P_i \in \WCORE$ and $P_i \in \FSet$}: In this case, $P_i$ has received the polynomials $q_i(x)$ from $\D$ and since 
                $P_i \in \FSet$, the condition $q_i(x) = Q^{\star}(x, \alpha_i)$ holds.
                \item[--] {\it $P_i \in \WCORE$ and $P_i \not \in \FSet$}: 
                In this case, $P_i$ has received the polynomial $q_i(x)$ from $\D$.
                Since $|\WCORE| \geq n-t_s$ and $|\FSet| \geq n-t_s$, $|\WCORE \cap \FSet| \geq n-2t_s > t_s + t_a$. 
                From the protocol steps, the polynomial $q_i(x)$ is pair-wise consistent with the polynomials $q_j(x)$ at least $n - t_s$ parties $P_j \in \WCORE$,
                since $P_i$ has edges with at least $n - t_s$ parties $P_j$ within $\WCORE$. Now among these $n - t_s$ parties, at least $n - 2t_s$ parties will be from $\FSet$, of which 
                at least $n - 2t_s - t_a > t_s$ parties will be {\it honest}. Hence, $q_i(x)$ is pair-wise consistent with the $q_j(x)$ polynomials of at least $t_s + 1$ 
                {\it honest} parties $P_j \in \FSet$. Now since the $q_j(x)$ polynomials of all the {\it honest} parties in $\FSet$ lie on $Q^{\star}(x, y)$, it implies that
                $q_i(x) = Q^{\star}(x, \alpha_i)$ holds.
                \item[--] {\it $P_i \notin \WCORE$}: In this case, similar to the proof of Lemma \ref{lemma:VSSAsynCorrectness}, one can show that 
                $P_i$ eventually includes at least $t_s + 1$ parties from $\FSet$ in $\Support_i $.
                 And the value computed by $P_i$ corresponding to
                {\it any} $P_j \in \Support_i$ will be the same as $Q^{\star}(\alpha_j, \alpha_i) $. 
                Hence, $P_i$ will eventually interpolate $Q^{\star}(x, \alpha_i) $.
            \end{myitemize}
  \item[--] \textbf{The output of $\HBA$ is 1}: Let $P_h$ be the {\it first honest} party, who computes an output in $\VSS$. 
         This means that $P_h$ has received some $(\ESet', \FSet')$ from the broadcast of $\D$, such that $(\ESet', \FSet')$ constitutes an $\Star{t_a}$
         in $P_h$'s consistency graph. By the {\it $t_a$-weak consistency} and {\it $t_a$-fallback consistency} properties of $\BCAST$ in
         the {\it asynchronous} network, 
         {\it all} honest parties eventually receive $(\ESet', \FSet')$
         from the broadcast of $\D$. Moreover, since the consistency graphs are constructed based on the
          broadcast $\OK$ messages and since
         $(\ESet', \FSet')$ constitutes an $\Star{t_a}$  in $P_h$'s consistency graph, from the {\it $t_a$-weak validity, $t_a$-fallback validity, $t_a$-weak consistency and 
         $t_a$-fallback consistency}
         properties of $\BCAST$ in the {\it asynchronous} network, the pair 
         $(\ESet', \FSet')$ will eventually constitute an $\Star{t_a}$ in every honest party's consistency graph, as the corresponding $\OK$ messages are eventually received by every
         honest party.         
         Since $|\ESet'| \geq n-2t_a > 2t_s + (t_s-t_a) > 2t_s$, it follows that $\ESet'$ has at least $t_s + 1$ {\it honest} parties $P_i$, whose 
         univariate polynomials $q_i(x)$ are pair-wise consistent and hence lie on a unique 
         $(t_s, t_s)$-degree symmetric bivariate polynomial, say $Q^{\star}(x, y)$. 
         Similarly, since the univariate polynomials $q_j(x)$ of every {\it honest} party $P_j$ in $\FSet'$ is pair-wise consistent with the univariate polynomials $q_i(x)$
          of all the {\it honest} parties $P_i$ in $\ESet'$,
         it implies that the polynomials $q_j(x)$ of all the {\it honest} parties $P_j$ in $\FSet'$ also lie on $Q^{\star}(x, y) $ as well.
         Let $q^{\star}(\cdot) \defined Q^{\star}(0, y)$. We show that {\it every honest} $P_i$ eventually
         outputs $q^{\star}(\alpha_i) $. For this it is enough to show that each honest $P_i$ eventually gets $q_i(x) = Q^{\star}(x, \alpha_i)$,
         as $P_i$ outputs $q_i(0)$, which will be the same as $q^{\star}(\alpha_i)$.
         Consider an arbitrary {\it honest} party $P_i$. There are two possible cases. 
         \begin{myitemize}
             \item[--] \textit{$P_i \in \FSet'$}: In this case, $P_i$ already has received $Q^{\star}(x, \alpha_i) $ from $\D$.
             \item[--] \textit{$P_i \notin \FSet'$}: Again in this case, one can show that 
                $P_i$ eventually includes at least $t_s + 1$ parties from $\FSet'$ in $\Support_i $. And the value computed by $P_i$ corresponding to
                {\it any} $P_j \in \Support_i$ will be the same as $Q^{\star}(\alpha_j, \alpha_i) $. 
                Hence, $P_i$ will eventually interpolate $Q^{\star}(x, \alpha_i) $.
         \end{myitemize}
    \end{myitemize}
\end{proof}
\begin{lemma}
 \label{lemma:VSSCommunication}
 Protocol $\VSS$ incurs a communication of $\Order(n^5 \log{|\F|})$ bits and invokes $n + 1$ instance of $\HBA$.		   
 \end{lemma}
\begin{proof}
The proof follows from Lemma \ref{lemma:WPSCommunication} and the fact that each party acts as a dealer and invokes an instance of $\WPS$ 
 with a $t_s$-degree polynomial. 
 Hence, the total communication cost due to the instances of
  $\WPS$ in $\VSS$ will be $\Order(n \cdot  n^4\log{|\F|}) = \Order(n^5 \log{|\F|})$ bits, along with $n$ instances
  of $\HBA$. Additionally, there is an instance of
  $\HBA$ invoked in $\VSS$ to agree on whether some $(\WCORE, \ESet, \FSet)$ is accepted.
\end{proof}

We next discuss the modifications needed in the protocol $\VSS$, if the input for $\D$ consists of $L$ number of $t_s$-degree polynomials.
\paragraph{\bf Protocol $\VSS$ for $L$ Polynomials:}
If $\D$ has $L$ polynomials as input for $\VSS$,
  then we make similar modifications as done for $\WPS$ handling $L$ polynomials, with each party broadcasting a 
 {\it single} $\OK/\NOK$ message for every other party. 
  To void repetition, we skip the formal details of the modified protocol and the proof of its properties, as stated in Theorem \ref{thm:VSS}.
\begin{theorem}
\label{thm:VSS}
Let $n > 3t_s + t_a$ and let $\D$ has $L$ number of $t_s$-degree polynomials $q^{(1)}(\cdot), \ldots, q^{(L)}(\cdot)$ as input for $\VSS$ where $L \geq 1$. Moreover, 
 let $\TimeVSS = \Delta + \TimeWPS + 2 \TimeBCAST + \TimeHBA$.
Then protocol $\VSS$ achieves the following properties.
  \begin{myitemize}
    \item[--] If $\D$ is honest, then the following hold.
       \begin{myitemize}
        \item[--] {\it $t_s$-correctness}: In a synchronous network, each (honest) $P_i$ outputs
         $\{q(\alpha_i)\}_{\ell = 1, \ldots, L}$ at time $\TimeVSS$.
        \item[--] {\it $t_a$-correctness}: In an asynchronous network, almost-surely, each (honest) $P_i$  
          eventually outputs $\{q(\alpha_i)\}_{\ell = 1, \ldots, L}$.
          \item[--] {\it $t_s$-privacy}: Irrespective of the network type, the view of the adversary 
   remains independent of the polynomials $q^{(1)}(\cdot), \ldots, q^{(L)}(\cdot)$.
        \end{myitemize}
    \item[--]  If $\D$ is corrupt, then either no honest party computes any output or there exist $t_s$-degree polynomials $\{{q^{\star}}^{(\ell)}(\cdot)\}_{\ell = 1, \ldots, L}$, such that
     the following hold.
       \begin{myitemize}
        \item[--] {\it $t_s$-strong commitment}: every honest $P_i$
          eventually outputs $\{{q^{\star}}^{(\ell)}(\alpha_i)\}_{\ell = 1, \ldots, L}$,
          such that one of the following hold.
           \begin{myitemize}
             \item[--] If any honest $P_i$ computes its output at time $\TimeVSS$, then all honest parties compute their output
          at time $\TimeVSS$.
             \item[--] If any honest $P_i$ computes its output at time $T$ where $T > \TimeVSS$, then every honest party computes its output by time $T + 2\Delta$.  
           \end{myitemize}        
        \end{myitemize}
      \item[--] Irrespective of the network type, the protocol incurs a communication of $\Order(n^3L \log{|\F|} + n^5 \log{|\F|})$ bits from the honest parties
      and invokes $n +1$ instances of $\HBA$.
   \end{myitemize}
\end{theorem}

\section{Agreement on a Common Subset (ACS)}
\label{sec:ACS}
In this section, we present a best-of-both-worlds protocol for agreement on a common subset, which will be later used in our
 preprocessing phase protocol, as well as in our circuit-evaluation protocol.  
 In the protocol, each party $P_i \in \PartySet$ has $L$ number of $t_s$-degree polynomials as 
  input for an instance of $\HSh$, which $P_i$ is supposed
  to invoke as a dealer.\footnote{The exact input of $P_i$ will be determined, based on where exactly the ACS protocol is used.}
  As {\it corrupt} parties may not invoke their instances
 of $\HSh$ as dealer, the parties may obtain points lying
  on the polynomials of only $n - t_s$ parties (even in a {\it synchronous} network).
  However, in an {\it asynchronous} network, different parties may obtain points on the polynomials of different subsets of $n - t_s$ parties.
  The ACS protocol allows the parties to agree on a {\it common subset} $\CoreSet$ of at least $n - t_s$ parties, 
   such that all (honest) parties are guaranteed to receive points lying on the polynomials of the parties in $\CoreSet$.
   Additionally, the protocol guarantees that in a {\it synchronous} network, all {\it honest} parties are present in $\CoreSet$. 
   Looking ahead, this property will be very crucial when the ACS protocol is used during circuit-evaluation, as it will ensure that in a 
   {\it synchronous} network, the inputs of {\it all} honest parties are considered for the circuit-evaluation.
   
      The ACS protocol is presented in Fig \ref{fig:acs}, where for simplicity we assume that $L = 1$.
       Later, we discuss the modifications required for $L > 1$.
     In the protocol, each party acts as a dealer and invokes an instance of $\HSh$ to verifiably distribute points on its polynomial. 
     If the network is {\it synchronous}, then after time $\TimeHSh$, all honest parties would have received points corresponding to the 
     polynomials of the {\it honest} dealers. Hence after (local) time $\TimeHSh$, 
     the parties {\it locally} check for the instances of $\HSh$ in which they have received an output. Based on this, the parties
   start participating in $n$ instances of $\HBA$, where the $j^{th}$ instance is used to decide
  whether $P_j$ should be included in $\CoreSet$. The input criteria for these $\HBA$ instances is the following: if a party has received an output
  in the $\HSh$ instance with $P_j$ as the dealer, then the party starts participating with input $1$ in the corresponding $\HBA$ instance. 
  Now once $1$ is obtained as the output from $n - t_s$ instances of $\HBA$, then the parties start participating with input $0$ in any of the 
   remaining $\HBA$ instances for which the parties may have not provided any input yet. Finally, once 
   an output is obtained from all the $n$ instances 
    of $\HBA$, party $P_j$ is included
   in $\CoreSet$ if and only if the output of the corresponding $\HBA$ instance is $1$.  
   Since the parties wait for time $\TimeHSh$ {\it before} starting the $\HBA$ instances, it is ensured that all {\it honest} dealers are 
   included in $\CoreSet$
   in a {\it synchronous} network.
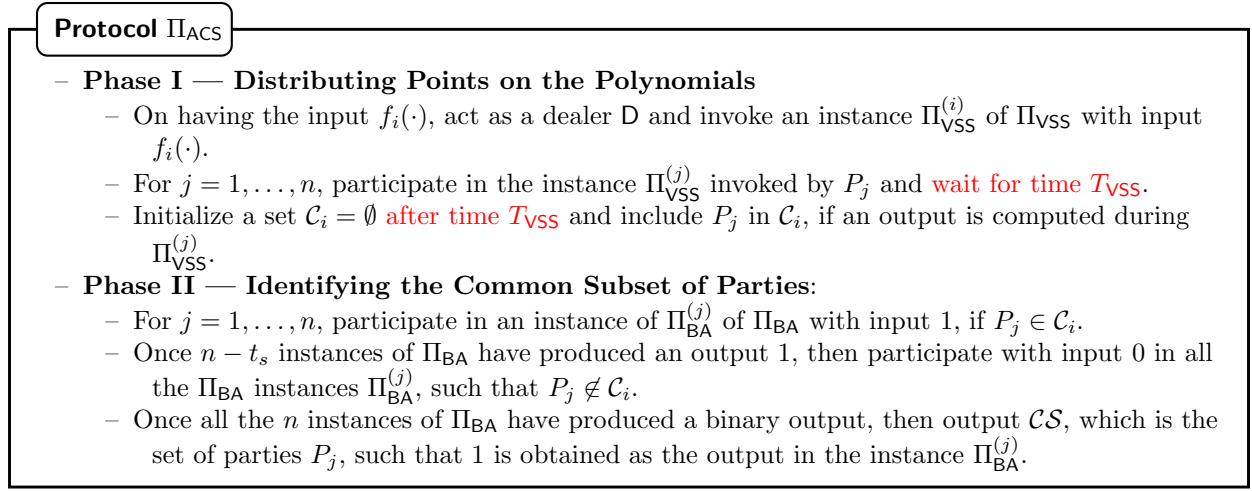
\begin{protocolsplitbox}{$\ACS$}{Agreement on common subset of $n-t_s$ parties where each party has a single $t_s$-degree polynomial 
 as input.
 The above code is executed by every $P_i \in \Partyset$.}{fig:acs}
  \begin{myitemize}
  \item[--] {\bf Phase I --- Distributing Points on the Polynomials}
       \begin{myitemize}
            \item[--] On having the input $f_i(\cdot)$, act as a dealer $\D$ and invoke an instance $\HSh^{(i)}$ of $\HSh$ with input
            $f_i(\cdot)$.
	    \item[--] For $j = 1, \ldots, n$, participate in the instance $\HSh^{(j)}$ invoked by $P_j$ and {\color{red} wait for time $\TimeHSh$}.        
             \item[--] Initialize a set $\CSet_i = \emptyset$ {\color{red} after time $\TimeHSh$}
              and include $P_j$ in 
            $\CSet_i$, if an output is computed during $\HSh^{(j)}$.  
            \end{myitemize}
    \item[--] {\bf Phase II --- Identifying the Common Subset of Parties}: 
           \begin{myitemize}
            \item[--]  For $j = 1, \ldots, n$, participate in an instance of $\HBA^{(j)}$ of $\HBA$ with input $1$, if $P_j \in \CSet_i$.
            \item[--] Once $n - t_s$ instances of $\HBA$ have produced an output $1$, then participate 
            with input $0$ in all the $\HBA$ instances $\HBA^{(j)}$, such that $P_j \not \in \CSet_i$.
            \item[--] Once all the $n$ instances of $\HBA$ have produced a binary output, then output $\CoreSet$, 
             which is the set of parties $P_j$, such that
                           $1$ is obtained as the output in the instance $\HBA^{(j)}$. 
             \end{myitemize} 
       \end{myitemize}
\end{protocolsplitbox}
 
 We next prove the properties of the protocol $\ACS$.
\begin{lemma}
\label{lemma:ACSProperties}
 Protocol $\ACS$ achieves the following properties, where every party $P_i$ has a $t_s$-degree polynomial $f_i(\cdot)$ as input.
  \begin{myitemize}
   \item[--] {\it Synchronous Network}: The following is achieved in the presence of up to $t_s$ corruptions.
      \begin{myitemize}
        \item[--] {\it $t_s$-Correctness}: at time $\TimeACS = \TimeHSh + 2\TimeHBA$, the parties output
         a common subset $\CoreSet$ of size at least $n - t_s$, 
        such that all the following hold:
          \begin{myitemize}
          \item[--] All honest parties will be present in $\CoreSet$.
          \item[--] Corresponding to every {\it honest} $P_j \in \CoreSet$, every {\it honest} $P_i$ has  $f_j(\alpha_i)$.
          \item[--] Corresponding to every {\it corrupt} $P_j \in \CoreSet$, there exists some $t_s$-degree polynomial, say $f^{\star}_j(\cdot)$, 
             such that
             every {\it honest} $P_i$ has $f^{\star}_j(\alpha_i)$.
          \end{myitemize}
      \end{myitemize}
    \item[--] {\it Asynchronous Network}: The following is achieved in the presence of up to $t_a$ corruptions.
       \begin{myitemize}
        \item[--] {\it $t_a$-Correctness}: almost-surely, 
         the honest parties eventually output a common subset $\CoreSet$ of size at least $n - t_s$, such that all the following hold:
          \begin{myitemize}
          \item[--] Corresponding to every {\it honest} $P_j \in \CoreSet$, every {\it honest} $P_i$ eventually has $f_j(\alpha_i)$.
          \item[--] Corresponding to every {\it corrupt} $P_j \in \CoreSet$, there exists some $t_s$-degree polynomial, say $f^{\star}_j(\cdot)$, such that
             every {\it honest} $P_i$ eventually has $f^{\star}_j(\alpha_i)$.
          \end{myitemize}
      \end{myitemize}
      \item[--] {\it $t_s$-Privacy}: Irrespective of the network type, the view of the adversary remains independent of the $f_i(\cdot)$ polynomials
    of the honest parties.      
     \item[--] Irrespective of the network type, the protocol incurs a communication of  $\Order(n^6 \log{|\F|})$ bits from the honest parties
      and 
     invokes
     $\Order(n^2)$ instances of $\HBA$.
  \end{myitemize}
\end{lemma}
\begin{proof}
The {\it $t_s$-privacy} property simply follows from the {\it $t_s$-privacy} property of $\HSh$, while communication complexity follows 
 from the communication complexity of $\HSh$ and 
  the fact that $\Order(n)$ instances of
  $\HSh$ are invoked. We next prove the {\it correctness} property.

We first consider a {\it synchronous} network, with up to $t_s$ corruptions. 
 Let $\Honest$ be the set of parties, where $|\Honest| \geq n - t_s$. Corresponding to each $P_j \in \Honest$, every {\it honest} $P_i$ 
  computes the output
 $f_j(\alpha_i)$ at time $\TimeHSh$ during $\HSh^{(j)}$, which follows from the {\it $t_s$-correctness} of $\HSh$ in
  the {\it synchronous} network.
  Consequently, at time $\TimeHSh$, the set $\CSet_i$ will be of size at least $n - t_s$ for every
  honest $P_i$. Now corresponding to each $P_j \in \Honest$, each honest $P_i$ participates with input $1$ in the instance $\HBA^{(j)}$ at time
  $\TimeHSh$. Hence, from the {\it $t_s$-validity} and
   {\it $t_s$-guaranteed liveness}
  of $\HBA$ in the {\it synchronous} network, it follows that at time $\TimeHSh + \TimeHBA$, every honest $P_i$ 
  computes the output $1$ during the instance $\HBA^{(j)}$,
  corresponding to every $P_j \in \Honest$. Consequently, at time $\TimeHSh + \TimeHBA$,
   every {\it honest} party will start participating in the remaining $\HBA$ instances for which no input has been provided yet (if there are any). 
  And from the {\it $t_s$-guaranteed liveness} and {\it $t_s$-consistency} of $\HBA$ in the {\it synchronous} network,
   these $\HBA$ instances will produce common outputs for every honest party at time $\TimeACS = \TimeHSh + 2 \TimeHBA$.
    Since the set $\CoreSet$ is determined deterministically based on the outputs
   computed from the $n$ instances of $\HBA$, it follows that  all the honest
  parties eventually output the same $\CoreSet$  of size at least $n - t_s$, such that 
   each $P_j \in \Honest$ will be present in $\CoreSet$.
  We next wish to show that corresponding to {\it every} $P_j \in \CoreSet$, every honest party has received its point on $P_j$'s polynomial.
  
 Consider an {\it arbitrary} party $P_j \in \CoreSet$. If $P_j$ is {\it honest}, then as argued above, every honest $P_i$ gets $f_j(\alpha_i)$ at time $\TimeHSh$ itself. 
  Next, consider a {\it corrupt} $P_j \in \CoreSet$. Since $P_j \in \CoreSet$, it follows that the instance $\HBA^{(j)}$ produces the output $1$.
  From the {\it $t_s$-validity} property of $\HBA$ in the {\it synchronous} network, it follows that
   at least one 
  {\it honest} $P_i$ must have participated
  with input $1$ in the instance $\HBA^{(j)}$. This implies that $P_i$ must have computed
  some output during the instance $\HSh^{(j)}$ by time $\TimeHSh + \TimeHBA$ and $P_j \in \CSet_i$.  
  This is because if at time $\TimeHSh + \TimeHBA$, 
  party  $P_j$ {\it does not} belong to the $\CSet_i$ set of {\it any} honest $P_i$,
   then it implies that {\it all honest} parties participate with input
   $0$ in the instance $\HBA^{(j)}$ from time $\TimeHSh + \TimeHBA$. 
   Then, from the {\it $t_s$-validity} of $\HBA$ in the {\it synchronous} network, every honest party would compute
    the output $0$ in the instance
  $\HBA^{(j)}$ and hence $P_j$ will not be present in $\CoreSet$, which is a contradiction. 
  Now if $P_i$ has computed
   some output during $\HSh^{(j)}$ at time $\TimeHSh + \TimeHBA$, then from the {\it $t_s$-strong-commitment} of $\HSh$, it follows that
  $P_j$ has some $t_s$-degree polynomial, say $f^{\star}_j(\cdot)$, such that every honest party $P_i$ computes
  $f^{\star}_j(\alpha_i)$ by time $\TimeHSh + \TimeHBA + 2 \Delta$.
  Since $2 \Delta < \TimeHBA$, it follows that at time $\TimeACS$, every honest $P_i$ has $f^{\star}_j(\alpha_i)$, thus proving the 
  {\it $t_s$-correctness} property in a {\it synchronous} network.
  
  We next consider an {\it asynchronous} network, with up to $t_a$ corruptions. 
    Let $\Honest$ be the set of parties, where $|\Honest| \geq n - t_a \geq n - t_s$.  We first note that irrespective of way messages are scheduled, 
  there will be at least $n - t_s$ instances of $\HBA$ in which all honest parties eventually participate with input $1$. This is because
  corresponding to every $P_j \in \Honest$, every {\it honest} $P_i$ {\it eventually} computes
   the output $f_j(\alpha_i)$ during the instance $\HSh^{(j)}$,
   which follows from the 
  {\it $t_a$-correctness} of $\HSh$ in the {\it asynchronous} network. 
   So even if the {\it corrupt} parties $P_j$ do not invoke their respective $\HSh^{(j)}$ instances, there will be at least $n - t_s$ instances of $\HBA$ in which all
  {\it honest} parties eventually participate with input $1$.
   Consequently, from the {\it $t_a$-almost-surely liveness} and {\it $t_a$-validity} properties of
   $\HBA$ in the {\it asynchronous} network, almost-surely, all honest parties eventually compute the output 
   $1$ during these 
   $\HBA$ instances. Hence, all honest parties eventually participate with some input in the remaining $\HBA$ instances.
   Consequently, from the {\it $t_a$-almost-surely liveness} and {\it $t_a$-consistency} properties of
   $\HBA$ in the {\it asynchronous} network, almost-surely, all honest parties will compute some common output
   in these $\HBA$ instances as well. Since the set $\CoreSet$ is determined deterministically based on the outputs
   computed from the $n$ instances of $\HBA$, it follows that  all the honest
  parties eventually output the same $\CoreSet$.
  
  Now consider an arbitrary party $P_j \in \CoreSet$. It implies that the honest parties 
  compute the output $1$ during the instance $\HBA^{(j)}$.
  From the {\it $t_a$-validity} of $\HBA$ in the {\it asynchronous} network, it follows 
    that at least one
   {\it honest} $P_i$ participated with input $1$ during $\HBA^{(j)}$,
    after computing some output in the instance $\HSh^{(j)}$. Now if $P_j$ is {\it honest}, then the
   {\it $t_a$-correctness} of $\HSh$ in the {\it asynchronous} network
    guarantees that every honest party $P_i$ {\it eventually} computes the output $f_j(\alpha_i)$ during $\HSh^{(j)}$.
   On the other hand, if $P_j$ is {\it corrupt}, then the {\it $t_a$-strong commitment} 
   of $\HSh$ in the {\it asynchronous} network
    guarantees that there exists some $t_s$-degree polynomial, say $f^{\star}_j(\cdot)$, such that 
    every honest party $P_i$ eventually computes the output $f^{\star}_j(\alpha_i)$ during the instance $\HSh^{(j)}$.
\end{proof}

We end this section by discussing the modifications needed in the protocol $\ACS$, if each party has
 $L$ number of polynomials as input.
\paragraph{\bf Protocol $\ACS$ for Multiple Polynomials:}
  Protocol $\ACS$ can be easily extended if each party has $L$ number of $t_s$-degree polynomials as input. In this case, each party $P_j$
  will invoke its instance of $\HSh$ with $L$ polynomials. The rest of the protocol steps remain the same. The protocol will incur a 
   communication of 
  $\Order(n^4 L \log{|\F|} + n^6 \log{|\F|})$ bits from the honest parties
  and invokes $\Order(n^2)$ instances of $\HBA$.

\section{The Preprocessing Phase Protocol}
\label{sec:preprocessing}
In this section, we present our best-of-both-worlds protocol for the preprocessing phase. The goal of the protocol is to generate
 $c_M$ number of $t_s$-shared multiplication-triples, which are random from the point of view of the adversary. The protocol is obtained
  by extending the framework of \cite{CP17} to the best-of-both-worlds setting. We first start by discussing the various 
   (best-of-both-worlds) building blocks used in the protocol.
 \subsection{best-of-both-worlds Beaver's Multiplication Protocol}
Given $t_s$-shared $x, y$ and a $t_s$-shared triple $(a, b, c)$, protocol $\BatchBeaver$ \cite{Bea91}
  outputs a $t_s$-shared $z$, where
 $z = x \cdot y$, if and only if $c = a \cdot b$. If $(a, b, c)$ is random for the adversary, then
  $x$ and $y$ remain random for the adversary. 
  In the protocol, the parties first 
  {\it publicly} reconstruct $x - a$ and $y - b$. A $t_s$-sharing of $z$ can be then computed locally, since 
   $[z] = (x - a) \cdot (y - b) + (x - a) \cdot [b] + (y - b) \cdot [a] + [c]$.
  The protocol takes $\Delta$ time in a {\it synchronous} network
   and in an {\it asynchronous} network, the parties
   {\it eventually} compute [z].    
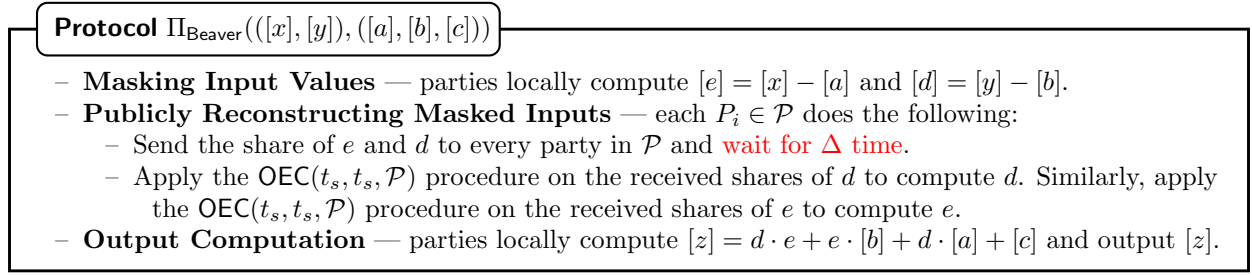
   \begin{protocolsplitbox}{$\BatchBeaver(([x],[y]),([a],[b],[c]))$}{Beaver's protocol for multiplying two $t_s$ shared values.}{fig:beaver}
\begin{myitemize}
    \item[--] {\bf Masking Input Values}  --- parties locally compute 
    $[e] = [x] - [a]$ and $[d] = [y] - [b]$.
    \item[--] {\bf Publicly Reconstructing Masked Inputs} --- each $P_i \in \Partyset$ does the following:
      \begin{myitemize}
      \item[--] Send the share of $e$ and $d$ to every party in $\Partyset$ and {\color{red} wait for $\Delta$ time}.
      \item[--] Apply the $\OEC(t_s, t_s, \Partyset)$ procedure on the received shares of $d$ to compute $d$.
      Similarly, apply the $\OEC(t_s, t_s, \Partyset)$ procedure on the received shares of $e$ to compute $e$.       
      \end{myitemize}
    \item[--] \textbf{Output Computation} --- parties locally compute $[z] = d \cdot e + e \cdot [b] +  d \cdot [a] + [c]$ and output $[z]$.
    \end{myitemize}
\end{protocolsplitbox}
\begin{lemma}
\label{lemma:BatchBeaverProperties}
Let $x$ and $y$ be two $t_s$-shared values and let $(a, b, c)$ be a $t_s$-shared triple. 
 Then protocol $\BatchBeaver$ achieves the following properties in the presence of up to $t_s$ corruptions.
  \begin{myitemize}
     \item[--] If the network is {\it synchronous}, then within time $\Delta$, the parties output
       a $t_s$-sharing of $z$.
      \item[--] If the network is asynchronous, then the parties eventually output a  
      $t_s$-sharing of $z$.
      \item[--] Irrespective of the network type, $z = x \cdot y$ holds, if and only if $(a, b, c)$ is a multiplication-triple.
      \item[--] Irrespective of the network type, if $(a, b, c)$ is random from the point of view of the adversary, then
        the view of the adversary 
        remains independent of $x$ and $y$.
       \item[--] The protocol incurs a communication of $\Order(n^2 \log{|\F|})$ bits from the honest parties.
   \end{myitemize}
\end{lemma}
\begin{proof}
Since $x, y$ and the triple $(a, b, c)$ are all $t_s$-shared, the values $d = (x - a)$ and $e = (y - b)$ will be $t_s$-shared, which follows from the linearity of $t_s$-sharing. Let there be 
 up to $t_s$ corruptions. If the network is {\it synchronous}, then from the properties
  of $\OEC$ in the {\it synchronous} network, 
  within time $\Delta$, {\it every} honest $P_i$ will have
 $d$ and $e$ and hence the parties output a $t_s$-sharing of $z$ within time $\Delta$. On the other hand, if the network is {\it asynchronous},
  then from the properties of $\OEC$ in the {\it asynchronous} network, {\it every} honest $P_i$ {\it eventually} reconstructs
 $d$ and $e$  and hence the honest parties eventually output  a $t_s$-sharing of $z$.
 
 In the protocol, $z = (x - a) \cdot (y - b) + (x - a) \cdot b + (y - b) \cdot a + c = x \cdot y - a \cdot b + c$ holds.
  Hence it follows that $z = x \cdot y$ holds if and only if $c = a \cdot b$ holds. 
 
In the protocol, adversary learns the values $d$ and $e$, as they are publicly reconstructed. However, if $a$ and $b$ are random from the point of view of the adversary,
  then $d$ and $e$ leak no information about $x$ and $y$. Namely, for every candidate $x$ and $y$, there exist unique
 $a$ and $b$, consistent with $d$ and $e$. 
 
 The communication complexity follows from the fact each party needs to send $2$ field elements to every other party.
\end{proof}
\subsection{best-of-both-worlds Triple-Transformation Protocol}
 Protocol $\TripTrans$ 
  takes input a set of $2d + 1$ 
   $t_s$-shared triples $\{(x^{(i)}, y^{(i)}, z^{(i)}) \}_{i = 1, \ldots, 2d + 1}$, where the triples may not be ``related". 
   The output of the protocol are  ``co-related" $t_s$-shared triples
   $\{(\mathbf{x}^{(i)}, \mathbf{y}^{(i)}, \mathbf{z}^{(i)}) \}_{i = 1, \ldots, 2d + 1}$, such that all the following hold (irrespective of the network type):
\begin{myitemize}
  \item[--] There exist $d$-degree polynomials $\mathsf{X}(\cdot), \mathsf{Y}(\cdot)$ and $2d$-degree polynomial $\mathsf{Z}(\cdot)$, 
   such that $\mathsf{X}(\alpha_i) = \mathbf{x}^{(i)}$, $\mathsf{Y}(\alpha_i) = \mathbf{y}^{(i)}$ and $\mathsf{Z}(\alpha_i) = \mathbf{z}^{(i)}$ holds
   for $i = 1, \ldots, 2d + 1$.
  \item[--] The triple $(\mathbf{x}^{(i)}, \mathbf{y}^{(i)}, \mathbf{z}^{(i)})$ is a multiplication-triple if and only if
    $(x^{(i)}, y^{(i)}, z^{(i)})$ is a multiplication-triple.  
    This further implies that $\mathsf{Z}(\cdot) = \mathsf{X}(\cdot) \cdot \mathsf{Y}(\cdot)$ holds if and only if all the
   $2d + 1$ input triples are multiplication-triples.
   \item[--] Adversary learns the triple $(\mathbf{x}^{(i)}, \mathbf{y}^{(i)}, \mathbf{z}^{(i)})$ if and only if it knows the input triple 
   $(x^{(i)}, y^{(i)}, z^{(i)})$. 
  \end{myitemize}
 The idea behind $\TripTrans$ is as follows: the polynomials $\mathsf{X}(\cdot)$ and $\mathsf{Y}(\cdot)$ are ``defined''
 by the first and second components of the {\it first} $d+1$ input triples.  
 Hence the first $d + 1$ points on these polynomials are already $t_s$-shared.
   The parties then compute $d$ ``new'' points on the 
  polynomials $\mathsf{X}(\cdot)$ and $\mathsf{Y}(\cdot)$ in a shared fashion. This step requires the parties to 
  perform only {\it local computations}. 
  This is because from the property of Lagrange's interpolation, computing any new point on $\mathsf{X}(\cdot)$
  and $\mathsf{Y}(\cdot)$ involves computing a {\it publicly-known linear} function (which we call {\it Lagrange's linear function})
  of ``old" points on these polynomials. Since the old points are $t_s$-shared, by applying corresponding Lagrange's functions, the parties can compute a $t_s$-sharing of
  the new points. Finally, the parties compute a $t_s$-sharing of the product of the $d$ new points using Beaver's technique, making use of
  the {\it remaining} $d$ input triples. The $\mathsf{Z}(\cdot)$ polynomial is then defined by the $d$ computed products
   and the third component of the first $d+1$ input triples.
  The protocol is formally presented in Fig \ref{fig:triptrans}.
 \begin{protocolsplitbox}{$\TripTrans(d, \{[x^{(i)}],[y^{(i)}],[z^{(i)}]\}_{i = 1, \ldots, 2d + 1})$}{Protocol for transforming a set of 
   $t_s$-shared triples into a set of correlated $t_s$-shared triples.}{fig:triptrans}
\begin{myitemize}
    \item[--] {\bf Defining $\mathsf{X}(\cdot)$ and $\mathsf{Y}(\cdot)$ Polynomials} --- The parties locally do the following:
     \begin{myitemize}
      \item[--]  For $i = 1, \ldots, d + 1$, set
        \begin{myitemize}
        \item[--] $[\mathbf{x}^{(i)}] = [x^{(i)}]$;
        \item[--] $[\mathbf{y}^{(i)}] = [y^{(i)}]$;
        \item[--] $[\mathbf{z}^{(i)}] = [z^{(i)}]$.    
        \end{myitemize}
      \item[--] Let $\mathsf{X}(\cdot)$ be the unique $d$-degree polynomial, 
      passing through the points $\{(\alpha_i,\mathbf{x}^{(i)})\}_{i = 1, \ldots, d + 1}$. And let
      $\mathsf{Y}(\cdot)$ be the unique $d$-degree polynomial, 
      passing through
      $\{(\alpha_i,\mathbf{y}^{(i)})\}_{i = 1, \ldots, d + 1}$.
          \begin{myitemize}
           \item[--] For $i = d + 2, \ldots, 2d + 1$, 
           locally compute $[\mathbf{x}^{(i)}] = [\mathsf{X}(\alpha_i)]$ from $\{ [\mathbf{x}^{(i)}]\}_{i = 1, \ldots, d + 1}$, by applying the corresponding
           Lagrange's linear function.
             \item[--] For $i = d + 2, \ldots, 2d + 1$, locally compute $[\mathbf{y}^{(i)}] = [\mathsf{Y}(\alpha_i)]$ from 
             $\{ [\mathbf{y}^{(i)}]\}_{i = 1, \ldots, d + 1}$, by applying the corresponding
           Lagrange's linear function.
          \end{myitemize}
     \end{myitemize}
    \item[--] {\bf Computing Points on the $\mathsf{Z}(\cdot)$ Polynomial} --- The parties do the following:
        \begin{myitemize}
          \item[--] For $i = d + 2, \ldots, 2d + 1$, participate in the instance 
          $\BatchBeaver(([\mathbf{x}^{(i)}],[\mathbf{y}^{(i)}]), ([x^{(i)}], [y^{(i)}], [z^{(i)}]))$ of $\BatchBeaver$.
          Let $[\mathbf{z}^{(i)}]$ be the output obtained from this instance. 
          \item[--] Output $\{[\mathbf{x}^{(i)}], [\mathbf{y}^{(i)}], [\mathbf{z}^{(i)}] \}_{i = 1, \ldots, 2d + 1}$. 
           \end{myitemize}
\end{myitemize}
\end{protocolsplitbox}

We next prove the properties of the protocol $\TripTrans$.
\begin{lemma}
\label{lemma:TripTrans}
Let $\{[x^{(i)}], [y^{(i)}], [z^{(i)}] \}_{i = 1, \ldots, 2d + 1}$ be a set of $t_s$-shared triples. Then protocol $\TripTrans$
 achieves the following properties in the presence of up to $t_s$ corruptions.
 \begin{myitemize}
   \item[--] If the network is {\it synchronous}, then 
   the parties output
    $t_s$-shared triples $\{[\mathbf{x}^{(i)}], [\mathbf{y}^{(i)}], [\mathbf{z}^{(i)}] \}_{i = 1, \ldots, 2d + 1}$, 
    within time $\Delta$.
    \item[--] If the network is asynchronous, then the parties eventually output
     $t_s$-shared triples $\{[\mathbf{x}^{(i)}], [\mathbf{y}^{(i)}], \allowbreak [\mathbf{z}^{(i)}] \}_{i = 1, \ldots, 2d + 1}$.
        \item[--] Irrespective of the network type,
         there exist $d$-degree polynomials $\mathsf{X}(\cdot), \mathsf{Y}(\cdot)$ and $2d$-degree polynomial $\mathsf{Z}(\cdot)$, such that
       $\mathsf{X}(\alpha_i) = \mathbf{x}^{(i)}$, $\mathsf{Y}(\alpha_i) = \mathbf{y}^{(i)}$ and $\mathsf{Z}(\alpha_i) = \mathbf{z}^{(i)}$ holds for 
	   $i = 1, \ldots, 2d + 1$. 
	\item[--] Irrespective of the network type,
	 $(\mathbf{x}^{(i)}, \mathbf{y}^{(i)}, \mathbf{z}^{(i)})$ is a multiplication-triple if and only if
	   $(x^{(i)}, y^{(i)}, \allowbreak z^{(i)})$ is a multiplication-triple.  
         \item[--] For $i = 1, \ldots, 2d + 1$, no additional information about $(\mathbf{x}^{(i)}, \mathbf{y}^{(i)}, \mathbf{z}^{(i)})$ is revealed to
        the adversary,  if the triple $(x^{(i)}, y^{(i)}, z^{(i)})$ is random from the point of view of the adversary.
        \item[--] The protocol incurs a communication of $\Order(dn^2 \log{|\F|})$ bits from the honest parties.
      \end{myitemize} 
\end{lemma}
\begin{proof}
Consider an adversary who controls up to $t_s$ parties. In the protocol, irrespective of the network type, the parties {\it locally} compute the
$t_s$-sharings $\{[\mathbf{x}^{(i)}], [\mathbf{y}^{(i)}]\}_{i = 1, \ldots, 2d + 1}$ and $t_s$-sharings $\{[\mathbf{z}^{(i)}]\}_{i = 1, \ldots, d + 1}$. 
 If the network is {\it synchronous}, then from the properties
  of $\BatchBeaver$ in the {\it synchronous} network,
  it follows that after time $\Delta$,
  all honest parties will have their respective output in all the $d$ instances of $\BatchBeaver$.
  Hence after time $\Delta$, the parties have 
  $t_s$-sharings $\{[\mathbf{z}^{(i)}]\}_{i = d + 2, \ldots, 2d + 1}$.
  On the other hand, if the network is {\it asynchronous}, then from the
   properties of $\BatchBeaver$ in the {\it asynchronous} network,
   all honest parties eventually compute their output in 
    all the $d$ instances of  $\BatchBeaver$. Hence the 
  parties eventually compute the $t_s$-sharings $\{[\mathbf{z}^{(i)}]\}_{i = d + 2, \ldots, 2d + 1}$ and hence eventually 
  compute  their output in the protocol.
 
 Next consider an {\it arbitrary}
  $i \in \{1, \ldots, d + 1 \}$. Since $(\mathbf{x}^{(i)}, \mathbf{y}^{(i)}, \mathbf{z}^{(i)}) = (x^{(i)}, y^{(i)}, z^{(i)})$, it follows that
 $(\mathbf{x}^{(i)}, \mathbf{y}^{(i)}, \mathbf{z}^{(i)})$ will be a multiplication-triple if and only if $(x^{(i)}, y^{(i)}, z^{(i)})$ is a multiplication-triple.
 Now consider an arbitrary $i \in \{d + 2, \ldots, 2d + 1 \}$. Since $[\mathbf{z}^{(i)}]$ is the output of the instance
 $\BatchBeaver(([\mathbf{x}^{(i)}],[\mathbf{y}^{(i)}]), ([x^{(i)}], [y^{(i)}],   [z^{(i)}]))$,
  it follows from the properties of $\BatchBeaver$ that 
  $\mathbf{z}^{(i)} = \mathbf{x}^{(i)} \cdot \mathbf{y}^{(i)}$ holds, if and only if $(x^{(i)}, y^{(i)}, z^{(i)})$ is a multiplication-triple.
  
  From the protocol steps, it is easy to see that the polynomials $\mathsf{X}(\cdot)$ and $\mathsf{Y}(\cdot)$ defined in the protocols are
  $d$-degree polynomials, as they are defined through $d + 1$ distinct points 
  $\{(\alpha_i, \mathbf{x}^{(i)} ) \}_{i = 1, \ldots, d + 1}$ and $\{(\alpha_i, \mathbf{y}^{(i)} ) \}_{i = 1, \ldots, d + 1}$ respectively.
  On the other hand, $\mathsf{Z}(\cdot)$ is a $2d$-degree polynomial, as it is defined through the 
  $2d + 1$ distinct points $\{(\alpha_i, \mathbf{z}^{(i)} ) \}_{i = 1, \ldots, 2d + 1}$.
  
  For any $i \in \{1, \ldots, d + 1 \}$, if $(x^{(i)}, y^{(i)}, z^{(i)})$ is random from the point of view of the adversary, then
        $(\mathbf{x}^{(i)}, \mathbf{y}^{(i)}, \mathbf{z}^{(i)})$ is also random from the point of view of the adversary, since
        $(\mathbf{x}^{(i)}, \mathbf{y}^{(i)}, \mathbf{z}^{(i)}) = (x^{(i)}, y^{(i)}, z^{(i)})$. 
        On the other hand for any $i \in \{d + 2, \ldots, 2d + 1 \}$, if $(x^{(i)}, y^{(i)}, z^{(i)})$ is random from the point of view of the adversary,
        then from the properties
         of $\BatchBeaver$, it follows that no additional information is learnt about 
        $(\mathbf{x}^{(i)}, \mathbf{y}^{(i)}, \mathbf{z}^{(i)})$.
        
 The communication complexity follows from the fact that there are $d$ instances of $\BatchBeaver$ invoked in the protocol.
\end{proof}
 \subsection{best-of-both-worlds Triple-Sharing Protocol}
We next present a triple-sharing protocol $\TripSh$, which
  allows a dealer $\D$ to {\it verifiably} $t_s$-share $L$ multiplication-triples.
  If $\D$ is {\it honest},  then the triples remain random from the point of view of the adversary
  and all honest parties output the shares of $\D$'s multiplication-triples.
    On the other hand, if $\D$ is {\it corrupt}, then 
   the protocol {\it need not} produce any output, even in a {\it synchronous} network, as a corrupt $\D$ may not invoke the protocol at the first
    place and the parties will not be aware of the network type.
  However, the ``verifiability" of $\TripSh$ guarantees that if the honest parties compute
   any output corresponding to a {\it corrupt} $\D$, then $\D$ has indeed $t_s$-shared multiplication-triples.

   For simplicity, we present the protocol assuming $\D$ has a {\it single} multiplication-triple to share and the 
   protocol can be easily generalized for any
   $L > 1$.
    The idea behind the protocol is as follows: $\D$ picks a random multiplication-triple and $t_s$-shares it by invoking an instance of $\HSh$.
   To prove that it has indeed shared a multiplication-triple, $\D$ actually $t_s$-shares $2t_s + 1$ random multiplication-triples.
   The parties then run an instance of $\TripTrans$ and ``transform" these shared triples into ``co-related" shared triples, constituting
   distinct points on the triplet of polynomials $(\mathsf{X}(\cdot), \mathsf{Y}(\cdot), \mathsf{Z}(\cdot))$, 
   which are guaranteed to exist during $\TripTrans$.
   Then, to check if all the triples shared by $\D$ are multiplication-triples, it is sufficient to verify if 
   $\mathsf{Z}(\cdot) = \mathsf{X}(\cdot) \cdot \mathsf{Y}(\cdot)$ holds. To verify the latter, we incorporate a mechanism which 
   enables the parties to {\it publicly} learn if 
   $\mathsf{Z}(\alpha_j) = \mathsf{X}(\alpha_j) \cdot \mathsf{Y}(\alpha_j)$ holds under $P_j$'s ``supervision"
   in such a way that if $P_j$ is {\it honest}, then the supervised verification of the triplet $(\mathsf{X}(\alpha_j), \mathsf{Y}(\alpha_j), \mathsf{Z}(\alpha_j))$
   is ``successful" if and only if
    the triplet is a multiplication-triple. Moreover, the privacy of the triplet will be maintained during the supervised verification for an {\it honest}
   $\D$ {\it and} $P_j$. The goal is to check whether there are at least $2t_s + 1$ successful supervised-verifications, performed under the supervision of 
   {\it honest} supervisors $P_j$, which will then confirm that indeed  $\mathsf{Z}(\cdot) = \mathsf{X}(\cdot) \cdot \mathsf{Y}(\cdot)$ holds. This is because
   $\mathsf{Z}(\cdot)$ is a $2t_s$-degree polynomial. Upon confirming that $\mathsf{Z}(\cdot) = \mathsf{X}(\cdot) \cdot \mathsf{Y}(\cdot)$ holds,
       the parties compute a ``new" point on the polynomials (in a shared fashion), which is taken as
   the output triple shared {\it on behalf of} $\D$. We stress that the output triple is well defined and will be ``known" to $\D$, as it is {\it deterministically} determined from
   the triples shared by $\D$. If $\D$ is {\it honest}, then the privacy of the output triple is guaranteed from the fact that during the supervised verification, an adversary may learn at most
   $t_s$ distinct points on the polynomials $\mathsf{X}(\cdot), \mathsf{Y}(\cdot)$ and $\mathsf{Z}(\cdot)$, corresponding to the {\it corrupt} supervisors.
   
   The supervised verification of the (shared) points on the polynomials is performed as follows: the parties invoke an instance of 
   $\ACS$, where the input for each party
   is a triplet of random $t_s$-degree polynomials, whose constant terms constitute a random multiplication-triple, called {\it verification-triple}.
    The instance of $\ACS$ is invoked
   in parallel with $\D$'s invocation of $\HSh$. Through the instance of $\ACS$, the parties agree upon a set $\WCORE$ of at least $n - t_s$ supervisors, whose shared verification-triples
   are used to verify the points on the polynomials $\mathsf{X}(\cdot), \mathsf{Y}(\cdot)$ and $\mathsf{Z}(\cdot)$. Namely, if 
   $P_j \in \WCORE$ has shared the verification-triple $(u^{(j)}, v^{(j)}, w^{(j)})$, then in the supervised verification under $P_j$, parties {\it publicly} reconstruct and check
   if $\mathsf{Z}(\alpha_j) - \mathsf{X}(\alpha_j) \cdot \mathsf{Y}(\alpha_j) = 0$ holds.
    For this, the parties recompute $\mathsf{X}(\alpha_j) \cdot \mathsf{Y}(\alpha_j)$ in a shared fashion using Beaver's method, by deploying the
    shared verification-triple $(u^{(j)}, v^{(j)}, w^{(j)})$. If $\mathsf{Z}(\alpha_j) - \mathsf{X}(\alpha_j) \cdot \mathsf{Y}(\alpha_j)$ {\it does not} turn out to be $0$
    (implying that either $\D$ is {\it corrupt} or $P_j$'s verification-triple is {\it not} a multiplication-triple), then the parties
    {\it publicly} reconstruct and check if $(\mathsf{X}(\alpha_j), \mathsf{Y}(\alpha_j), \mathsf{Z}(\alpha_j))$ is a multiplication-triple
    and discard $\D$ if the triple {\it does not} turn out to be a multiplication-triple.

    An {\it honest} $\D$ will {\it never} be discarded. Moreover, in a {\it synchronous} network, {\it all honest} parties $P_j$ are guaranteed to be present in
    $\WCORE$ (follows from the {\it $t_s$-correctness} of $\ACS$ in the {\it synchronous} network) 
    and hence, there will be at least $n - t_s > 2t_s + 1$ {\it honest} supervisors in $\WCORE$. On the other hand, even in an {\it asynchronous}
    network, there will be at least $n - t_s - t_a > 2t_s$ {\it honest} supervisors  in $\WCORE$. Hence if a {\it corrupt}
    $\D$ is {\it not} discarded, then it is guaranteed that $\D$ has shared multiplication-triples.       
\begin{protocolsplitbox}{$\TripSh$}{A protocol for verifiably sharing a single multiplication triple.}{fig:tripsh}
\begin{myitemize}
    \item[--] {\bf Phase I --- Sharing Triples and Verification-Triples}:
	    \begin{myitemize}
	        \item[--] $\D$ selects $2t_s + 1$ random multiplication-triples $\{(x^{(j)}, y^{(j)}, z^{(j)})\}_{j = 1, \ldots, 2t_s + 1}$. It then selects
	        random $t_s$-degree polynomials $\{f_{x^{(j)}}(\cdot), f_{y^{(j)}}(\cdot), f_{z^{(j)}}(\cdot)\}_{j = 1, \ldots, 2t_s + 1}$, such that
	        $f_{x^{(j)}}(0) = x^{(j)}$, $f_{y^{(j)}}(0) = y^{(j)}$ and $f_{z^{(j)}}(0) = z^{(j)}$.
	        $\D$ then invokes an instance of $\HSh$ with input $\{f_{x^{(j)}}(\cdot), f_{y^{(j)}}(\cdot), f_{z^{(j)}}(\cdot)\}_{j = 1, \ldots, 2t_s + 1}$
	        and the parties in $\Partyset$ participate in this instance.
	        \item[--] In parallel, each party $P_i \in \Partyset$ randomly selects a 
	        {\it verification multiplication-triple} $(u^{(i)}, v^{(i)}, \allowbreak w^{(i)})$ and random $t_s$-degree polynomials
              $f_{u^{(i)}}(\cdot), f_{v^{(i)}}(\cdot)$ and $f_{w^{(i)}}(\cdot)$ where $f_{u^{(i)}}(0) = u^{(i)}$, 
                $f_{v^{(i)}}(0) = v^{(i)}$ and $f_{w^{(i)}}(0) = w^{(i)}$. With these polynomials as inputs, $P_i$ participates
                 in an instance of $\ACS$ and {\color{red} waits for time $\TimeACS$}.
              Let $\WCORE$ be the set of parties, computed as the output during the instance of $\ACS$,
                where $|\WCORE| \geq n - t_s$.
	    \end{myitemize}
    \item[--] {\bf Phase II --- Transforming $\D$'s Triples}:
	    \begin{myitemize}
	    \item[--] Upon computing an output in the instance of $\HSh$ invoked by $\D$,
	     the parties participate in an instance $\TripTrans(t_s, \{[x^{(j)}], [y^{(j)}], [z^{(j)}]\}_{j = 1, \ldots, 2t_s + 1})$ of $\TripTrans$.
	    \item[--] Let $\{[\mathbf{x}^{(j)}], [\mathbf{y}^{(j)}], [\mathbf{z}^{(j)}]\}_{j = 1, \ldots, 2t_s + 1})$ be the set of $t_s$-shared triples
	     computed during
	    $\TripTrans$. And let $\mathsf{X}(\cdot)$ and $\mathsf{Y}(\cdot)$ be the 
	    $t_s$-degree polynomials and $\mathsf{Z}(\cdot)$ be the $2t_s$-degree polynomial, which are
	    guaranteed to exist during the instance of $\TripTrans$, such that $\mathsf{X}(\alpha_j) = \mathbf{x}^{(j)}$, 	    
	    $\mathsf{Y}(\alpha_j) = \mathbf{y}^{(j)}$ and $\mathsf{Z}(\alpha_j) = \mathbf{z}^{(j)}$, for $j = 1, \ldots, 2t_s + 1$.
	    \item[--] For $j = 2t_s + 2, \ldots, n$, the parties do the following.
	      \begin{myitemize}
	       \item[--] Locally compute 
	      $[\mathbf{x}^{(j)}] = [\mathsf{X}(\alpha_j)]$ from $\{[\mathbf{x}^{(j)}] \}_{j = 1, \ldots, t_s + 1}$,
	      by using appropriate Lagrange's linear functions.
	      \item[--] Locally compute $[\mathbf{y}^{(j)}] = [\mathsf{Y}(\alpha_j)]$ from $\{[\mathbf{y}^{(j)}] \}_{j = 1, \ldots, t_s + 1}$,
	      by using appropriate Lagrange's linear functions.	      
	      \item[--] Locally compute 
	   $[\mathbf{z}^{(j)}] = [\mathsf{Z}(\alpha_j)]$ from $\{[\mathbf{z}^{(j)}] \}_{j = 1, \ldots, 2t_s + 1}$ respectively, 
	   by using appropriate Lagrange's linear functions.
	    \end{myitemize}
	\end{myitemize}    
     \item[--] {\bf Phase III --- Verifying Transformed Triples}: The parties do the following.
             \begin{myitemize}
              \item[--] {\bf Phase III(a) --- Recomputing the Products}:
                \begin{myitemize}
	               \item[--] Corresponding to each $P_j \in \WCORE$, participate in
	              the instance $\BatchBeaver(([\mathbf{x}^{(j)}], [\mathbf{y}^{(j)}]), 
	                ([u^{(j)}],  \allowbreak [v^{(j)}],  [w^{(j)}]))$ of $\BatchBeaver$ to compute $[\z^{(j)}]$.
	        \end{myitemize}      
              \item[--] {\bf Phase III(b) --- Computing and Publicly Reconstructing the Differences}:
                  \begin{myitemize}
	              \item[--] Corresponding to every $P_j \in \WCORE$, the parties locally compute 
	              $[\gamma^{(j)}] = [\mathbf{z}^{(j)}] - [\z^{(j)}]$.
               		\item[--] Corresponding to every $P_j \in \WCORE$, the parties publicly reconstruct 
		$\gamma^{(j)}$, by exchanging their respective shares of
		$\gamma^{(j)}$, followed by applying the $\OEC(t_s, t_s, \Partyset)$ procedure on the received shares.		
        		       \item[--] Corresponding to $P_j \in \WCORE$, party $P_i \in \Partyset$ upon reconstructing $\gamma^{(j)}$,
                       sets a Boolean variable $\flag^{(j)}_i$ to $0$ if $\gamma^{(j)} = 0$, else it sets $\flag^{(j)}_i$ to $1$.      
	             \end{myitemize}
               \item[--] {\bf Phase III(c) --- Checking the Suspected Triples}: Each $P_i \in \Partyset$ does the following.
        		    \begin{myitemize}
		             \item[--] For every $P_j \in \WCORE$ such that $\flag^{(j)}_i = 1$, send the shares
		              corresponding to $[\mathbf{x}^{(j)}], [\mathbf{y}^{(j)}]$
		             and $[\mathbf{z}^{(j)}]$ to every party.		             
		             \item[--] For every $P_j \in \WCORE$ such that $\flag^{(j)}_i = 1$, apply the
		              $\OEC(t_s, t_s, \Partyset)$ procedure on the received shares corresponding to
		             $[\mathbf{x}^{(j)}], [\mathbf{y}^{(j)}]$
		             and $[\mathbf{z}^{(j)}]$, to reconstruct the triple $(\mathbf{x}^{(j)}, \mathbf{y}^{(j)}, \mathbf{z}^{(j)})$.
		             \item[--] For every $P_j \in \WCORE$ such that $\flag^{(j)}_i = 1$,
		             reset $\flag^{(j)}_i$ to $0$ if  $(\mathbf{x}^{(j)}, \mathbf{y}^{(j)}, \mathbf{z}^{(j)})$ 
		             is a multiplication-triple.
		            \item[--] If $\flag^{(j)}_i = 0$, corresponding to every $P_j \in \WCORE$, then set $\flag_i = 0$, else set $\flag_i = 1$. 
	            \end{myitemize}
            \end{myitemize}
      \item[--] {\bf Output Computation}: Each party $P_i \in \Partyset$ does the following.
             \begin{myitemize}      
                \item[--] If $\flag_i = 0$ then output shares corresponding to $t_s$-shared triple
                $([a], [b], [c])$ {\it on behalf of} $\D$, where
                $a = \mathsf{X}(\beta)$, $b = \mathsf{Y}(\beta)$ and $c = \mathsf{Z}(\beta)$
                and where 
                $[a]$, $[b]$ and $[c]$ are locally computed from $\{ [\mathbf{x}^{(j)}] \}_{j = 1, \ldots, t_s + 1}$, 
                $\{ [\mathbf{y}^{(j)}] \}_{j = 1, \ldots, t_s + 1}$ and $\{ [\mathbf{z}^{(j)}] \}_{j = 1, \ldots, 2t_s + 1}$ respectively by
                 using appropriate Lagrange's linear functions.
                Here $\beta$ is a non-zero element from $\F$, distinct from $\alpha_1, \ldots, \alpha_{2t_s + 1}$.
              \item[--] If $\flag_i = 1$ then output default-shares (namely all shares being $0$) corresponding to $t_s$-shared triple
              $([0], [0], [0])$ {\it on behalf of} $\D$.
              \end{myitemize}
\end{myitemize}
\end{protocolsplitbox}

We next prove the properties of the protocol of $\TripSh$.
\begin{lemma}
\label{lemma:TripSh}
Protocol $\TripSh$ achieves the following properties.
  \begin{myitemize}
    \item[--] If $\D$ is {\it honest}, then the following hold:
        \begin{myitemize}
         \item[--] {\it $t_s$-Correctness}:  If the network is synchronous, then after
          time $\TimeTripSh = \TimeACS + 4\Delta$, the honest parties output a 
          $t_s$-shared multiplication-triple on the behalf of
           $\D$.
          \item[--]  {\it $t_a$-Correctness}: If the network is asynchronous, then almost-surely,
           the (honest) parties eventually output a $t_s$-shared multiplication-triple on the behalf of
           $\D$.
   \item[--] {\it $t_s$-Privacy}: Irrespective of the network type, the view of the adversary remains independent of the output
    multiplication-triple, shared on the behalf of $\D$.
        \end{myitemize}
     \item[--] If $\D$ is {\it corrupt}, then either no honest party computes any output or depending upon the network type, the following hold
            \begin{myitemize}
	            \item[--] {\it $t_s$-Strong Commitment}:
	            If the network is synchronous, then the (honest) parties eventually output a $t_s$-shared multiplication-triple on behalf of $\D$. 
	             Moreover, if some honest party computes its output shares at time $T$, 
                     then by time $T + 2\Delta$, all honest parties will compute their respective output shares.
                     \item[--] {\it $t_a$-Strong Commitment}: The (honest) parties eventually output a $t_s$-shared multiplication-triple on
                      the behalf of $\D$. 			   
	    \end{myitemize}
       \item[--] The protocol incurs a communication of $\Order(n^6 \log{|\F|})$ bits from the honest parties
        and invokes $\Order(n^2)$ instances of $\HBA$.
  \end{myitemize}
\end{lemma}
\begin{proof}
We first consider an {\it honest} $\D$ and prove the corresponding properties. We first consider
  a {\it synchronous} network with up to $t_s$ corruptions. At time $\TimeHSh$, the multiplication-triples $\{(x^{(j)}, \allowbreak y^{(j)}, 
  z^{(j)})\}_{j = 1, \ldots, 2t_s + 1}$
 will be $t_s$-shared. This follows from the {\it $t_s$-correctness} property of $\HSh$ in 
  the {\it synchronous} network.
   Moreover, these triples will be random from the
 point of view of the adversary, which follows from the {\it $t_s$-privacy} property of $\HSh$.
 Since the instance of $\ACS$ is invoked in parallel with the instance of $\HSh$ invoked by $\D$, 
  at time $\TimeACS$, all honest parties will have a common subset $\WCORE$ from the instance of $\ACS$, 
  with {\it every honest} $P_j$ being present in the $\WCORE$. This follows from the 
   properties of $\ACS$ in the {\it synchronous} network. 
  At time $\TimeACS + \Delta$, the multiplication-triples shared by $\D$ will be transformed and parties will have
  $t_s$-shared multiplication-triples $\{(\mathbf{x}^{(j)}, \mathbf{y}^{(j)},  \mathbf{z}^{(j)})\}_{j = 1, \ldots, 2t_s + 1}$
   and there will exist $t_s$-degree polynomials $\mathsf{X}(\cdot), \mathsf{Y}(\cdot)$ and $2t_s$-degree polynomial $\mathsf{Z}(\cdot)$
    where $\mathsf{Z}(\cdot) = \mathsf{X}(\cdot) \cdot \mathsf{Y}(\cdot)$ holds.
     This follows from the  properties of $\TripTrans$ in the {\it synchronous} network. 
 
 Next, corresponding to {\it every honest} $P_j \in \WCORE$, the value $\z^{(j)}$ will be the same 
  as $\mathbf{x}^{(j)} \cdot \mathbf{y}^{(j)}$, which follows from the
  properties of $\BatchBeaver$ and the fact that the corresponding verification-triple
 $(u^{(j)}, v^{(j)}, w^{(j)})$ will be a multiplication-triple.
   Hence, $\gamma^{(j)} = \mathbf{z}^{(j)} - \z^{(j)}$ will be $0$ and so each {\it honest} $P_i$ will set $\flag^{(j)}_i$ to $0$, without
  suspecting and reconstructing the triple $(\mathbf{x}^{(j)}, \mathbf{y}^{(j)}, \mathbf{z}^{(j)})$. Moreover, 
  in this case, no additional information about $(\mathbf{x}^{(j)}, \mathbf{y}^{(j)}, \mathbf{z}^{(j)})$ is revealed, which follows from the
   properties of $\BatchBeaver$ and the fact that the verification-triple $(u^{(j)}, v^{(j)}, w^{(j)})$ remains random 
    from the point of view of the adversary. 
   On the other hand, if $P_j \in \WCORE$ is {\it corrupt}, then
  $\gamma^{(j)}$ may not be $0$. However, in this case each {\it honest} $P_i$ will reset $\flag^{(j)}_i$ to $0$ after reconstructing
  the corresponding suspected-triple $(\mathbf{x}^{(j)}, \mathbf{y}^{(j)}, \mathbf{z}^{(j)})$, since it will be a multiplication-triple. 
  The process of computing $\z^{(j)}$ and the difference $\gamma^{(j)}$ will take $2\Delta$ time and additionally $\Delta$ time might be required
  to publicly reconstruct suspected-triples corresponding to {\it corrupt} $P_j \in \WCORE$.
  Hence, at time $\TimeACS + 4\Delta$, each {\it honest} $P_i$ sets $\flag_i = 1$ and hence, the honest parties output
  $t_s$-shared triple $(a, b, c)$. Moreover, the triple will be a multiplication-triple, since $(a, b, c)$ is the same as
  $(\mathsf{X}(\beta), \mathsf{Y}(\beta), \mathsf{Z}(\beta))$.
  Since at most $t_s$ triples $(\mathbf{x}^{(j)}, \mathbf{y}^{(j)}, \mathbf{z}^{(j)})$
   may be publicly reconstructed corresponding to the {\it corrupt} parties $P_j \in \WCORE$, it follows that
   adversary will learn at most 
  $t_s$ distinct points on the $\mathsf{X}(\cdot), \mathsf{Y}(\cdot)$ and $\mathsf{Z}(\cdot)$ polynomials. 
  This further implies that $(\mathsf{X}(\beta), \mathsf{Y}(\beta), \mathsf{Z}(\beta))$ will be random from the point of view of the adversary, since
  $\mathsf{X}(\cdot), \mathsf{Y}(\cdot)$ are $t_s$-degree polynomials and $\mathsf{Z}(\cdot)$ is a $2t_s$-degree polynomial. 
  This completes the proof of the {\it $t_s$-correctness} in the {\it synchronous} network, as well as the proof of the 
   {\it $t_s$-privacy} property.
  
  If $\D$ is {\it honest} and the network is {\it asynchronous} with up to $t_a$ corruptions, then the proof of the
   {\it $t_a$-correctness} property is similar to the above proof, 
  except that now we now use the {\it $t_a$-correctness} property of $\HSh$ and the properties of
  $\ACS$, $\BatchBeaver$ in the {\it asynchronous} network.
  Moreover, the privacy property holds since adversary now corrupt $t_a < t_s$ parties.
  
  We next consider a {\it corrupt} $\D$ and prove the strong-commitment properties. 
   We first consider a {\it synchronous} network with up to $t_s$ corruptions.
  Note that irrespective of whether $\D$ shares any triples through instance of $\HSh$ or not,
  all honest parties will output a set $\WCORE$ at time $\TimeACS$ during the instance of $\ACS$, 
   with {\it every honest} $P_j$ being present in the $\WCORE$. This follows from the 
  properties of $\ACS$ in the {\it synchronous} network. 
 If no honest party computes any output in the protocol, then strong-commitment holds trivially.
  So consider the case when some honest party
  computes an output. This implies that 
  at least one {\it honest} party, say $P_h$, must have computed an output during the instance of $\HSh$ invoked by $\D$, as otherwise
  no honest party computes any output in the protocol.
  Let $T$ be the time at which $P_h$ has the output for the  instance of $\HSh$ invoked by $\D$. 
     Note that $T \geq \TimeHSh$ and $T$ could be greater than $\TimeACS$, as a {\it corrupt} $\D$ may delay the start of the instances of
      $\HSh$. 
  From the {\it $t_s$-strong commitment} property of $\HSh$ in the {\it synchronous} network,
   it then follows that by time $T + 2\Delta$, {\it all} honest parties compute their output in the instance
  of $\HSh$ invoked by $\D$. Hence at time $T + 2\Delta$, there are 
    $2t_s + 1$ triples which are $t_s$-shared by $\D$.
  
  If $T \leq \TimeACS - 2\Delta$, then at time $\TimeACS$, all honest parties will have their respective shares corresponding to the
  $t_s$-shared triples of $\D$, as well as the set $\WCORE$ and the 
  shares corresponding to the verification-triples, shared by the parties in $\WCORE$.
  The instance of $\TripTrans$ will produce its output at time $\TimeACS + \Delta$. The follow up instances of $\BatchBeaver$
  to recompute the products will take $\Delta$ time, followed by $\Delta$ time for publicly reconstructing the difference values $\gamma^{(j)}$.
  Additionally, the parties may take $\Delta$ time to publicly reconstruct any suspected triples. Hence in this case, {\it all honest} parties will have their respective output shares
  at time $\TimeACS + 4\Delta$. 
  
  On the other hand, if $T > \TimeACS - 2\Delta$, then each
  honest party computes its output during the instance of $\TripTrans$, either at time $T + \Delta$ or at time $T + 3\Delta$.
  Then, each
  honest party computes its output from the instances of $\BatchBeaver$, either at time $T + 2\Delta$ or at time $T + 4\Delta$.
  This implies that the difference values $\gamma^{(j)}$ are available with the honest parties, either at time $T + 3\Delta$ or
  $T + 5 \Delta$. Consequently, the suspected triples (if any) will be available with the honest parties, either at time
  $T + 4\Delta$ or  $T + 6 \Delta$. Hence each honest party computes its output share in the protocol either at time $T + 4\Delta$ or $T + 6\Delta$.
   Notice that in this case there might be a difference of at most $2\Delta$ time within which the honest parties compute
    their output in the protocol, due to a possible difference
  of $2\Delta$ time in getting the output in the instances of $\HSh$ invoked by the {\it corrupt} $\D$.
  
  If the triples shared by $\D$ during the instances of $\HSh$ are all {\it multiplication-triples}, then similar to the proof of the
   correctness property for an {\it honest}
  $\D$, it follows that the honest parties will output a $t_s$-shared multiplication-triple on behalf of $\D$. 
   So consider the case when all the triples shared by $\D$ are {\it not} multiplication-triples. This implies that  $\mathsf{Z}(\cdot) \neq \mathsf{X}(\cdot) \cdot \mathsf{Y}(\cdot)$,
   where $\mathsf{X}(\cdot), \mathsf{Y}(\cdot)$ are the $t_s$-degree polynomials and $\mathsf{Z}(\cdot)$ is the $2t_s$-degree polynomial,
   which are guaranteed to exist from the protocol $\TripTrans$. 
   Let $P_j$ be an {\it honest} party, such that $\mathsf{Z}(\alpha_j) \neq \mathsf{X}(\alpha_j) \cdot \mathsf{Y}(\alpha_j)$. This further implies that
  the transformed triple $(\mathbf{x}^{(j)}, \mathbf{y}^{(j)}, \mathbf{z}^{(j)})$ is {\it not} a multiplication-triple. Such a $P_j$ is bound to exist. This is because there are at least $2t_s + 1$ honest parties
  $P_j$. And if $\mathsf{Z}(\alpha_j) = \mathsf{X}(\alpha_j) \cdot \mathsf{Y}(\alpha_j)$ holds corresponding to {\it every} honest $P_j$, then it implies that
  $\mathsf{Z}(\cdot) = \mathsf{X}(\cdot) \cdot \mathsf{Y}(\cdot)$ holds (due to the degrees of the respective polynomials), which is a contradiction. 
  
    We next show that each {\it honest} $P_i$ will set
  $\flag^{(j)}_i = 1$ and hence $\flag_i = 1$. For this, we note that $P_j \in \WCORE$. This follows from the
   properties of $\ACS$ in the {\it synchronous} network, 
   which guarantees that {\it all} honest parties (and not just $P_j$) will be present in $\WCORE$.
     Since the verification-triple  $(u^{(j)}, v^{(j)}, w^{(j)})$ shared by $P_j$ will be a multiplication-triple, from the
    properties of $\BatchBeaver$ in the {\it synchronous} network, 
  it follows that $\z^{(j)} = \mathbf{x}^{(j)} \cdot \mathbf{y}^{(j)}$ holds. But since
  $\mathbf{z}^{(j)} \neq \mathbf{x}^{(j)} \cdot \mathbf{y}^{(j)}$, it follows that $\gamma^{(j)} = \mathbf{z}^{(j)} - \z^{(j)} \neq 0$.
  Consequently, the parties will publicly reconstruct the suspected-triple $(\mathbf{x}^{(j)}, \mathbf{y}^{(j)}, \mathbf{z}^{(j)})$ and
   find that it is not a multiplication-triple.
  Hence each {\it honest} $P_i$ will set $\flag^{(j)}_i = 1$ and hence $\flag_i = 1$. So the parties output a default $t_s$-sharing
   of the multiplication-triple
  $(0, 0, 0)$ on behalf of $\D$.
  
  The proof for the {\it $t_a$-strong-commitment} property in the {\it asynchronous} network is similar to the above proof,
  except that we now use the {\it $t_a$-strong-Commitment} property of $\HSh$ and the properties
  of $\ACS$ and $\BatchBeaver$ in the {\it asynchronous} network. 
  Moreover, there will be at least $n - t_s - t_a \geq 2t_s + 1$ honest parties in $\WCORE$, who will lead the verification
  of at least $2t_s + 1$ distinct points on the polynomials  $\mathsf{X}(\cdot), \mathsf{Y}(\cdot)$ and $\mathsf{Z}(\cdot)$, through their respective
  verification-triples.
  
  The communication complexity follows from communication complexity of $\ACS, \HSh, \TripTrans$ and $\BatchBeaver$.
\end{proof}

We next discuss the modifications needed in the protocol $\TripSh$ to handle $L$ multiplication-triples.
\paragraph{\bf Protocol $\TripSh$ for Sharing $L$ Multiplication-Triples:}
Protocol $\TripSh$ can be easily generalized so that $L$ multiplication-triples are shared on behalf of $\D$. Namely, $\D$ now has to share
 $L \cdot (2t_s + 1)$  random multiplication-triples through $\HSh$, while each party $P_j$ will need to select $L$ verification-triples during the instance of
 $\ACS$. Moreover, there will be $L$ instances of $\TripTrans$ to transform $\D$'s shared triples, resulting in $L$
  triplets of shared polynomials $(\mathsf{X}(\cdot), \mathsf{Y}(\cdot), \mathsf{Z}(\cdot))$, each of which is independently verified by performing supervised verification.
  To avoid repetition, we do not provide the formal details.
  The modified $\TripSh$ protocol will incur a communication of 
  $\Order(n^4 L \log{|\F|} + n^6 \log{|\F|})$ bits from the honest parties and invokes $\Order(n^2)$ instances of $\HBA$. 
  \subsection{\bf best-of-both-worlds Triple-Extraction Protocol}
Protocol $\TripExt$ (Fig \ref{fig:tripext}) takes as input a {\it publicly-known} subset $\CoreSet$ 
 of $2d + 1$ parties, where $d \geq t_s$ and where it will be {\it ensured} that each party
 $P_j \in \CoreSet$ has $t_s$-shared a multiplication-triple. It will also be ensured that if $P_j$ is {\it honest}, then the multiplication-triple is random from the point of view of the adversary.
  The protocol outputs $d + 1 - t_s$ number of $t_s$-shared multiplication-triples, which will be random from the point of view of the adversary. 
  The high level idea of the protocol is very simple. The parties first invoke an instance of $\TripTrans$ to ``transform" the input triples into a set of co-related triples. 
  Since all the input triples are multiplication-triples, the output triples will also be multiplication-triples. Let $(\mathsf{X}(\cdot), \mathsf{Y}(\cdot), \mathsf{Z}(\cdot))$
  be the triplet of shared polynomials which is guaranteed to exist after $\TripTrans$. From the 
   properties of $\TripTrans$, it follows that adversary will know
  at most $t_s$ distinct points on these polynomials and hence at least $d + 1 - t_s$ points on these polynomials are random for the adversary.
   Hence, the parties output $d + 1 - t_s$ ``new" points on these polynomials (in a $t_s$-shared fashion), which are guaranteed to be random from the point of view of the adversary.
   This requires the parties to perform {\it only} local computation.
\begin{protocolsplitbox}{$\TripExt(\CoreSet, \{[x^{(j)}], [y^{(j)}],[ z^{(j)}]\}_{P_j \in \CoreSet})$}{Protocol for extracting $d + 1 - t_s$ 
 random $t_s$-shared random multiplication-triples from a set of $2d + 1$ $t_s$-shared multiplication triples, where $d \geq t_s$.}{fig:tripext}
\begin{myitemize}
    \item[--] {\bf Transforming the Input Multiplication-Triples} --- The parties jointly do the following:
    \begin{myitemize}
    \item[--] Participate in an instance $\TripTrans(d, \{[x^{(j)}], [y^{(j)}], [z^{(j)}] \}_{P_j \in \CoreSet})$ of $\TripTrans$.
      \begin{myitemize}
	     \item[--] Let $\{[\mathbf{x}^{(j)}], [\mathbf{y}^{(j)}], [\mathbf{z}^{(j)}] \}_{P_j \in \CoreSet}$ be the shared multiplication-triples
	     obtained from $\TripTrans$. Moreover, let $\mathsf{X}(\cdot), \mathsf{Y}(\cdot)$ be the $d$-degree polynomials
	     and $\mathsf{Z}(\cdot)$ be the $2d$-degree polynomial where $\mathsf{Z}(\cdot) = \mathsf{X}(\cdot) \cdot \mathsf{Y}(\cdot)$
	     and where $\mathsf{X}(\alpha_j) = \mathbf{x}^{(j)}$, $\mathsf{Y}(\alpha_j) = \mathbf{y}^{(j)}$ and $\mathsf{Z}(\alpha_j) = \mathbf{z}^{(j)}$ holds
	     corresponding to every $P_j \in \CoreSet$.
	       \item[--] For $j = 1, \ldots, d + 1 - t_s$, locally compute $[\mathbf{a}^{(j)}]$, $[\mathbf{b}^{(j)}]$
	       and $[\mathbf{c}^{(j)}]$ from $\{[\mathbf{x}^{(j)}]\}_{j = 1, \ldots, d + 1}$, $[\{\mathbf{y}^{(j)}]\}_{j = 1, \ldots, d + 1}$
	        and $\{[\mathbf{z}^{(j)}]\}_{j = 1, \ldots, 2d + 1}$ respectively
	       	       	  by applying the corresponding Lagrange's linear function. 
	     Here $\mathbf{a}^{(j)} = \mathsf{X}(\beta_j)$, 	$\mathbf{b}^{(j)} = \mathsf{Y}(\beta_j)$
	     and $\mathbf{c}^{(j)} = \mathsf{Z}(\beta_j)$, where $\beta_1, \ldots, \beta_{d + 1 - t_s}$ are distinct, non-zero elements
	     from $\F$, different from $\alpha_1, \ldots, \alpha_n$.
	       \item[--] Output $\{[\mathbf{a}^{(j)}], [\mathbf{b}^{(j)}], [\mathbf{c}^{(j)}]\}_{j = 1, \ldots, d + 1 - t_s}$.	   
	       \end{myitemize}
      \end{myitemize} 	    
    \end{myitemize}
\end{protocolsplitbox}

\begin{lemma}
\label{lemma:TripExt}
Let $\CoreSet$ be a set of $2d + 1$ parties where $d \geq t_s$, such that each party $P_j \in \CoreSet$ has a multiplication-triple
 $(x^{(j)}, y^{(j)}, z^{(j)})$ which is $t_s$-shared. Moreover, if $P_j$ is {\it honest}, then the multiplication-triple is random from the point of view of the adversary.
  Then protocol $\TripExt$ achieves the following properties.
   \begin{myitemize}
     \item[--] {\it $t_s$-Correctness}:  If the network is {\it synchronous}, then after time $\Delta$,
   the parties output
    $t_s$-shared multiplication-triples $\{\mathbf{a}^{(j)}, \mathbf{b}^{(j)}, \mathbf{c}^{(j)} \}_{j = 1, \ldots, d + 1 - t_s}$. 
    \item[--] {\it $t_a$-Correctness}: If the network is asynchronous, then the parties eventually output $t_s$-shared multiplication-triples
    $\{\mathbf{a}^{(j)}, \mathbf{b}^{(j)}, \mathbf{c}^{(j)} \}_{j = 1, \ldots, d + 1 - t_s}$. 
       \item[--] {\it $t_s$-Privacy}: Irrespective of the network type, 
       the triples $\{\mathbf{a}^{(j)}, \mathbf{b}^{(j)}, \mathbf{c}^{(j)} \}_{j = 1, \ldots, d + 1 - t_s}$ will be random
        from the point of view of the adversary.
      \item[--] The protocol incurs a communication of $\Order(dn^2 \log{|\F|})$ bits from the honest parties.
   \end{myitemize}
\end{lemma}
\begin{proof}
 If the network is {\it synchronous}, then from the properties of $\TripTrans$ in the 
 {\it synchronous} network, it follows that
  after time $\Delta$, the honest parties have $t_s$-shared triples
  $\{\mathbf{x}^{(j)}, \mathbf{y}^{(j)}, \mathbf{z}^{(j)} \}_{P_j \in \CoreSet}$. Moreover, all these triples will be multiplication-triples,
   since all the input triples
  $\{{x}^{(j)}, {y}^{(j)}, {z}^{(j)} \}_{P_j \in \CoreSet}$ are guaranteed to be multiplication-triples. This further implies that the condition
  $\mathsf{Z}(\cdot) = \mathsf{X}(\cdot) \cdot \mathsf{Y}(\cdot)$ holds, where 
  $\mathsf{X}(\cdot), \mathsf{Y}(\cdot)$ and $\mathsf{Z}(\cdot)$ are the $d, d$ and $2d$-degree polynomials respectively,
  which are guaranteed to exist from $\TripTrans$, such that 
  $\mathsf{X}(\alpha_j) = \mathbf{x}^{(j)}$, $\mathsf{Y}(\alpha_j) = \mathbf{y}^{(j)}$ and $\mathsf{Z}(\alpha_j) = \mathbf{z}^{(j)}$ holds
	     for every $j$, such that $P_j \in \CoreSet$. 
  It now follows that the honest parties output the $t_s$-shared triples $\{(\mathbf{a}^{(j)}, \mathbf{b}^{(j)}, \mathbf{c}^{(j)}) \}_{j = 1, \ldots, d + 1 - t_s}$
  after time $\Delta$, where $\mathsf{X}(\beta_j) = \mathbf{a}^{(j)}$, $\mathsf{Y}(\beta_j) = \mathbf{b}^{(j)}$ and $\mathsf{Z}(\beta_j) = \mathbf{c}^{(j)}$.
  Moreover, the triples will be multiplication-triples, because $\mathsf{Z}(\cdot) = \mathsf{X}(\cdot) \cdot \mathsf{Y}(\cdot)$ holds.
	     
 If the network is {\it asynchronous}, then the proof of {\it $t_a$-correctness} property will be similar as above,
  except that we now depend upon the
  properties of $\TripTrans$ in the {\it asynchronous} network.
 
 For {\it privacy}, we note that there will be at most $t_s$ {\it corrupt} parties in $\CoreSet$ and hence adversary will know at most $t_s$ multiplication-triples
 in the set $\{\mathbf{x}^{(j)}, \mathbf{y}^{(j)}, \mathbf{z}^{(j)} \}_{P_j \in \CoreSet}$, which follows from the properties of $\TripTrans$.
   This implies that adversary will know at most 
  $t_s$
  distinct points on the polynomials
 $\mathsf{X}(\cdot), \mathsf{Y}(\cdot)$ and $\mathsf{Z}(\cdot)$, leaving $d + 1 - t_s$ degrees of freedom on the these polynomials. This further implies that
 the multiplication-triples $\{(\mathsf{X}(\beta_j), \mathsf{Y}(\beta_j), \mathsf{Z}(\beta_j)) \}_{j = 1, \ldots, d + 1 - t_s}$, which are the
 same as $\{\mathbf{a}^{(j)}, \mathbf{b}^{(j)}, \mathbf{c}^{(j)} \}_{j = 1, \ldots, d + 1 - t_s}$,
  will be random from the point of view of the
 adversary. Namely, there will be a one-to-one correspondence between the $d + 1 - t_s$ multiplication-triples
  in the set $\{\mathbf{x}^{(j)}, \mathbf{y}^{(j)}, \mathbf{z}^{(j)} \}_{P_j \in \CoreSet}$ which are {\it unknown} to the adversary and 
  the output multiplication-triples $\{\mathbf{a}^{(j)}, \mathbf{b}^{(j)}, \mathbf{c}^{(j)} \}_{j = 1, \ldots, d + 1 - t_s}$.
  Hence adversary's view will be consistent with every candidate value of the output  $d + 1 - t_s$ multiplication-triples.
  
  The  communication complexity simply follows from the fact that the protocol requires one instance of $\TripTrans$. 
\end{proof}
\subsection{The best-of-both-worlds Preprocessing Phase Protocol}
  We finally present our best-of-both-worlds preprocessing phase protocol $\Offline$, which generates
   $c_M$ number of $t_s$-shared multiplication-triples, which will be random from the point of view of the adversary.
   The protocol is formally presented in Fig \ref{fig:Offline}. In the protocol, each party acts as a dealer and invokes an instance of
    $\TripSh$, so that
   $\frac{c_M}{(\frac{n - t_s - 1}{2} + 1 - t_s)}$
   random multiplication-triples are shared on its behalf. 
   As {\it corrupt} dealers may not invoke their instances of $\TripSh$ (even in a {\it synchronous} network), 
   the  parties agree on a {\it common} subset $\CoreSet$ of $n - t_s$ parties, who have shared multiplication-triples,
  by executing instances of $\HBA$ (similar to the protocol $\ACS$).
  The multiplication-triples shared on the behalf
  of up to $t_s$ {\it corrupt} triple-providers in $\CoreSet$ will be known to adversary, while the multiplication-triples shared on the behalf
   of the {\it honest}
  triple-providers in $\CoreSet$ will be random for the adversary. 
  Since the exact identity of the {\it honest} triple-providers in $\CoreSet$ will not be known, the parties execute 
    $\frac{c_M}{(\frac{n - t_s - 1}{2} + 1 - t_s)}$ instances of
   $\TripExt$ to securely extract $c_M$
    shared multiplication-triples, which will be random for the adversary. 
   In the protocol, for simplicity and without loss of generality, we assume that $n - t_s$  is of the form $2d + 1$.
   \begin{protocolsplitbox}{$\Offline$}{The best-of-both-worlds preprocessing phase
    protocol for generating shared random multiplication-triples.}{fig:Offline}
Let $L \defined \frac{c_M}{(\frac{n - t_s - 1}{2} + 1 - t_s)}$.
\justify
  \begin{myitemize}
      \item[--] {\bf Phase 1 --- Sharing Random Multiplication-Triples}:  Each $P_i \in \Partyset$ does the following.
	    \begin{myitemize}
	        \item[--]  Act as a dealer $\D$ and invoke an instance $\TripSh^{(i)}$ of $\TripSh$,
	         so that $L$ random multiplication-triples are shared on $P_i$'s behalf.
	        \item[--] For $j = 1, \ldots, n$, participate in the instance $\TripSh^{(j)}$ invoked by $P_j$ and 
	        {\color{red} wait for time $\TimeTripSh$}.     
        	     \item[--] Initialize a set $\CSet_i = \emptyset$ and include $P_j$ in 
            $\CSet_i$, if an output is obtained from $\TripSh^{(j)}$.  
    \end{myitemize}
      \item[--] {\bf Phase II --- Agreement on a Common Subset of Triple-Providers}:
       Each $P_i \in \Partyset$ does the following.
          \begin{myitemize}
            \item[--]  For $j = 1, \ldots, n$, participate in an instance of $\HBA^{(j)}$ of $\HBA$ with input $1$, if $P_j \in \CSet_i$.
            \item[--] Once $n - t_s$ instances of $\HBA$ have produced an output $1$, then participate 
            with input $0$ in all the $\HBA$ instances $\HBA^{(j)}$, such that $P_j \not \in \CSet_i$.
            \item[--] Once a binary output is computed in all the $n$ instances of $\HBA$, 
            set $\CoreSet$ to be the set of {\it first} $n - t_s$ parties $P_j$, such that
                           $1$ is computed as the output in the instance $\HBA^{(j)}$. 
             \end{myitemize} 
    \item[--] {\bf Phase III --- Extracting Random Multiplication-Triples}: The parties do the following.
         \begin{myitemize}
         \item[--] For every $P_j \in \CoreSet$, let $\{[x^{(j, \ell)}], [y^{(j, \ell)}], [z^{(j, \ell)}] \}_{\ell = 1, \ldots, L}$
         be the $t_s$-shared multiplication-triples, shared on $P_j$'s behalf, during the instance $\TripSh^{(j)}$.
         \item[--] For $\ell = 1, \ldots, L$, the parties participate in an instance 
         $\TripExt(\CoreSet, \{[x^{(j, \ell)}], [y^{(j, \ell)}],[ z^{(j, \ell)}]\}_{P_j \in \CoreSet})$ of $\TripExt$
         and compute the output $\{[\mathbf{a}^{(j, \ell)}], [\mathbf{b}^{(j, \ell)}], [\mathbf{c}^{(j, \ell)}]\}_{j = 1, \ldots, \frac{n - t_s - 1}{2} + 1 - t_s}$.         
         \item[--] Output the shared triples
          $\{[\mathbf{a}^{(j, \ell)}], [\mathbf{b}^{(j, \ell)}], [\mathbf{c}^{(j, \ell)}]\}_{j = 1, \ldots, \frac{n - t_s - 1}{2} + 1 - t_s, \ell = 1, \ldots, L}$. 
         \end{myitemize}
  \end{myitemize}
\end{protocolsplitbox}
\begin{theorem} \label{thm:Offline}
Protocol $\Offline$ achieves the following properties.
 \begin{myitemize}
       \item[--] In a synchronous network, by time $\TimeOffline = \TimeTripSh + 2\TimeHBA + \Delta$, the 
      honest parties output a $t_s$-sharing of $c_M$ multiplication-triples.
         \item[--] In an asynchronous network, almost-surely,
   the honest parties eventually output a $t_s$-sharing of $c_M$ multiplication-triples.
  \item[--] Irrespective of the network type, the view of the adversary
   remains independent of the output multiplication-triples.
   \item[--] 
    The protocol incurs a communication of $\Order(\frac{n^5}{\frac{t_a}{2}+1}c_M \log{|\F|} + n^7 \log{|\F|})$ bits 
    from the honest parties and invokes $\Order(n^3)$ instances of $\HBA$.    
 \end{myitemize}
\end{theorem}
\begin{proof}
Let the network be {\it synchronous} with up to $t_s$ corruptions.
  Let $\Honest$ be the set of parties, where $|\Honest| \geq n - t_s$. Corresponding to each $P_j \in \Honest$, at time $\TimeTripSh$, $L$ multiplication-triples 
  $\{x^{(j, \ell)}, y^{(j, \ell)}, z^{(j, \ell)} \}_{\ell = 1, \ldots, L}$ will be $t_s$-shared on the 
  behalf of $P_j$ during the instance $\TripSh^{(j)}$,
   which follows from the 
   {\it $t_s$-correctness} of $\TripSh$ in the {\it synchronous} network. 
   Consequently, the set $\CSet_i$ will be of size at least $n - t_s$ for every
  honest $P_i$. After time $\TimeTripSh$, corresponding to each $P_j \in \Honest$, each honest $P_i$ participates with input $1$
   in the instance $\HBA^{(j)}$.
   It then follows from the {\it $t_s$-validity} and
  {\it $t_s$-guaranteed liveness}
  of $\HBA$ in the {\it synchronous} network that corresponding to every $P_j \in \Honest$, every honest $P_i$
  computes the output $1$ during the instance $\HBA^{(j)}$, at time
  $\TimeTripSh + \TimeHBA$. Consequently, after time $\TimeTripSh + \TimeHBA$,
   every honest party will start participating in the remaining $\HBA$ instances for which no input has been provided yet (if there are any). 
  And from the {\it $t_s$-guaranteed liveness} and {\it $t_s$-consistency} of $\HBA$ in the {\it synchronous} network,
  all honest parties will compute a common output in these 
  $\HBA$ instances, at time $\TimeTripSh + 2 \TimeHBA$.
   Consequently, by time  $\TimeTripSh + 2 \TimeHBA$, every honest party has a common $\CoreSet$ of size $n - t_s$.
   
   Consider an {\it arbitrary} party $P_j \in \CoreSet$. If $P_j$ is {\it honest}, then as shown above, 
   $L$ multiplication-triples will be shared on behalf of $P_j$ at time $\TimeTripSh$ during the instance $\TripSh^{(j)}$.
  Next consider a {\it corrupt} $P_j \in \CoreSet$. Since $P_j \in \CoreSet$, it follows that the honest parties computed the output $1$
  during the instance $\HBA^{(j)}$.
  This implies 
   that at least one 
  {\it honest} $P_i$ must have computed
   some output during
    the instance $\TripSh^{(j)}$, within time $\TimeTripSh + \TimeHBA$  (implying that $P_j \in \CSet_i$) and participated
  with input $1$ in the instance $\HBA^{(j)}$. This is because if $P_j$ {\it does not} belong to the $\CSet_i$ set of {\it any} honest $P_i$ 
  at time $\TimeTripSh + \TimeHBA$,
   then it implies that {\it all} honest parties participate with input
   $0$ in the instance $\HBA^{(j)}$ after time $\TimeTripSh + \TimeHBA$. 
   And then from the {\it $t_s$-validity} of $\HBA$ in the {\it synchronous} network, every honest party would compute
    the output $0$ in the instance
  $\HBA^{(j)}$ and hence $P_j$ will {\it not} be present in $\CoreSet$, which is a contradiction. 
  Now if $P_i$ has computed some output in $\TripSh^{(j)}$ by time $\TimeTripSh + \TimeHBA$, then from the
   {\it $t_s$-strong-commitment} of $\TripSh$ in the {\it synchronous} network, it follows that
  there exist $L$ multiplication-triples, say $\{(x^{(j, \ell)}, y^{(j, \ell)}, z^{(j, \ell)}) \}_{\ell = 1, \ldots, L}$, 
  which will be $t_s$-shared among the parties 
  on behalf of $P_j$
   by time  $(\TimeTripSh + \TimeHBA + 2\Delta) < (\TimeTripSh + 2\TimeHBA)$; the latter follows because $2\Delta < \TimeHBA$.
  
  From the above discussion, it follows that there will be $L$ multiplication-triples, which will be $t_s$-shared on behalf of {\it each} $P_j \in \CoreSet$ by time
  $\TimeTripSh + 2 \TimeHBA$. Hence each instance of $\TripExt$ will output $\frac{n - t_s - 1}{2} + 1 - t_s$ number of 
  $t_s$-shared multiplication-triples
  by time $\TimeOffline = \TimeTripSh + 2 \TimeHBA + \Delta$. This follows from the
  {\it $t_s$-correctness} property of $\TripExt$ in the {\it synchronous} network
   by substituting $|\CoreSet| = n - t_s$ and $d = \frac{n - t_s - 1}{2}$ in Lemma \ref{lemma:TripExt}.
  Since there are $L = \frac{c_M}{(\frac{n - t_s - 1}{2} + 1 - t_s)}$ instances of $\TripExt$, it follows that
  at time $\TimeOffline$, the parties have $L \cdot \Big ( \frac{n - t_s - 1}{2} + 1 - t_s  \Big )= c_M$ number of $t_s$-shared multiplication-triples.
  This completes the proof of the {\it $t_s$-correctness} property in the {\it synchronous} network.
 
 The proof of the {\it $t_a$-correctness} property in the {\it asynchronous} network is similar as above, except that we now use the 
   {\it $t_a$-correctness} and {\it $t_a$-strong-commitment} properties of $\TripSh$ in the {\it asynchronous} network, the 
  {\it $t_a$-correctness} property of $\TripExt$ in the {\it asynchronous} network and
  the properties of $\HBA$ in the {\it asynchronous} network.
  
  From the {\it $t_s$-privacy} property of $\TripSh$, it follows that the multiplication-triples which 
  are $t_s$-shared on behalf of the {\it honest} parties $P_j \in \CoreSet$
  will be random from the point of view of the adversary under the presence of up to $t_s$ corrupt parties,
  irrespective of the network type.
   It then follows from the {\it $t_s$-privacy} property of $\TripExt$ that the $t_s$-shared
  multiplication-triples generated from each instance of
  $\TripExt$ will be random from the point of view of the adversary. This proves the {\it $t_s$-privacy} property.
  
  The communication complexity follows from the communication complexity of $\TripSh$ and $\TripExt$, and
   from the fact that $\frac{n-t_s-1}{2} + 1 - t_s \geq \frac{t_a}{2} + 1$.
\end{proof}


\section{The best-of-both-worlds Circuit-Evaluation Protocol}
\label{sec:MPC}
  The best-of-both-worlds protocol $\PiMPC$ for evaluating $\ckt$
    has four phases. 
  In the first phase, the parties generate $t_s$-sharing of $c_M$ random multiplication-triples through $\Offline$. 
   The parties also invoke an instance of $\ACS$ to generate $t_s$-sharing of their respective inputs for $f$ and 
     agree on a {\it common} subset  $\CoreSet$ of {\it at least} $n - t_s$ parties, whose inputs for $f$ are $t_s$-shared,
      while the remaining inputs are set to
     $0$. In a {\it synchronous} network,
   all {\it honest} parties will be in $\CoreSet$, thus
   ensuring that the inputs of {\it all} honest parties are considered for the circuit-evaluation.
     In the second phase, each gate is evaluated in a $t_s$-shared fashion after which the parties {\it publicly} reconstruct 
  the secret-shared output in the third phase. The fourth phase is the  {\it termination phase}, where the parties check
    whether ``sufficiently many" parties have obtained the same output, in which case the parties
    ``safely" take that output
   and terminate the protocol (and all the underlying sub-protocols). 
    Protocol $\PiMPC$ is formally presented in Fig \ref{fig:MPCProtocol}. 
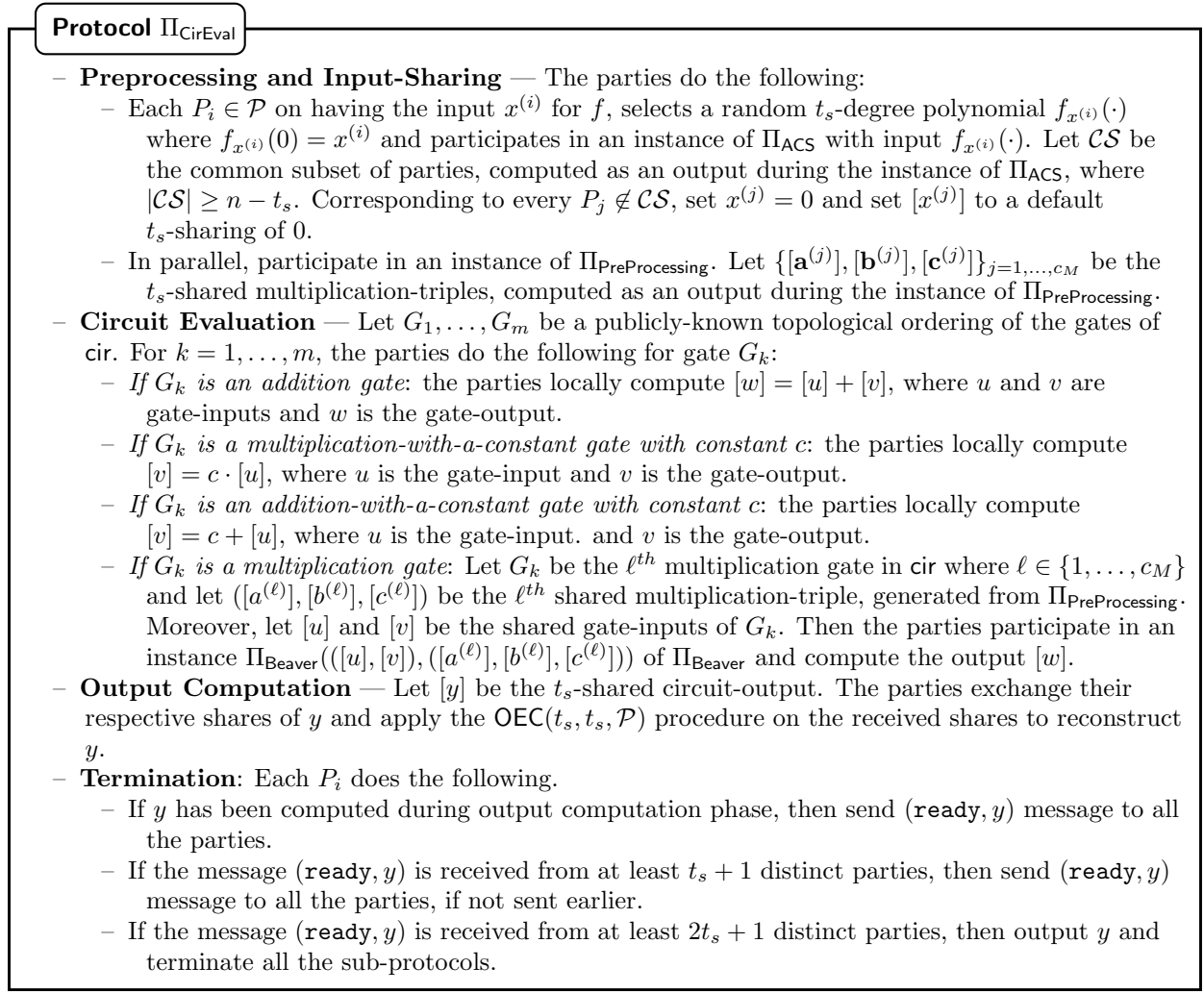
\begin{protocolsplitbox}{$\PiMPC$}{A best-of-both-worlds perfectly-secure protocol for securely evaluating the arithmetic circuit $\ckt$.}{fig:MPCProtocol}
  \begin{myitemize}
  \item[--] {\bf Preprocessing and Input-Sharing} --- The parties do the following:
    \begin{myitemize}
     \item[--] Each $P_i \in \Partyset$ on having the input $x^{(i)}$ for $f$, selects a random $t_s$-degree polynomial $f_{x^{(i)}}(\cdot)$ where $f_{x^{(i)}}(0) = x^{(i)}$ and participates in an instance
     of $\ACS$ with input  $f_{x^{(i)}}(\cdot)$. 
     Let $\CoreSet$ be the common subset of parties, computed as an output during the instance
     of $\ACS$, where $|\CoreSet| \geq n - t_s$. Corresponding to every
     $P_j \not \in \CoreSet$, set $x^{(j)} = 0$ and set $[x^{(j)}]$ to a default 
     $t_s$-sharing of $0$.
     \item[--] In parallel, participate in an instance of $\Offline$. 
     Let $\{[\mathbf{a}^{(j)}], [\mathbf{b}^{(j)}], [\mathbf{c}^{(j)}] \}_{j = 1, \ldots, c_M}$ be the 
     $t_s$-shared multiplication-triples, computed as an output during the instance of $\Offline$.
     \end{myitemize}  
   \item[--] {\bf Circuit Evaluation} --- Let $G_1, \ldots, G_{m}$ be a publicly-known topological ordering
       of the gates of $\ckt$. For $k = 1, \ldots, m$, the parties do the following for gate $G_k$:
         \begin{myitemize}
	   \item[--] {\it If $G_k$ is an addition gate}: the parties locally compute $[w] = [u] + [v]$, where $u$ and $v$ are
	    gate-inputs and $w$ is the gate-output.
	   \item[--] {\it If $G_k$ is a multiplication-with-a-constant gate with constant $c$}: the parties locally compute $[v] = c \cdot [u]$, 
	    where $u$ is the gate-input and $v$ is the gate-output.
	   \item[--] {\it If $G_k$ is an addition-with-a-constant gate with constant $c$}: the parties locally compute 
	   $[v] = c +  [u]$, where $u$ is the gate-input. and $v$ is the gate-output.
	   \item[--] {\it If $G_k$ is a multiplication gate}: Let $G_k$ be the $\ell^{th}$ multiplication gate in $\ckt$ where $\ell \in \{1, \ldots, c_M \}$
	    and let $([a^{(\ell)}], [b^{(\ell)}], [c^{(\ell)}])$ be the $\ell^{th}$ shared multiplication-triple, generated from $\Offline$.
	   Moreover,  let $[u]$ and $[v]$ be the shared gate-inputs of $G_k$. 
	   Then the parties participate in an instance $\BatchBeaver(([u], [v]), ([a^{(\ell)}], [b^{(\ell)}], [c^{(\ell)}]))$ of 
	   $\BatchBeaver$ and compute the output $[w]$.	   
     \end{myitemize}
    \item[--] {\bf Output Computation} --- Let $[y]$ be the $t_s$-shared circuit-output. 
    The parties exchange their respective shares of $y$ and apply the $\OEC(t_s, t_s, \Partyset)$ procedure on the received shares to 
    reconstruct $y$.    
      \item[--] {\bf Termination}: Each $P_i$ does the following.
          \begin{myitemize}
          \item[--] If $y$ has been computed during output computation phase, then send $(\ready, y)$ message to all the parties.
          \item[--] If the message $(\ready, y)$ is received from at least $t_s + 1$  distinct parties, then 
          send $(\ready, y)$ message to all the parties, if not sent earlier.
           \item[--] If the message $(\ready, y)$ is received from at least $2t_s + 1$  distinct parties, then
           output $y$ and  terminate all the sub-protocols.
          \end{myitemize}   
   \end{myitemize}
\end{protocolsplitbox}

We now prove the properties of the protocol $\PiMPC$.
\begin{theorem}
\label{thm:MPC}
Let $t_a < t_s$, such that $3t_s + t_a < n$. Moreover, let $f: \F^n \rightarrow \F$ be a function represented by an arithmetic circuit $\ckt$ over $\F$ consisting of
 $c_M$ number of multiplication gates, and whose multiplicative depth is $D_M$. Moreover, let party $P_i$ has input $x^{(i)}$ for $f$.
 Then, $\PiMPC$ achieves the following.
 \begin{myitemize}
 \item[--] In a synchronous network, all honest parties output $y = f(x^{(1)}, \ldots, x^{(n)})$ at time $(120n + D_M + 6k - 20) \cdot \Delta$,
       where $x^{(j)} = 0$ for every $P_j \not \in \CoreSet$, such that
       $|\CoreSet| \geq n - t_s$
       and every honest party $P_j \in \PartySet$ is present in $\CoreSet$.
       Here $k$ is the constant from Lemma \ref{lemma:ABAGuarantees}, as determined by the underlying
       (existing) perfectly-secure ABA protocol $\PiABA$.
\item[--] In an asynchronous network, almost-surely, the honest parties eventually output $y = f(x^{(1)}, \allowbreak \ldots, x^{(n)})$ 
       where $x^{(j)} = 0$ for every $P_j \not \in \CoreSet$
       and where $|\CoreSet| \geq n - t_s$.           
  \item[--] Irrespective of the network type, the view of the adversary will be independent of the inputs of the honest parties in $\CoreSet$.
     \item[--] The protocol incurs a communication of $\Order(\frac{n^5}{\frac{t_a}{2}+1 }c_M \log{|\F|} + n^7 \log{|\F|})$ bits from the 
     honest parties and
  invokes $\Order(n^3)$ instances of $\HBA$. 
 \end{myitemize}
\end{theorem}
\begin{proof}
Consider a {\it synchronous} network with up to $t_s$ corruptions. From the 
 properties of $\Offline$ in the {\it synchronous} network, 
 at time $\TimeOffline$, the (honest)
 parties output $c_M$ number of
  $t_s$-shared multiplication-triples, during the instance of $\Offline$. 
  From the {\it $t_s$-correctness} property of $\ACS$ in the {\it synchronous} network,
   at time $\TimeACS$, the (honest)
   parties output a common subset of parties
    $\CoreSet$ during the instance of $\ACS$, where {\it all} {\it honest} parties will be present in $\CoreSet$ and where $|\CoreSet| \geq n - t_s$.
  Moreover, corresponding to {\it every} $P_j \in \CoreSet$, there will be some $x^{(j)}$ available with
   $P_j$ (which will be the same as $P_j$'s input for $f$ for an {\it honest} $P_j$), 
  such that $x^{(j)}$ will be
  $t_s$-shared. As $\CoreSet$ will be known {\it publicly}, the parties take a default $t_s$-sharing of $0$ on the behalf of the parties $P_j$
  outside $\CoreSet$ by considering $x^{(j)} = 0$.
  Since $\TimeACS < \TimeOffline$, it follows that at time $\TimeOffline$, the parties will hold $t_s$-sharing of $c_M$ multiplication-triples
  and $t_s$-sharing of $x^{(1)}, \ldots, x^{(n)}$. 
  
  The circuit-evaluation will take $D_M \cdot \Delta$ time. This follows from the fact that linear gates are
   evaluated locally, while all the {\it independent} multiplication gates can be evaluated in parallel
  by running the corresponding instances of $\BatchBeaver$ in {\it parallel}, where each such instance requires $\Delta$ time.
   From the properties of $\BatchBeaver$
   in the {\it synchronous} network, the
  multiplication-gates will be evaluated correctly and hence, during the output-computation phase, the parties will hold a 
  $t_s$-sharing of $y$, where 
  $y = f(x^{(1)}, \ldots, x^{(n)})$. From the properties of $\OEC$, it will take $\Delta$ time for every party to reconstruct $y$. 
  Hence, during the termination phase,
  {\it all} honest parties will send a $\ready$ message for $y$. Since there are at least $2t_s + 1$ honest parties, every honest party will then terminate with output $y$ at time
  $\TimeOffline + (D_M + 2) \cdot \Delta$. By substituting the values of $\TimeOffline, \TimeTripSh, \TimeACS,
   \TimeHSh, \TimeWPS, \TimeBCAST, \TimeHBA, \TimeSBA$ and
  $\TimeABA$ and by noting that all instances of $\BCAST$ in $\PiMPC$ are invoked with $t = t_s$, we get that
  the parties terminate the protocol at time $(120n + D_M + 6k -20) \cdot \Delta$, where $k$ is the constant from Lemma \ref{lemma:ABAGuarantees}, as determined by the underlying
       (existing) perfectly-secure ABA protocol $\PiABA$.

      The proof of the properties
       in an {\it asynchronous} network is similar as above, except that we now use the security properties of the building blocks
  $\Offline, \ACS, \BatchBeaver$ and $\RecPriv$ in the {\it asynchronous} network. During the termination phase,
   at most $t_a$ {\it corrupt} parties can send $\ready$ messages for $y' \neq y$
  and there will be at least $2t_s + 1$ honest parties, who eventually send $\ready$ messages for $y$. Moreover, if some honest party $P_h$ terminates with output $y$, then every honest party eventually
  terminates the protocol with output $y$. 
  This is because $P_h$ must have received $\ready$ messages for $y$ from at least $ t_s + 1$ {\it honest} parties before termination, which are eventually delivered to {\it every} 
  honest party. Consequently, irrespective of which stage of the protocol an honest party is in, every honest party (including $P_h$) eventually
   sends a $\ready$ message for $y$ which are eventually delivered. As there
  are at least $2t_s + 1$ honest parties, this implies that every honest party eventually terminates with output $y$.
  
  From the {\it $t_s$-privacy} property of $\ACS$, 
   corresponding to every {\it honest} $P_j \in \CoreSet$, the input $x^{(j)}$ will be random from the point of view
  of the adversary. Moreover, from the properties
   of $\Offline$, the multiplication-triples generated through $\Offline$ will be random from the point of view of the
  adversary. During the evaluation of linear gates, no interaction happens among the parties and hence, no additional information about the inputs of the honest parties is revealed.
  The same is true during the evaluation of multiplication-gates as well, which follows from the properties of $\BatchBeaver$.

  The communication complexity of the protocol follows from the communication complexity of $\Offline, \ACS$ and $\BatchBeaver$.
 \end{proof}
\section{Conclusion and Open Problems}
In this work, we presented the {\it first} best-of-both-worlds 
 perfectly-secure MPC protocol,
  which remains secure both in a synchronous as well as an asynchronous network. 
  To design the protocol, we presented a best-of-both-worlds perfectly-secure VSS protocol
  and a best-of-both-worlds perfectly-secure BA protocol. 
   Our work leaves the following
  interesting open problems.
    \begin{myitemize}
     \item[--] We could not prove whether the condition $3t_s + t_a < n$ is also necessary for any best-of-both-worlds
     perfectly-secure MPC protocol and conjecture that it is indeed the case.
    \item[--] Our main focus in this work is on the {\it existence} of 
    best-of-both-worlds perfectly-secure MPC protocols. Improving the efficiency of the protocol is left open for future work.
    \end{myitemize}
  \paragraph{\bf Acknowledgements:} We would like to sincerely thank the anonymous reviewers of PODC 20222 for their excellent
  reviews on the preliminary version of this article, which got published as an extended abstract.


\bibliographystyle{plain}

\bibliography{main}

\appendix

\section{Properties of the Existing (Asynchronous) Primitives}
\label{app:ExistingPrimitives}
In this section we discuss the existing asynchronous primitives in detail.
\subsection{Online Error-Correction (OEC)}
 The OEC procedure uses a Reed-Solomon (RS) error-correcting procedure $\RSDec(d, r, {\cal W})$,
   that takes as input a set ${\cal W}$ of distinct points
   on a $d$-degree polynomial and tries to output a $d$-degree polynomial,
   by correcting at most $r$  {\it incorrect} points in ${\cal W}$. Coding theory \cite{MS81}
    says that RS-Dec can correct up to $r$ errors in ${\cal W}$ and
   correctly interpolate back the original polynomial if and only if $|{\cal W}| \geq d + 2r + 1$ holds. There are several efficient
   implementations of $\RSDec$ (for example, the algorithm of Berlekamp-Welch).
   
   Suppose $\Partyset' \subseteq \Partyset$ contains at most $t$ corrupt parties
    and let there exist some $d$-degree polynomial $q(\cdot)$, with every (honest)
   $P_i \in \Partyset'$ having a point $q(\alpha_i)$. 
       The goal is 
    to make some {\it designated} party $P_R$ reconstruct $q(\cdot)$. 
   For this, each $P_i \in \Partyset'$ sends $q(\alpha_i)$ to $P_R$, who then applies the OEC procedure $\OEC$ as described in Fig \ref{fig:OEC}.
   \begin{protocolsplitbox}{$\OEC(d, t, \Partyset')$}{The online error-correction procedure.}{fig:OEC}
 {\bf Setting}: There exists a subset of parties $\Partyset'$ containing at most $t$ corrupt parties,
 with each $P_i \in \Partyset'$ having a point $q(\alpha_i)$ on some $d$-degree polynomial $q(\cdot)$.
  Every (honest) party in $\Partyset'$ is supposed to send its respective point to $P_R$, who is designated to reconstruct $q(\cdot)$. 
 \begin{myitemize}
    \item[--] \textbf{Output Computation} --- For $r = 0, \ldots, t$, party $P_R$ does the following in iteration $r$:
       \begin{myitemize}
         \item[--] Let ${\cal W}$ denote the set of parties in $\Partyset'$ from whom $P_R$ has received the points and let $\mathcal{I}_r$ denote the points received
	from the parties in ${\cal W}$, when ${\cal W}$ contains exactly $d + t + 1 + r$ parties.
         \item[--] Wait until ${\cal W} \geq d + t + 1 + r$. Execute $\RSDec(d, r, {\cal I}_r)$ to get a $d$-degree polynomial,
         say $q_r(\cdot)$. If no
	polynomial is obtained, then skip the next step and proceed to the next iteration.
         \item[--] If for at least $d + t + 1$ values $v_i \in {\cal I}_r$ it holds that $q_r(\alpha_i) = v_i$, then output $q_r(\cdot)$. Otherwise,
	proceed to the next iteration.        
       \end{myitemize}       
      \end{myitemize}
\end{protocolsplitbox}
\begin{lemma}[\cite{CanettiThesis}]
\label{lemma:OECProperties}
Let $\Partyset' \subseteq \Partyset$ contain at most $t$ corrupt parties
    and let there exist some $d$-degree polynomial $q(\cdot)$, with every (honest)
   $P_i \in \Partyset'$ having a point $q(\alpha_i)$. 
   Then the OEC protocol prescribed in Fig \ref{fig:OEC} 
    achieves the following
  for an honest $P_R$ in the presence of up to $t$ corruptions.
       \begin{myitemize}
       \item[--] If $d < (|\Partyset'| - 2t)$, then in a  {\it synchronous} network, it takes at most $\Delta$ time
       for $P_R$ to output $q(\cdot)$. And 
      in an {\it asynchronous} network, $P_R$ eventually outputs $q(\cdot)$.
       \item[--] If $P_R$ obtains any output, then irrespective of the network type, the output polynomial is the same as
       $q(\cdot)$.
         \item[--] The protocol incurs a communication of $\Order(n \log{|\F|})$ bits from the honest parties.
   \end{myitemize}
\end{lemma}
\begin{proof}
The {\it communication complexity} follows from the fact that each party send its point to $P_R$. 
  We next show that if $P_R$ outputs a $d$-degree polynomial, say $q_r(\cdot)$, 
  during the iteration number $r$, then $q_r(\cdot)$ is the same as $q(\cdot)$, {\it irrespective} of the network type.
  However, this easily follows from the fact that  
   $q_r(\cdot)$ is consistent with $d + t + 1$ values from ${\cal I}_r$, out of which at least $d + 1$ values belong to the
  {\it honest} parties and thus, they lie on the polynomial $q(\cdot)$ as well. Furthermore, two {\it different} 
   $d$-degree polynomials can have at most $d$ distinct points in common.

We next prove the first property, assuming an {\it asynchronous} network. We first argue that an {\it honest}
 $P_R$ {\it eventually} obtains some output, provided $d < (|\Partyset'| - 2t)$. 
   Let adversary control $\hat{r}$ parties in $\Partyset'$, where $\hat{r} \leq t$. 
 Assume that $\hat{r}_1$ corrupt parties send incorrect points to $P_R$ and the remaining 
 $\hat{r}_2 = \hat{r} - \hat{r}_1$ corrupt parties do not send anything at all. Then, consider iteration number $t - \hat{r}_2$. 
 Since $\hat{r}_2$ parties never send any value, $P_R$ will receive at least $d + t + 1 + t - \hat{r}_2$ distinct points on $q(\cdot)$,
 of which $\hat{r}_1$ could be corrupted. Since $|{\cal I}_{d + t + 1 + t - \hat{r}_2}| \geq d + 2\hat{r}_1 + 1$ holds, 
  the algorithm $\RSDec$ will correct $\hat{r}_1$ errors and will return the polynomial $q(\cdot)$ during the iteration number $t - \hat{r}_2$.
   Therefore $P_R$ will obtain an output, latest after $(t - \hat{r}_2)$ iterations.
   
 The proof of the first property in the {\it synchronous} network is the same as above. In this case, it should be noted that the points of all honest parties reach $P_R$ within $\Delta$ time.
\end{proof}

\subsection{Bracha's Acast Protocol}
Bracha's Acast protocol \cite{Bra84} tolerating $t < n/3$ corruptions is presented in Fig \ref{fig:ACast}.
 \begin{protocolsplitbox}{$\PiACast$}{Bracha's Acast protocol. 
  The above code is executed by every $P_i \in \Partyset$ including the sender $\Sender$.}{fig:ACast}
\justify
   \begin{mydescription}
       \item[1.] If $P_i = \Sender$, then on input $m$,  send $(\init, \Sender, m)$ to all the parties.
        \item[2.] Upon receiving the message $(\init, \Sender, m)$ from $\Sender$, send 
       $(\echo, \Sender, m)$ to all
            the parties. Do not execute this step, more than once.
      \item[3.] Upon receiving $(\echo, \Sender, m^{\star})$ from $n - t$ parties,
            send $(\ready, \Sender, m^{\star})$ to all the parties.
      \item[4.] Upon receiving $(\ready, \Sender, m^{\star})$ from $t + 1$ parties, send
            $(\ready, \Sender, m^{\star})$ to all the parties.
       \item [5.] Upon receiving  $(\ready, \Sender, m^{\star})$ from $n - t$ parties, output $m^{\star}$.
   \end{mydescription}
\end{protocolsplitbox}

We now prove the properties of the protocol $\PiACast$. \\~\\
{\bf Lemma \ref{lemma:Acast}.}
{\it Bracha's Acast protocol $\PiACast$ achieves the following in the presence of up to $t < n/3$ corruptions, where $\Sender$
 has an input $m \in \{0, 1 \}^{\ell}$ for the protocol.
 \begin{myitemize}
\item[--] {\it Asynchronous Network}: 
     \begin{myitemize}
      \item[--] {\it (a) $t$-Liveness}: If $\Sender$ is {\it honest}, then all honest parties eventually obtain some output.
       \item[--] {\it (b) $t$-Validity}: If $\Sender$ is {\it honest}, then every honest party with an output, outputs $m$.
       \item[--] {\it (c) $t$-Consistency}: If $\Sender$ is {\it corrupt} and some honest party outputs $m^{\star}$,
        then every honest party {\it eventually} outputs $m^{\star}$.
    \end{myitemize}    
  \item[--] {\it Synchronous Network}:
         \begin{myitemize}
            \item[--] {\it (a) $t$-Liveness}: If $\Sender$ is {\it honest}, then all honest parties obtain an output within time $3\Delta$.
            \item[--] {\it (b) $t$-Validity}: If $\Sender$ is {\it honest}, then every honest party with an output, outputs $m$.
            \item[--] {\it (c) $t$-Consistency}: If $\Sender$ is {\it corrupt} and some honest party outputs $m^{\star}$ at time $T$, then every
             honest $P_i$ outputs $m^{\star}$ by the end of time $T + 2 \Delta$.
        \end{myitemize}      
  \item[--] Irrespective of the
  network type, $\Order(n^2 \ell)$ bits are communicated by the honest parties.
\end{myitemize}
}
\begin{proof}
We first prove the properties assuming an {\it  asynchronous} network with up to $t$ corruptions. 
 We start with the {\it validity} and {\it liveness} properties, for which we consider an {\it honest} $\Sender$. We show that 
  all honest parties eventually output $m$. 
  This is because all honest parties complete steps $2-5$ in the protocol,
   even if the corrupt parties do not send their messages. This is because there
  are at least $n - t$ honest parties, whose messages are eventually selected for delivery. Moreover, the adversary may send 
     at most $t$ $\echo$ messages for $m'$, where $m' \neq m$, on behalf of corrupt parties.
   Similarly, the adversary may send 
     at most $t$ $\ready$ messages for $m'$, where $m' \neq m$, on behalf of corrupt parties.
   Consequently, no honest party ever generates a $\ready$ message for $m'$,  neither in step $3$, nor in step $4$.  
   This is because $n - t > t$, as $t < n/3$.
   
   For {\it consistency}, we consider a {\it corrupt} $\Sender$ and let $P_h$ be an {\it honest} party, who outputs $m^{\star}$.
    We next show that all honest parties
   eventually outputs $m^{\star}$.
      Since $P_h$ outputs $m^{\star}$, it implies that it receives $n - t$ $\ready$ messages for $m^{\star}$ during step $5$ of the protocol.
   Let $\Honest$ be the set of {\it honest} parties whose $\ready$ messages are received by $P_h$ during step $5$. It is easy to see that
   $|\Honest| \geq t+1$. The $\ready$ messages of the parties in $\Honest$ are eventually delivered to every honest party
   and hence each honest party (including $P_h$) eventually executes step $4$ and sends a
   $\ready$ message for $m^{\star}$.
   As there are at least $n - t$ honest parties, it follows that eventually $n - t$ $\ready$ messages for $m^{\star}$ are delivered
   to every honest party (irrespective of whether adversary sends all the required
   messages). This guarantees that all honest parties eventually obtain some output. To complete the proof,
   we show that this output is $m^{\star}$.
   
    On contrary, let
   $P_{h'}$ be another honest party, different from $P_h$, who outputs
   $m^{\star \star} \neq m^{\star}$. This implies that
   $P_{h'}$ received $\ready$ messages for $m^{\star \star}$ from at least
   $t + 1$ {\it honest} parties during step 5 of the protocol.
   Now from the protocol steps, it follow that an honest party generates a $\ready$
   message for some potential $m$, only if it either receives $n - t$ $\echo$ messages for the $m$ during step 3
   or $t  + 1$ $\ready$ messages for $m$ (one of which has to come from an honest party) during step 4.
   So all in all, in order that $n - t$ $\ready$ messages are eventually generated for some potential $m$ during step 5, it must be
   the case that some honest party has to receive $n - t$ $\echo$ messages for $m$ during step 2 and generate a
   $\ready$ message for $m$. Now since $P_h$ receives $n - t$ $\ready$ messages for $m^{\star}$, 
   some honest party must have received $n - t$ $\echo$ messages for $m^{\star}$, at most $t$ of which could come from the corrupt parties.
   Similarly, since $P_{h'}$ receives $n - t$ $\ready$ messages for $m^{\star \star}$, 
   some honest party must have received $n - t$ $\echo$ messages for $m^{\star \star}$.
   However, since $n - t > 2t$, it follows that in order that $n - t$ $\echo$ messages are produced for both $m^{\star}$ as well as $m^{\star \star}$,
   it must be the case that some honest party must have generated an $\echo$ message, both for $m^{\star}$, as well
   as $m^{\star \star}$ during step 2, which is impossible. This is because an honest party executes step 2 at most once
   and hence generates an $\echo$ message at most once.
   
   The proofs of the properties in the {\it synchronous} network closely follow the proofs of the properties in the {\it asynchronous} network.
   If $\Sender$ is {\it honest}, then it will send the $\init$ message for $m$ to all the parties, which will be delivered within time
   $\Delta$. Consequently, every
   {\it honest} party will send an $\echo$ message for $m$ to all the parties, which will be delivered within time $2\Delta$. 
   Hence every
   {\it honest} party will send a $\ready$ message for $m$ to all the parties, which will be delivered within time
   $3 \Delta$. As there are at least $n - t$ honest parties, every honest
   party will receive $\ready$ messages for $m$ from at least $n - t$ parties within time $3 \Delta$ and output $m$.
   
   If $\Sender$ is {\it corrupt} and some honest party $P_h$ outputs $m^{\star}$ at time $T$, then it implies that
   $P_h$ has received $\ready$ messages for $m^{\star}$ during step $5$ of the protocol at time $T$ from a set $\Honest$ of
   at least $t + 1$ honest parties. 
   These ready messages are guaranteed to be received by every other honest party within time $T + \Delta$. Consequently, 
   every {\it honest} party who has not yet executed step $4$ will do so  and will send a $\ready$ message for $m^{\star}$ at time $T + \Delta$.
   Consequently, by the end of time $T + \Delta$, every honest party would have sent a $\ready$ message for $m^{\star}$ to every other honest party, which will be
   delivered within time $T + 2\Delta$.
   Hence, every honest party will output $m^{\star}$ latest at time $T + 2 \Delta$.
   
   The communication complexity (both in a synchronous as well as asynchronous network) simply follows from the fact that every party may need to send 
   an $\echo$ and  $\ready$ message for $m$ to every other party.
\end{proof}

\section{An Overview of the Existing ABA Protocols \cite{ADH08,BCP20}}
\label{app:ExistingABAAnalysis}
In this section, we give a very high level overview of the existing 
 {\it $t$-perfectly-secure} ABA protocols of \cite{ADH08,BCP20}. Both these protocols 
 are perfectly-secure and can tolerate up to $t < n/3$ corruptions. The protocols
  follow the standard framework of Rabin and Ben-Or \cite{Rab83,Ben83}, which uses two building-blocks to get a BA protocol. 
   The first building-block is a {\it voting} protocol (often called {\it gradecast} or {\it graded consensus} in the literature) 
   and which is a 
  deterministic protocol. 
  The second building-block is a {\it coin-flipping} protocol which is a randomized protocol.
  In the sequel, we review these building blocks and discuss how they are ``combined" to get an ABA protocol. While presenting these building-blocks, unless it is explicitly stated,
   we assume an {\it asynchronous} network. Also, for simplicity, we present these building-blocks {\it without} specifying any {\it termination} criteria and hence, the parties may keep on 
  running these building-blocks (as well as the ABA protocol)
   even after obtaining an output.\footnote{Recall that we do not put any termination criteria for any of our sub-protocols, as the termination of the MPC
   protocol will automatically ensure that all the underlying sub-protocols also get terminated.}
\subsection{The Voting Protocol}
Informally, the voting protocol does ``whatever can be done deterministically'' to reach agreement.
In a voting protocol, every party has
a single bit as input.  The protocol tries to find out whether there is a detectable
majority for some value among the inputs of the parties. In the protocol, each
party's output can have {\it five}
different forms:
\begin{myitemize}
	\item[--] For $\sigma \in \{0,1\}$, the output $(\sigma, 2)$ stands for ``overwhelming majority
	for $\sigma$'';
	\item[--] For $\sigma \in \{0, 1\}$, the output $(\sigma, 1)$ stands for ``distinct majority for $\sigma$'';
	\item[--] The output $(\Lambda, 0)$ stands for ``non-distinct majority''.
\end{myitemize}
The protocol code of the voting protocol taken from \cite{CanettiThesis} is presented in Fig \ref{fig:Vote}. 
\begin{protocolsplitbox}{$\Vote$}{The vote protocol. The above code is executed by every $P_i \in \Partyset$.}{fig:Vote}
\justify
\begin{myitemize}
\item[--] On having the input $x_i$, Acast $(\texttt{input}, P_i, x_i)$.
\item[--] Create a dynamic set ${{\cal X}}_i$ which is initialized to $\emptyset$.
  Add $(P_j, x_j)$ to ${\cal X}_i$ if $(\texttt{input},P_j, x_j)$ is received from the Acast of $P_j$.
\item[--] Wait until $|{\cal X}_i| =n-t$. Then assign $X_i = {\cal X}_i$, set $a_i$ to the majority bit among $\{x_j \; | \; (P_j,x_j) \in X_i\}$.
 Acast $(\texttt{vote}, P_i, X_i, a_i)$.
					\item[--] Create a dynamic set ${\cal Y}_i$, which is initialized to $\emptyset$.
					Add $(P_j, X_j, a_j)$ to ${\cal Y}_i$
					if $(\texttt{vote},P_j, X_j, a_j)$ is received from the Acast  of $P_j$,
					$X_j \subseteq {\cal X}_i$, and $a_j$ is the majority bit of $X_j$.
					\item[--] Wait until $|{\cal Y}_i| =n-t$. Then assign $Y_i = {\cal Y}_i$, set $b_i$
					to the majority bit among $\{a_j \; | \; (P_j, X_j, a_j) \in Y_i\}$ and
					Acast $(\texttt{re-vote}, P_i, Y_i, b_i)$.
					\item[--] Create a  set $Z_i$, which is initialized to $\emptyset$.
					Add $(P_j,Y_j,b_j)$ to $Z_i$
					if $(\texttt{re-vote},P_j,Y_j,b_j)$ is received from the Acast of $P_j$,
					$Y_j \subseteq {\cal Y}_i$, and $b_j$ is the majority bit of $Y_j$.
					\item[--] Wait until $|Z_i| =  n-t$. Then compute the output as follows.
					\begin{myitemize}
					 \item[--] If all the parties $P_j \in Y_i$ have the same
					{\texttt{vote}} $a_j = \sigma$, then output $(\sigma,2)$.
					 \item[--] Else if all the parties $P_j \in Z_i$ have the same {\texttt{re-vote}}
					$b_j = \sigma$, then output $(\sigma,1)$.
					\item[--] Else output $(\Lambda,0)$.
					\end{myitemize}

\end{myitemize}
\end{protocolsplitbox}
The properties of the voting protocol are stated in Lemma \ref{lemma:Vote}. While these properties hold in an {\it asynchronous} network, it automatically implies that they hold even for 
 a {\it synchronous} network. We refer the readers to \cite{CanettiThesis,BCP20} for the proof of these properties.
\begin{lemma}[\cite{CanettiThesis,BCP20}]
\label{lemma:Vote}
Protocol $\Vote$ achieves the following properties, both in the synchronous as well as asynchronous network, if the adversary
 corrupts up to $t < n/3$ parties,
 where all the parties participate with an input bit.
\begin{myitemize}
\item[--] If each honest party has the same input $\sigma$,
	then each honest party outputs $(\sigma, 2)$;
\item[--] If  some honest party outputs $(\sigma, 2)$, then every other honest
	party outputs either $(\sigma,2)$ or $(\sigma,1)$;
\item[--] If some honest party outputs $(\sigma,1)$ and no honest party outputs
	$(\sigma, 2)$ then each honest party outputs either
	$(\sigma, 1)$ or $(\Lambda,0)$.
\item[--] The protocol incurs a communication of $\Order(n^3)$ bits from the honest parties.
\end{myitemize}
\end{lemma}
 An additional property which protocol $\Vote$ achieves in a {\it synchronous} network is that all {\it honest} parties will have
  their output by the end of time $9\Delta$.
 Intuitively, this is because the protocol involves three different ``phases" of Acast, each of which will produce an output
  within time $3 \Delta$ for {\it honest} sender parties in a {\it synchronous}
 network. Moreover, from Lemma \ref{lemma:Vote}, this output will be $(\sigma, 2)$, if all the {\it honest} parties have the same input $\sigma$.
 We will require this property later while claiming the properties of the resultant ABA protocol in a {\it synchronous} network. Hence, we prove this property.
\begin{lemma}
\label{lemma:VoteSynProperties}
 If the network is synchronous and if the adversary corrupts up to $t < n/3$ parties, 
  then in protocol $\Vote$, all honest parties obtain their output within time $9 \Delta$.
 Moreover, the output will be $(\sigma, 2)$, if all the {\it honest} parties have the same input $\sigma$.
 \end{lemma}
 \begin{proof}
 Consider an arbitrary {\it honest} $P_j$. Party $P_j$ will Acast its input $x _j$ and from the {\it $t$-liveness} and {\it $t$-validity} properties of 
  Acast
  in the {\it synchronous} network, every honest party will receive the output $x_j$, from the corresponding
  Acast instance within time $3\Delta$. As there are at least $n - t$ {\it honest} parties, it implies that every {\it honest} $P_i$ will obtain a set $X_i$ of size $n - t$ within time $3 \Delta$.
  Hence each {\it honest} $P_i$ will Acast a $(\texttt{vote}, P_i, X_i, a_i)$ message latest at time $3 \Delta$ and every honest party receives 
  this message from the corresponding Acast instance within time $6 \Delta$. We also note that if there is a {\it corrupt}
  $P_j$ such that $(P_j, x_j)$ is included by an {\it honest} $P_i$ in its set $X_i$ when $P_i$ Acasts $X_i$, 
   then from the {\it $t$-consistency} property of Acast
   in the {\it synchronous} network, every {\it honest} party $P_k$
  will include $(P_j, x_j)$ in its set ${\cal X}_k$, latest by time $5 \Delta$. This further implies that upon receiving the message $(\texttt{vote}, P_i, X_i, a_i)$
  from the Acast of any {\it honest} $P_i$, all {\it honest} parties $P_k$ will be able to verify this message and include 
  $(P_i, X_i, a_i)$ in their respective ${\cal Y}_k$ sets within time $6 \Delta$.
  
  As there are at least $n - t$ honest parties $P_i$ whose $\texttt{vote}$ messages are received and verified by all honest parties $P_k$ within time $6 \Delta$, 
  it follows that every honest party Acasts a $\texttt{re-vote}$ message, latest at time $6 \Delta$, 
  which is received by every honest party within time $9 \Delta$.
  Moreover, as argued for the case of $\texttt{vote}$ messages, every {\it honest} party will be able to verify these $\texttt{re-vote}$ messages and include in their respective
  $Z_i$ within time $9 \Delta$. Since there are at least $n - t$ honest parties, it follows that the $Z_i$ sets of every honest party will attain the size of
  $n - t$ within time $9 \Delta$ and hence every honest party will obtain an output, latest at time $9 \Delta$.
  
  If all the honest parties have the same input $\sigma$, then there will be at most $t$ {\it corrupt} parties who may Acast $1 - \sigma$. Hence {\it every} party (both honest as well as corrupt)
  will send a $\texttt{vote}$ message only for
   $\sigma$.\footnote{If a {\it corrupt} party $P_j$ sends a $\texttt{vote}$ message for $1 - \sigma$, then it will never be accepted and no honest
  party $P_i$ will ever include $(P_j, X_j, 1 - \sigma)$ in its  ${\cal Y}_i$ set. This is because $1 - \sigma$ will not be the majority among the inputs
  of the honest parties in $X_j$.}  Consequently, every honest party will output $(\sigma, 2)$.  
 \end{proof}
\subsection{Coin-Flipping Protocol}
The {\it coin-flipping} protocol denoted by $\CoinFlip$ (also called as the {\it common-coin} protocol) is an $n$-party (asynchronous) protocol, where the parties have local random inputs
 and the protocol outputs a bit for all the parties. The protocol achieves the following properties in an {\it asynchronous} (and hence {\it synchronous}) network in the
  presence of any $t < n/3$ corruptions.
  \begin{myitemize}
  \item[--] In an {\it asynchronous} network, all honest parties eventually obtain an output, while in a {\it synchronous} network, the honest parties
   obtain an output within some fixed time $c \cdot \Delta$, where $c$ is a {\it publicly-known} constant. 
  \item[--] One of the following holds:
     \begin{myitemize}
      \item[--]  If no party deviates from the protocol, then 
      with probability {\it at least} $p$, the output bits of all the honest parties are same. The probability $p$ where $p< 1$ is often called as the {\it success-probability}
       of the protocol and is a parameter of the protocol.
      \item[--] Else, all honest parties will have the same output bit with probability {\it less than} $p$. But in this case, 
      the protocol allows some {\it honest} party(ies) to {\it locally} identify and shun a (subset) of corrupt
      party(ies) from any future communication. Namely, the protocol {\it locally} outputs ordered pairs of the form $(P_i, P_j)$, 
      where $P_i$ is some {\it honest} party and $P_j$ is some {\it corrupt} party, such that $P_i$ identifies $P_j$ as a corrupt party
      and does not consider any communication from $P_j$ for the rest of the protocol execution. Such pairs are called as {\it local-conflicts}.
      We stress that the local-conflicts are identified only {\it locally}. For instance, if an {\it honest} $P_i$ has shunned a {\it corrupt} $P_j$ during an instance of the
      coin-flipping protocol, then it is {\it not} necessary
      that every other {\it honest} party $P_k$ also shuns $P_j$ during the same instance, as $P_j$ may decide to behave ``honestly" towards $P_k$.

      The coin-flipping protocol of \cite{ADH08} guarantees that at least one {\it new} local-conflict is identified if, during an instance of the coin-flipping protocol,
      the parties obtain the same output bit with probability less than $p$. On the other hand, the coin-flipping protocol of \cite{BCP20} guarantees that 
      $\Theta(n)$ number of {\it new} local-conflicts are identified, if the parties obtain the same output bit with probability less than $p$.
    \end{myitemize}  
  \end{myitemize}
   Protocol $\CoinFlip$ is designed using a {\it weaker} variant of perfectly-secure AVSS called {\it shunning} AVSS (SAVSS), introduced in \cite{ADH08}.
   The SAVSS primitive is weaker than AVSS in the following aspects:
   \begin{myitemize}
   \item[--] It is {\it not} guaranteed that {\it every} honest party obtains a point on $\D$'s sharing-polynomial
   (and hence 
    a share of $\D$'s secret), even if $\D$ is {\it honest};
   \item[--] If $\D$ is {\it corrupt}, then it may not participate with a $t$-degree polynomial and hence, 
   the underlying shared value could be $\bot$, which is different from every element of $\F$;  
   \item[--] Irrespective of $\D$, depending upon the behaviour of the corrupt parties,
    the honest parties later may either reconstruct the same secret as shared by $\D$ or an all-together  different value. However, in the latter case, 
    the protocol ensures that at least one {\it new} local-conflict is identified.
   \end{myitemize}
    In \cite{ADH08}, a perfectly-secure SAVSS protocol is designed with $t < n/3$.
     By executing $n^2$ instances of this protocol in {\it parallel} using the framework of \cite{FM97,CanettiThesis}, a
    coin-flipping protocol is presented in \cite{ADH08}, where the success-probability $p$ is $\frac{1}{4}$.
    The protocol incurs a communication of $\Order(\mbox{poly}(n) \log{|\F|})$ bits from the honest parties.
    
    The coin-flipping protocol of \cite{BCP20} also uses the same framework of  \cite{FM97,CanettiThesis}, but substitutes the SAVSS of \cite{ADH08}
   with a  ``better" and more efficient SAVSS with $t < n/3$. Their SAVSS ensures that 
    $\Theta(n)$ number of {\it new} local-conflicts are identified, if the value reconstructed by the parties is {\it different}
    from the one shared by $\D$. The success-probability $p$ remains $\frac{1}{4}$ and the communication complexity of the protocol is $\Order(\mbox{poly}(n) \log{|\F|})$ bits.
    \subsection{Vote $+$ Coin-Flipping $\Rightarrow$ ABA}
    We now show how to ``combine" protocols $\Vote$ and $\CoinFlip$ to get the protocol $\PiABA$ (see Fig \ref{fig:ABA}). The current description of $\PiABA$ is taken
     from \cite{BKL19}. The protocol consists of several iterations, where each iteration consists of two instances of $\Vote$ protocol and one instance of $\CoinFlip$, which are carefully
     ``stitched" together. 
     
     In the first instance of $\Vote$, the parties participate with their ``current input", which is initialized to their respective bits for ABA in the first iteration.
     Then, independent of the output received from the instance of $\Vote$, the parties participate in an instance of $\CoinFlip$.
     Next, the parties decide their respective inputs for the second instance of $\Vote$ protocol, based on the output they received from the first instance.
     If a party has received the {\it highest} grade (namely $2$) during the first instance of $\Vote$, then the party {\it continues} with 
     the bit received from that $\Vote$ instance for the second $\Vote$ instance. Otherwise, the party {\it switches} to the
     output received from $\CoinFlip$. 
     The output from the second instance of $\Vote$ is
     then set as the modified input for the next iteration, if it is obtained
      with a grade higher than $0$. Otherwise, the output of $\CoinFlip$ is taken as the
     modified input for the next iteration.
     
     If during any iteration a party obtains the highest grade from the second instance of $\Vote$, then it indicates this publicly by
      sending a $\ready$ message to every party, along with the bit received. The $\ready$ message is an indication for the others about the ``readiness" of the
      sender party to consider the corresponding bit as the output.
     Finally, once a party receives this readiness indication for a common bit $b$ from at least $2t + 1$
      parties, then that bit is taken as the output. To ensure that every other party also outputs the same bit, a party upon receiving the $\ready$ message for a common bit from at least
      $t + 1$ honest parties, itself sends a $\ready$ message for the same bit (if it has not done so earlier).
      
      The idea behind the protocol is the following. In the protocol there can be two cases. 
      The {\it first} case is when all the honest parties start with the {\it same} input bit, say $b$. Then, they will obtain the output $b$ from all the instances of $\Vote$ protocol in all
      the iterations and the outputs from $\CoinFlip$ will be never considered. Consequently, each honest party will eventually send a $\ready$ message for $b$. Moreover, 
      there can be at most $t$ corrupt parties who may send a $\ready$ message for $1 - b$ and hence no honest party ever sends a $\ready$ message for $1 - b$.
      Hence, each honest party eventually
      outputs $b$. 
      
      The {\it second} case is when
      the honest parties start the protocol with {\it different} input bits. In this case, the protocol tries to take the help of $\CoinFlip$ to ensure that all honest parties reach an iteration with a common input bit for that iteration. Once such an iteration is reached, this {\it second} case gets ``transformed" to the {\it first} case
       and hence all honest parties will eventually output that common bit. 
       In more detail, in each iteration $k$, it will be ensured that either every honest party have the same input bit for the second instance of
      $\Vote$ with probability at least $p \cdot \frac{1}{2}$ or else certain number of new 
      local-conflicts are identified.\footnote{The number of local-conflicts identified will depend upon the $\CoinFlip$ protocol: while the $\CoinFlip$
      protocol of \cite{ADH08} will ensure that at least $1$ new local-conflict is identified, the $\CoinFlip$ protocol of \cite{BCP20} will ensure that $\Theta(n)$ number of new local-conflicts are
      identified.} This is because the input for second instance of $\Vote$ is either the output bit of the first instance of $\Vote$ or the output of $\CoinFlip$, both of which are {\it independent}
      of each other. Hence if the output of $\CoinFlip$ is same for all the parties with probability $p$, then with probability
       $p \cdot \frac{1}{2}$, this bit will be the same as output bit from the 
      first instance of $\Vote$. 
       If in any iteration $k$, it is {\it guaranteed} that all honest parties have the same inputs for the second instance of $\Vote$, then
       the parties will obtain a common output and with highest grade from the second instance of $\Vote$.
       And then from the next iteration onward, all parties will stick to that common bit and eventually output that common bit.

       One can show that it requires $\Order(\mbox{poly}(n))$ number of iterations in {\it expectation} before a ``good" iteration is reached,
       where an iteration is considered good, if 
       it is
        {\it guaranteed} that all honest parties have the same input for the second instance of $\Vote$. Intuitively, this is because
         there can be $\Order(\mbox{poly}(n))$ number of ``bad" iterations in which the honest parties may have different outputs from the corresponding instances
         of $\CoinFlip$. This follows from the fact that the corrupt parties may deviate from the protocol instructions during the instances of
         $\CoinFlip$. There can be at most $t(n - t)$ local-conflicts which may occur ($t$ potentially corrupt parties getting in conflict with $n - t$ honest parties) {\it overall}
         during various ``failed" instances of $\CoinFlip$ (where a failed instance means that different honest parties obtain different outputs)
         and only after all these local-conflicts are identified, the parties may start witnessing ``clean" instances of $\CoinFlip$ where all honest parties shun communication from
         all corrupt parties and where
         it is ensured that all honest parties obtain the same output bit with probability $p$. Now depending upon the number of new local-conflicts which are revealed
         from a single failed instance of $\CoinFlip$, the parties may witness $\Order(\mbox{poly}(n))$ number of bad 
         iterations.\footnote{Since each failed 
         instance of the $\CoinFlip$ protocol of \cite{ADH08} may reveal only $1$ new local-conflict, the number of bad iterations could be $\Order(n^2)$. On the other hand,
         each failed instance of the $\CoinFlip$ protocol of \cite{BCP20} reveals $\Theta(n)$ new local-conflicts and hence there can be
         $\Order(n)$ number of bad iterations.} 
         Now, once all the bad iterations are over and all potential local-conflicts are identified, in each subsequent iteration, all honest parties will then have the same output from
         $\CoinFlip$ (and hence, same input for the second instance of $\Vote$) with probability at least $\frac{p}{2}$. Consequently, 
         if $p$ is a {\it constant}, then 
          it will take
         $\Theta(1)$ expected number of such iterations before the parties reach a good iteration where it is guaranteed that all honest parties have the same inputs for the second 
         instance of $\Vote$.\footnote{One can show that if one sets $p = \frac{1}{4}$ as done in \cite{CR93,ADH08,BCP20}, then it takes expected $16$ iterations {\it after} all the local-conflicts are identified
         to reach a good iteration.}

    \begin{protocolsplitbox}{$\PiABA$}{The ABA protocol from $\Vote$ and $\CoinFlip$. The above code is executed by every $P_i \in \Partyset$.}{fig:ABA}
	\justify
    {\bf Input}: Party $P_i$ has the bit $b_i$ as input for the ABA protocol.
	\begin{myitemize}
         \item[--] {\bf Initialization}: Set $b = b_i$, $\committed = \false$ and $k = 1$. Then do the following.
            \begin{mydescription}
            \item Participate in an instance of $\Vote$ protocol with input $b$.            
            \item Once an output $(b, g)$ is received from the instance of $\Vote$, participate in an instance of $\CoinFlip$. Let
            $\Coin_k$ denote the output received from $\CoinFlip$.
            \item  If $g < 2$, then set $b = \Coin_k$.
            \item  Participate in an instance of $\Vote$ protocol with input $b$ and let $(b', g')$ be the output received.
                  If $g' > 0$, then set $b = b'$.
            \item  If $g' = 2$ and $\committed = \false$, then set $\committed = \true$ and send 
            $(\ready, b)$ to all the parties.
            \item  Set $k = k + 1$ and repeat from $1$.
            \end{mydescription}
         \item[--] {\bf Output Computation}: 
            \begin{myitemize}
              \item[--] If $(\ready, b)$ is received from at least $t + 1$ parties, then send $(\ready, b)$ to all the parties.
               \item[--] If $(\ready, b)$ is received from at least $2t + 1$ parties, then output $b$.
            \end{myitemize}
          \end{myitemize}
    \end{protocolsplitbox}
    
Lemma \ref{lemma:ABAGuarantees} now follows easily from the above discussion. Let $c \cdot \Delta$ be the time within which 
 the protocol $\CoinFlip$ generates output for the
 honest parties in a {\it synchronous} network, where $c$ is some publicly-known constant. Note  that $c$ is determined by
 the underlying SAVSS protocol and is different for the SAVSS protocols of \cite{ADH08} and \cite{BCP20}.
 If all honest parties have the same input $b$ in a {\it synchronous} network, then at the end of the first iteration itself, every party will send a $\ready$ message
 for $b$ to every other party. Consequently, in this case, all honest parties will obtain their output within time $(c + 18 + 1) \cdot \Delta$.
  This is because each instance of $\Vote$ during the first iteration will take at most $9 \Delta$ time to produce output,
   while the instance of $\CoinFlip$ will take at most  $c \cdot \Delta$ time.
  Additionally, $\Delta$ time will be taken by each party to send a $\ready$ message for $b$ to every other party. 
    Consequently, $\TimeABA$ will be $(c + 19) \cdot \Delta$.

\end{document}